\documentclass[aps,twocolumn,showpacs,reprint,superscriptaddress,floatfix,longbibliography]{revtex4-2}

\usepackage{mathrsfs}
\usepackage{xcolor}
\usepackage[dvipsnames]{xcolor}
\usepackage{multirow}
\usepackage{amsthm}
\usepackage{amsmath}
\usepackage{mdframed}
\usepackage{amsfonts}
\usepackage{bbm}
\usepackage{dsfont}
\usepackage[ruled,vlined]{algorithm2e}
\usepackage{comment}
\usepackage{tabularx}
\usepackage{subfigure}
\usepackage{amssymb}
\usepackage[utf8]{inputenc}
\usepackage{graphicx}
\usepackage[english]{babel}
\usepackage{wasysym}
\usepackage{enumerate}
\usepackage{float}
\usepackage{afterpage}
\usepackage{textgreek}
\usepackage{tikz}
\usepackage[most]{tcolorbox}
\usetikzlibrary{decorations.pathreplacing}
\usepackage{cancel}
\usepackage{genyoungtabtikz}
\usepackage{dcolumn}% Align table columns on decimal point
\usepackage{bm}% bold math
\usepackage{bbm}
\usepackage[colorlinks=true, linkcolor=blue, citecolor=blue, urlcolor=blue]{hyperref}
\newcommand{\e}{\text{e}}
\renewcommand{\i}{\text{i}}
\newcommand{\C}{\mathbb{C}}
\newcommand{\R}{\mathbb{R}}
\newcommand{\F}{\mathcal{F}}

\newcommand{\B}{\mathcal{B}}

\renewcommand{\L}{\mathcal{L}}
\newcommand{\I}{\mathbb{I}}
\renewcommand{\H}{\mathcal{H}}
\renewcommand{\O}{\mathcal{O}}
\newcommand{\tr}{\mathrm{Tr}} 
\newcommand{\Tr}{\mbox{Tr}}
\renewcommand{\dag}{\dagger}

\newcommand{\til}{~}
\newcommand{\bt}{`}

\newcommand{\bluecheck}{\textcolor{blue}{$\checkmark ~ $}}
\renewcommand{\bluecheck}{}
\newcommand{\nit}{\noindent}
\newcommand{\spec}{\textup{spec}}

\newcommand{\specdown}{\textup{spec}^\downarrow}
\newcommand{\bn}{\boldsymbol{n}}
\newcommand{\bv}{\boldsymbol{v}}
\newcommand{\bs}{\boldsymbol{s}}
\newcommand{\btt}{\boldsymbol{t}}
\newcommand{\bb}{\boldsymbol{b}}
\newcommand{\ba}{\boldsymbol{a}}
\newcommand{\bx}{\boldsymbol{x}}
\newcommand{\by}{\boldsymbol{y}}
\newcommand{\bz}{\boldsymbol{z}}
\newcommand{\be}{\boldsymbol{e}}
\newcommand{\bbeta}{\boldsymbol{\beta}}
\newcommand{\btheta}{\boldsymbol{\theta}}
\newcommand{\balpha}{\boldsymbol{\alpha}}
\newcommand{\bP}{\boldsymbol{P}}
\newcommand{\bp}{\boldsymbol{p}}
\newcommand{\blam}{\boldsymbol{\lambda}}
\newcommand{\sbs}{\; | \;} %space - bar - space
\newcommand{\bset}{\mathfrak{B}}

\newcommand*{\exclaeq}{\stackrel{\text{!}}{=}}

\newtheorem{defn}{Definition}

\newtheorem{lemm}{Lemma}
\newtheorem{coro}{Corollary}
\newtheorem{thrm}{Theorem}

\newtheorem{conj}{Conjecture}
\newtheorem{prop}{Proposition}

\newcommand{\hideorblue}[1]{}

\definecolor{mygray}{gray}{0.65}

\newcommand{\bra}[1]{\mbox{$\langle #1 |$}}
\newcommand{\ket}[1]{\mbox{$| #1 \rangle$}}
\newcommand{\braket}[1]{\mbox{$\langle #1 \rangle$}}

\newcommand{\evalat}[2]{%
	#1\,
	{%
		\sbox0{$#1\left|#1\right|$}%
		\vrule height \dimexpr\dp0+0.6ex\relax depth \dimexpr\dp0+0.4ex\relax
	}_{\,#2}
}

\newtheorem*{prop*}{Proposition}
\newtheorem{result}{Result}

\begin{document}

\title{Maximal Orbital Entanglement from Non-Interacting Fermions}

\author{Damiano Aliverti-Piuri}
\affiliation{Department of Physics, Arnold Sommerfeld Center for Theoretical Physics,
	Ludwig-Maximilians-Universit\"at M\"unchen, Theresienstrasse 37, 80333 M\" unchen, Germany}
\affiliation{Munich Center for Quantum Science and Technology (MCQST), Schellingstrasse 4, 80799 M\"unchen, Germany}

\author{Julia Liebert}
\affiliation{Department of Chemistry, Princeton University, Princeton, NJ 08540, USA}
\affiliation{Department of Physics, Arnold Sommerfeld Center for Theoretical Physics,
	Ludwig-Maximilians-Universit\"at M\"unchen, Theresienstrasse 37, 80333 M\" unchen, Germany}
\affiliation{Munich Center for Quantum Science and Technology (MCQST), Schellingstrasse 4, 80799 M\"unchen, Germany}

\author{Christian Schilling}
\email{christian.schilling@unige.ch}
\affiliation{Department of Physical Chemistry, University of Geneva, 30 Quai Ernest Ansermet, 1211 Geneva, Switzerland}
\affiliation{Department of Physics, Arnold Sommerfeld Center for Theoretical Physics,
	Ludwig-Maximilians-Universit\"at M\"unchen, Theresienstrasse 37, 80333 M\" unchen, Germany}
\affiliation{Munich Center for Quantum Science and Technology (MCQST), Schellingstrasse 4, 80799 M\"unchen, Germany}

\date{\today}

\begin{abstract}
Orbital rotations are used in quantum many-body physics and chemistry to reduce the complexity of quantum states by compressing their correlation into a small set of relevant orbitals. In striking contrast, quantum information processing strives for the opposite: orbitals that maximize entanglement as a resource. Motivated by this, we prove that Slater determinants of $N$ fermions on a lattice of $d_A + d_B = d$ sites, although free of particle correlation, can carry up to $\min\{d_A, d_B, N, d-N\}$ ebits of orbital entanglement. Given enough particles and holes, this equals the maximum over all fermionic states. For arbitrary free states, we derive the maximal orbital total correlation in closed form for all bipartitions, the finest partition, and all partitions of up to five sites, using new majorization relations for Hermitian matrices. Taken together, these results turn mean-field states, the starting point of most many-body methods in physics and chemistry, into a viable resource for quantum information processing.
\end{abstract}

\maketitle

%\tableofcontents

\section{Introduction} \label{sec:intro}

Entanglement, a central feature of quantum mechanics~\cite{EPR1935, schrodinger1935, bell1964}, is a key resource for quantum information processing~\cite{horodecki_entanglement, resource_theory_chitambar, acin2018roadmap}. Its quantification is based on the paradigm of local operations and classical communication (LOCC), which requires a division of the system into subsystems~\cite{horodecki_entanglement}. This division is not unique, and the entanglement of a given state depends on the chosen subsystems~\cite{zanardi, 2004zanardi_ent_relativity_observables}. For identical particles, this \emph{entanglement relativity} is fundamental~\cite{DingSchilling2020paradox, AlivertiPiuri2024QIT}. In second quantization, groups of orthonormal orbitals form proper subsystems, and the natural notion is \emph{orbital} entanglement~\cite{ciancio_zanardi,zanardi_mode_pic, shi_mode_entanglement, friis_mode_ent, Friis_2016,ent_ferm_sys_banuls, wolf_2009_mode_ent_fermions, lexin_orb_ent, kraus_mode_ent, wisemanvaccaro, 2018_experiment_ent_between_modes, Szalay2021fermionicQI}. It naturally appears in tensor-network methods and in quantum chemistry~\cite{white_orb_corr, Szalay2015tensor, orb_ent_quantum_chem}. In first quantization, in contrast, single particles are not subsystems in the LOCC sense~\cite{AlivertiPiuri2024QIT}. \emph{Particle} entanglement~\cite{schliemanncirac,ECKERT200288, 2bosonschmidtdeco,  gigena2015,Iemini_quantumness_of_corr,ieminiquantumness, Li_ent_two_id_parts, vanhala_complexity_2024, yaroslav_classifyingfermionicstatesmanybody, LoFranco_id_particles, Balachandran_GNS_PRL, adesso_yadin_resource,plenio_extracting_ent, Giorgadze2025characterizing} is instead understood as the deviation of a state from the Slater determinants (SDs) or, more generally, free states (number-conserving Gaussian states), and quantifies the intrinsic complexity of a many-body state~\cite{AlivertiPiuri2024QIT, Liao2024complexity}.

\begin{figure}[t]
	\centering
	\includegraphics[width=0.47\textwidth]{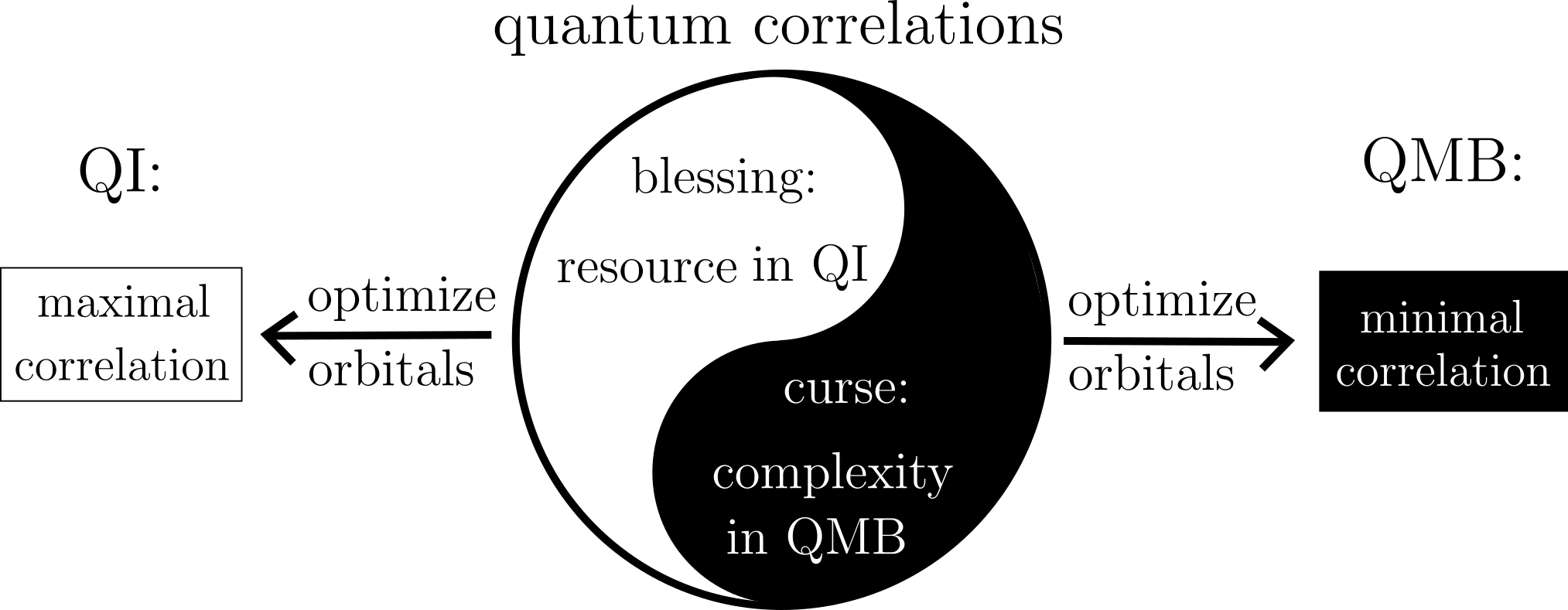}
	\caption{Yin and yang of orbital rotations: maximizing orbital correlation as a resource in quantum information (QI) versus minimizing it to reduce representational complexity in quantum many-body physics (QMB); see text for details.}
	\label{fig:yinyang}
\end{figure}

The two pictures are linked by orbital rotations, i.e., one-particle (1P) unitaries, which leave the particle correlation invariant but change the orbital correlation. In computational physics and chemistry, orbitals are rotated to reduce the representational complexity of a state, i.e., to compress its correlation into a small set of relevant orbitals, e.g., in tensor-network methods~\cite{Krumnow2016orbitalopt, Mate2023compressing, Friesecke2024globalmode}, quantum chemistry~\cite{DingKnechtSchilling2023QICAS, Liao2024complexity, Ding2022faithful}, and quantum algorithms~\cite{Mizukami2020OOUCC, Li2025EMO}. For free states, which are uncorrelated in their natural orbitals, the minimal orbital correlation vanishes~\cite{AlivertiPiuri2024QIT}. Quantum information processing strives for the opposite (Fig.~\ref{fig:yinyang}), namely orbitals that \emph{maximize} the entanglement, rendering a state maximally resourceful from an operational point of view~\cite{DingKnecht2023bonding, ding2026bonding, gesa_lexin_2024, 2025troiani_GME_from_SDs}. Free states are ideal candidates: they are ubiquitous in physics and chemistry~\cite{fetter_walecka, Szab1982ModernQC}, can be prepared efficiently on quantum hardware~\cite{Arute2020HF}, and the orbital rotations connecting them decompose into beam-splitter-type rotations of orbital pairs~\cite{zeilinger_unitary_breakdown, 2013_bocquillon_fermionicbeamsplitters, bloch_review, kim_double_qdot}. However, unlike the minimization, the maximization has been explored only for specific states~\cite{DingKnecht2023bonding, 2025troiani_GME_from_SDs}, on average~\cite{hackl2017ent_eigstates, 2021bianchi_average_ent_over_gauss, 2021Hackl_Gaussian_states, 2022_bianchi_tutorial, yale2024_average_ent_fixed_N, 2025hackl_average_mutual_info}, or for bosonic and distinguishable-particle systems~\cite{ciancio_zanardi, wolf_ent_power, Jevtic2012_unitary_orbits_PRL}. This leaves open a fundamental question that bridges the particle and orbital pictures and addresses the maximization side of Fig.~\ref{fig:yinyang}: \emph{How much orbital entanglement and total correlation can non-interacting fermions, although uncorrelated in the particle picture, carry under optimal orbital rotations?}

\begin{table*}
\begin{center}
\renewcommand{\arraystretch}{1.4}

	\centering
	\begin{tabular}{|llll|}
		\hline
		\multicolumn{4}{|l|}{\textbf{Roadmap of results}}
		\\ \hline
		\multicolumn{1}{|l|}{\textbf{State class}} & \multicolumn{1}{l|}{\textbf{Partition} $\blam$} & \multicolumn{1}{l|}{\textbf{Status}} & \textbf{Result} \\ \hline
		\multicolumn{1}{|l|}{Slater determinants}         & \multicolumn{1}{l|}{$(\lambda_1, \lambda_2)$}       & \multicolumn{1}{l|}{\textcolor{Green}{\textbullet} Solved}       & $E_{\max} = \min \{ \lambda_1, \lambda_2, N, d-N \}$ ebits \hfill [Result~\ref{prop:Emax_SD_bipartition}]         \\ \hline
		\multicolumn{1}{|l|}{Slater determinants}         & \multicolumn{1}{l|}{Any}       & \multicolumn{1}{l|}{\textcolor{YellowOrange}{\textbullet} Partial}       & Trivial-eigenvalue pattern; small-blocks reduction \hfill [App.~\ref{app:SDs}]         \\ \hline
		\multicolumn{1}{|l|}{Free states}         & \multicolumn{1}{l|}{$(\lambda_1, \lambda_2)$}       & \multicolumn{1}{l|}{\textcolor{Green}{\textbullet} Solved}       & $\bb_{\max}$ contains averages $(n_j + n_{d+1-j})/2$ \hfill [Result~\ref{res:bipartitions}]         \\ \hline
		\multicolumn{1}{|l|}{Free states}         & \multicolumn{1}{l|}{$(1, \ldots ,1)$}       & \multicolumn{1}{l|}{\textcolor{Green}{\textbullet} Solved}       & $\bb_{\max} = (\eta \sbs \ldots \sbs \eta)$, with $\eta = \sum_j n_j / d$ \hfill [Result~\ref{res:finest}]         \\ \hline
		\multicolumn{1}{|l|}{Free states}         & \multicolumn{1}{l|}{Any, $d \leq 5$}       & \multicolumn{1}{l|}{\textcolor{Green}{\textbullet} Solved}       & $\bb_{\max}$ explicitly, regime by regime \hfill [Result~\ref{res:small_d}]         \\ \hline
		\multicolumn{1}{|l|}{Free states}         & \multicolumn{1}{l|}{Any, $d \geq 6$}       & \multicolumn{1}{l|}{\textcolor{red}{\textbullet} Open}       & $\bb_{\max}$ conjectured $\prec$-least and with averages of $\bn$ entries \hfill [Conjs.~\ref{conj:1},~\ref{conj:2}]         \\ \hline
		\multicolumn{1}{|l|}{Local particle-number SSR}         & \multicolumn{1}{l|}{$d \leq 3$; $(2,2)$}       & \multicolumn{1}{l|}{\textcolor{YellowOrange}{\textbullet} Partial}       & Same optimal unitary; loss vanishes as $N \to \infty$ \hfill [Sec.~\ref{subsec:extensions_SSR}, App.~\ref{app:SSR}]         \\ \hline
		\multicolumn{1}{|l|}{Pure Gaussian states}         & \multicolumn{1}{l|}{$(k,k)$}       & \multicolumn{1}{l|}{\textcolor{Green}{\textbullet} Solved}       & $E^G_{\max} = k$ ebits \hfill [Prop.~\ref{prop:gaus}]         \\ \hline
		\multicolumn{1}{|l|}{Pure/mixed Gaussian states}         & \multicolumn{1}{l|}{Any}       & \multicolumn{1}{l|}{\textcolor{red}{\textbullet} Open}       & Needs a covariance-matrix-based generalization of our method         \\ \hline
	\end{tabular}
	\caption{Roadmap of the fermionic orbital entanglement and correlation maximization problems addressed in this work. The parameter vector $\bn$ contains the natural occupation numbers, and $\bb_{\max}$ denotes the maximizer of our optimization problem (as recast in Prop.~\ref{prop:Imax_max_over_bset}). The free-state maximizers carry over to bosonic free states, which, however, carry no entanglement (see App.~\ref{app:bosons}).}
	\label{tab:roadmap}
\end{center}
\end{table*}

Answering this question is the central goal and main achievement of this work. To this end, we use the fact that free states are fully determined by their one-particle reduced density matrices, which turns the problem into a spectral optimization for $d \times d$ matrices governed by majorization (Fig.~\ref{fig:strategy}). We first prove that the maximal orbital entanglement of an SD of $N$ fermions on a lattice of $d_A + d_B = d$ sites is $\min\{d_A, d_B, N, d-N\}$ ebits, reached by beam-splitter-type rotations of orbital pairs (Result~\ref{prop:Emax_SD_bipartition}). For enough particles and holes, this equals the maximum over all fermionic states. We then maximize the orbital total correlation of arbitrary, possibly mixed, free states under orbital rotations. We solve this problem for all bipartitions (Result~\ref{res:bipartitions}) based on a new majorization relation for blocks of Hermitian matrices (Thm.~\ref{thrm:mikes_maj_rel}), for the finest partition (Result~\ref{res:finest}), and for all partitions of up to five sites (Result~\ref{res:small_d}), and we conjecture the form of its general solution (Conjs.~\ref{conj:1} and~\ref{conj:2}). We also show that natural occupation numbers bound the orbital bipartite entanglement of any fermionic pure state (Result~\ref{res:entanglement_bound}) and analyze the effect of superselection rules (SSRs)~\cite{wisemanvaccaro, ent_ferm_sys_banuls, Friis_2016}. Table~\ref{tab:roadmap} summarizes all results.

These results identify mean-field states as entanglement reservoirs and complement orbital optimization for minimal complexity with its maximization counterpart. Since bosonic free states are separable (Appendix~\ref{app:bosons}), orbital entanglement is, among free states, a genuinely fermionic resource.

The paper is organized as follows. Section~\ref{sec:main_problem} introduces the basic problem for SDs and states its solution, which is proven in Sec.~\ref{sec:math_structure}; there, we also prepare the mathematical framework for the general problem. Section~\ref{sec:generalizing} formulates the general problem for free states, Sec.~\ref{sec:math_struc_gen_prob} studies its structure, and Sec.~\ref{sec:solutions_general} derives its solutions. Section~\ref{sec:extensions} discusses SSRs and pure fermionic Gaussian states, and Sec.~\ref{sec:conclusions} concludes. Technical details are given in the appendices.

\section{Basic problem: Maximal bipartite orbital entanglement} \label{sec:main_problem}

In this section, we introduce, motivate, and solve the basic problem of this work: the maximal bipartite orbital entanglement attainable by $N$-fermion Slater determinants. After a single-particle example, we define the setting and formulate the problem; its solution is Result~\ref{prop:Emax_SD_bipartition}, proven in Sec.~\ref{sec:math_structure}.

As a first example, consider a single particle delocalized across two spatial regions. Let its orbital state be
\begin{equation} \label{eq:left+right_vector}
	\ket{\psi} = \frac{1}{\sqrt{2}} ( \ket{\phi_1} + \ket{\phi_2} ) \; ,
\end{equation}
where $\ket{\phi_j}$ is supported in region $j$. This state has no particle entanglement, since there is only one particle. Yet the two orthonormal orbitals $\ket{\phi_1}$ and $\ket{\phi_2}$ are entangled: if a particle detector placed in the first region detects a particle, the one in the second region will not. This point of view is captured by the occupation-number representation
\begin{equation}\label{eq:left+right_vector_occupations}
\ket{\psi} =\frac{1}{\sqrt{2}}(\ket{10} + \ket{01})
\end{equation}
of the same state. If the SSRs~\cite{wisemanvaccaro} are ignored, the occupation-number representation~\eqref{eq:left+right_vector_occupations} shows that there is one ebit of entanglement between the two spatial regions. This entanglement value is \emph{maximal}. Indeed, any given 1P pure state can be decomposed as a linear combination of two normalized orbitals, each of which is supported in one of the two regions. This enables any state to be represented as a two-qubit state, by using the occupation-number representation relative to these two orbitals:
\begin{equation}\label{eq:left+right_vector_occupations_alphabeta}
\ket{\psi} =\alpha \ket{10} + \beta \ket{01} \; ,
\end{equation}
where $\alpha, \beta \in \C$ satisfy $|\alpha|^2 + | \beta |^2 = 1$. Most measures (e.g., entanglement entropy) of the orbital entanglement of the above state are maximized if and only if $|\alpha| = |\beta| = 1/\sqrt{2}$, as in Eq.~\eqref{eq:left+right_vector_occupations}.

We now generalize this observation to many-body states. We consider a system of $N$ identical fermions positioned on a bipartite lattice of $d$ sites of arbitrary geometry and dimension, which typically represent nuclear positions. We focus on spinless fermions (or, equivalently, describe spin-orbitals as lattice sites). The lattice is partitioned into a first region, consisting of the sites $j = 1, \dots , \lambda_1$, and a second region, consisting of the sites $j = \lambda_1 + 1, \dots , \lambda_1 + \lambda_2 = d$, where $\blam = (\lambda_1, \lambda_2)$ is a bipartition of the integer $d$. We write $\lambda_{1/2}$ instead of $d_{A/B}$, since this notation extends directly to multipartitions $\blam = (\lambda_1, \ldots, \lambda_k)$ of $d$. Our aim is to study the orbital entanglement between the regions. As in Hubbard-like models, we assume that only one relevant fermionic orbital $\ket{j}$ is associated with any given site $j$: each site can be empty or singly occupied; hence $d \geq N$. The orbitals $\ket{1}, \ldots, \ket{d}$ are assumed orthonormal, and form the reference orthonormal basis (ONB) of the finite-dimensional 1P space $\mathcal{H}_1 \cong \C^d$. The corresponding fermionic Fock space of many-body state vectors will be denoted by $\F_{\H_1}$.

Suppose that the $N$ fermions are in a pure, non-interacting state, i.e., they are described by an SD state vector
\begin{equation}
	\ket{SD_0} = f^{\dag}_{\phi_1} \dots f^{\dag}_{\phi_N} \ket{ 0 } \; ,
\end{equation}
where $\ket{\phi_1}, \dots, \ket{\phi_N} \in \H_1 $ are arbitrary orthonormal orbitals and $f_{\phi_j}^\dagger$ the corresponding creation operators. For example, the fermions may be localized on the lattice sites (e.g., $\ket{\phi_j} = \ket{j}$ for $j = 1, \dots, N$, as in the upper panel of Fig.~\ref{fig:experiment}) or in momentum space (i.e., the vectors $\ket{\phi_k} \propto \sum_{j} \e^{\i 2 \pi j \cdot k/d} \ket{j}$, $k \in \{ 0, \ldots, d-1 \}$, are lattice plane waves). While the particles are unentangled, in general there is \emph{orbital} entanglement between the two lattice regions. This is a consequence of the fact that the occupied orbitals $\phi_j$ are generally delocalized across both regions of the lattice, as in the situation described by Eq.~\eqref{eq:left+right_vector}. We regard this orbital entanglement as a \emph{resource}~\cite{resource_theory_chitambar} to be harvested from the system (e.g., by transferring the system state to a quantum register of qubits~\cite{wiseman_vaccaro_ferreting, gesa_lexin_2024}). To \emph{maximize} this resource, we assume the ability to apply external fields or potentials, use beam splitters, or, more generally, interact with the system in a way that allows the $N$ fermions to undergo an arbitrary 1P (i.e., non-interacting) unitary evolution---see Fig.~\ref{fig:experiment}. In short, for any unitary $u$ on $\H_1$, we assume it is experimentally possible to change the initial SD to
\begin{equation}
	\ket{SD} = u^{\otimes N} \ket{SD_0} = f^{\dag}_{\phi'_1} \dots f^{\dag}_{\phi'_N} \ket{ 0 } \; .
\end{equation}
This is a new, generic SD built from the transformed orbitals $\ket{\phi'_j} = u \ket{\phi_j}$\footnote{Note the common, slight abuse of notation: the operator $u^{\otimes N}$, defined on $\H_1^{\otimes N}$, here acts on $N$-fermion vectors in $\F_{\H_1}$.}.

We now analytically compute the \emph{maximal orbital entanglement entropy} $E_{\max}$ between the lattice regions that can be achieved through 1P unitary evolution, as illustrated in Fig.~\ref{fig:experiment}.

	\begin{figure}[h!]
	\centering
	\includegraphics[width=0.45\textwidth]{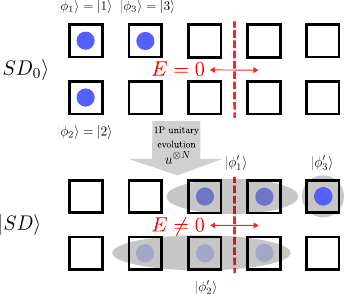}
	\caption{1P unitary evolution of an SD of three fermions (circles) in a lattice of $6 + 4$ sites (squares), numbered from top to bottom and from left to right. The initial SD $\ket{SD_0} = f_{\phi_1}^\dagger f_{\phi_2}^\dagger f_{\phi_3}^\dagger \ket{0}$ (above) exhibits no orbital entanglement. The evolution is described by a unitary on $\H_1$ mapping $\ket{\phi_j} \mapsto u \ket{\phi_j} = \ket{\phi'_j}$. $\ket{SD_0}$ is transformed into $\ket{SD} = f_{\phi'_1}^\dagger f_{\phi'_2}^\dagger f_{\phi'_3}^\dagger \ket{0}$ (below). Due to the delocalization of the new orbitals $\ket{\phi'_1}$ and $\ket{\phi'_2}$, $\ket{SD}$ exhibits exploitable orbital entanglement $E$, which we intend to maximize.}
	\label{fig:experiment}
\end{figure}

We quantify the orbital entanglement of a pure $N$-fermion state vector $\ket{\Psi}$ using the entanglement entropy
\begin{equation} \label{eq:def_entanglement}
	E(\ket{\Psi}) = S(\rho_{1/2}) = - \tr [ \rho_{1/2} \log \rho_{1/2} ] \; ,
\end{equation}
i.e., the von Neumann entropy of the orbital-reduced state $\rho_{1/2}$ of either region of the lattice. These concepts are reviewed in Ref.~\cite{Friis_2016}; see also Appendix~\ref{app:mutual_info_fermions}. Since we are considering a global pure state with a definite particle number (and thus a definite number parity), $\rho_1$ and $\rho_2$ have the same nonzero eigenvalues~\cite{Friis_2016}, and hence the same entropy---just as in non-fermionic bipartite systems. We emphasize that this is generally false for state vectors with both even- and odd-parity components~\cite{Friis_2016}.

For such pure bipartite states, $E$ is directly related to the orbital \emph{total correlation}, i.e., the quantum mutual information $C = S(\rho_1) + S(\rho_2) - S(\rho)$ (see Appendix~\ref{app:mutual_info_fermions}): since $S(\rho) = 0$ and $\rho_1$, $\rho_2$ are isospectral, $C = 2E$. %Geometrically, $C$ and $E$ are the minimal relative entropies of $\rho$ to the uncorrelated states~\cite{ModiVedralrecapgeom} and to the separable states~\cite{VedralPlenio1998} (for fermionic separability, see Ref.~\cite{ent_ferm_sys_banuls}), attained at $\rho_1 \otimes \rho_2$ and at the Schmidt-dephased state, respectively; for pure bipartite states, both are fixed by the spectrum of $\rho_1$.
Maximizing $E$ or $C$ over SDs is thus the same problem.

With this framework established, we turn to the central question of this work.

\begin{tcolorbox}[
colframe=blue!70!black, % Line color
colback=blue!17,         % Background color
boxrule=1.2pt,          % Line width
arc=6pt,                % Corner radius (increase for more "roundness")
boxsep=0pt,
left=10pt,              % Inner margins
right=10pt,
top=7pt,
bottom=7pt,
enhanced            
]
\textbf{Basic Scientific Problem.} For a lattice divided into two subsystems of $\lambda_1$ and $\lambda_2$ sites, determine the maximal orbital entanglement entropy between them that a Slater determinant of $N$ fermions can attain,
\begin{equation} \label{eq:E_max}
	E_{\max} := \max_{u \in \mathcal{U}(\H_1)} E\!\left(u^{\otimes N} \ket{SD_0}\right) = \max_{|SD\rangle} \! E(\ket{SD}) \; .
\end{equation}
\end{tcolorbox}
Here, the first maximization is over all unitaries $u$ on $\H_1$, while the second, equivalent one is over all $N$-fermion Slater determinants $\ket{SD}$. In either formulation, notice that we are maximizing a continuous function over a compact set, so the existence of maximizers and maxima is guaranteed. Note that if we were interested in \emph{minimizing} $E$ over Slater determinants, the result would be trivially zero: for example, Slater determinants obtained by filling some of the lattice site orbitals $\ket{j}$ have zero orbital entanglement.

$E_{\max}$ implicitly depends on the numbers $N$ of fermions, $d$ of lattice sites, and $\lambda_1$ of sites in the first region (or, equivalently, $\lambda_2 = d - \lambda_1$), but not on the initial SD, $\ket{SD_0}$. Additionally, $E_{\max}$ remains unchanged if the number of fermions, $N$, and the number of \emph{holes}, $d - N$, are swapped. Indeed, the orbital entanglement $E$ of any given $N$-particle SD $f^{\dag}_{\phi_1} \dots f^{\dag}_{\phi_N} \ket{ 0 }$ equals the orbital entanglement of the SD $f_{\phi_1} \dots f_{\phi_N} \ket{ F }$ of $N$ holes (i.e., $d-N$ particles), where $\ket{F}$ is the \bt full' state vector of $d$ fermions. This particle-hole symmetry property of our problem will simplify the discussion of Result~\ref{prop:Emax_SD_bipartition}, our first main result.

Following Refs.~\cite{zanardi, 2004zanardi_ent_relativity_observables, ciancio_zanardi} on entanglement relativity, $E_{\max}$ can be reinterpreted from a `passive' rather than `active' point of view, by considering a \emph{fixed} SD and transformations of the \emph{orbitals} that define the two complementary subsystems. To that end, consider a fixed $N$-fermion SD, built from orbitals in $\H_1$. The corresponding orbital entanglement between the first and second regions of the lattice can most conveniently be associated with the \bt splitting' of $\H_1$ as the direct sum of the orthogonal subspaces $V_1 = \text{span}(\ket{1}, \dots , \ket{\lambda_1})$ and $V_2 = \text{span}(\ket{\lambda_1 + 1}, \dots , \ket{d})$. In fact, this entanglement is often viewed as entanglement \emph{between} these two subspaces. Analogously, one can define the orbital entanglement between two arbitrary orthogonal subspaces $W_1$ and $W_2$ of $\H_1$. For example, $W_{1}$ ($W_2$) could be the span of lattice plane waves whose energy lies below (above) a certain threshold. From this point of view, $E_{\max}$ can be interpreted as the bipartite orbital entanglement of the fixed SD between $W_1$ and $W_2$ maximized over \emph{any} pair of orthogonal subspaces $W_1$ and $W_2$ of $\H_1$ such that $W_1 \oplus W_2 = \H_1$, $\text{dim}(W_1) = \lambda_1$, and $\text{dim}(W_2) = \lambda_2$. We will return to this alternative formulation at several points below. This freedom of choosing the orbitals has been exploited in two opposite directions. Minimizing the orbital correlation over orbital rotations simplifies many-body wave functions and their numerical treatment~\cite{Krumnow2016orbitalopt, DingKnechtSchilling2023QICAS, Liao2024complexity}, and its minimum defines the particle correlation~\cite{AlivertiPiuri2024QIT}. Maximizing the orbital entanglement, in turn, reveals chemical bonds as entanglement between atomic-like orbitals~\cite{DingKnecht2023bonding, ding2026bonding}.

The value of $E_{\max}$ for arbitrary $N, d, \lambda_{1}$, and $\lambda_2 = d - \lambda_1$ is presented in the following result.
\begin{result}[Slater determinants] \label{prop:Emax_SD_bipartition}
	For $N$ fermions on a lattice of $d$ sites, the maximal orbital entanglement entropy $E_{\max}$ [Eq.~\eqref{eq:E_max}] that an SD can attain between a region of $\lambda_1$ sites and the complementary region of $\lambda_2 = d - \lambda_1$ sites evaluates to
	\begin{equation} \label{eq:Emax_result}
		E_{\max} = \min\{ \lambda_1, \lambda_2, N, d - N \} \times \log 2 \; ;
	\end{equation}
    %in particular, if $N, d- N \geq \min \{ \lambda_1, \lambda_2 \}$  $E_{\max}$ equals the maximum over arbitrary pure states in $\F_{\H_1}$, i.e., $\min \{ \lambda_1, \lambda_2 \} \times \log 2$.
    in particular, $E_{\max}$ saturates the maximal orbital entanglement entropy of \emph{any} parity-superselected pure state in $\F_{\H_1}$, namely $\min \{ \lambda_1, \lambda_2 \} \times \log 2$, if and only if $\min \{ N, d- N \} \geq \min \{ \lambda_1, \lambda_2 \}$.
    
	Under the non-restrictive assumptions $\lambda_1 \geq \lambda_2$ and $N \leq d - N$, we can construct an SD with orbital entanglement $E_{\max}$ by placing as many fermions as possible in the orbitals 
	\begin{equation}
		\frac{\ket{1} + \ket{\lambda_1 + 1}}{\sqrt{2}}, \frac{\ket{2} + \ket{\lambda_1 + 2}}{\sqrt{2}}, \ldots \; ;
	\end{equation}
	if $N > \lambda_2$, the remaining $N - \lambda_2$ fermions are placed in the orbitals $\ket{\lambda_2 + 1}, \ket{\lambda_2 + 2}, \dots, \ket{N}$ of the first region.
	
\end{result}
The assumptions $\lambda_1 \geq \lambda_2$ and $N \leq d - N$ in the second part of Result~\ref{prop:Emax_SD_bipartition} are non-restrictive in the following sense: if they do not hold, one can simply swap the roles of the first and second regions of the lattice, or of fermions and holes, respectively. In Sec.~\ref{sec:math_structure}, we will prove the first part of Result~\ref{prop:Emax_SD_bipartition}, i.e., Eq.~\eqref{eq:Emax_result}, by reformulating the maximization over $N$-fermion SDs as one over their idempotent rank-$N$ one-particle reduced density matrices (1RDMs). Here, we explain the physical insights behind the construction of the optimal SD described in the result and illustrated in Fig.~\ref{fig:max_SD_construction}. The construction will show that the optimal SD indeed has entanglement $E_{\max}$, as given by Eq.~\eqref{eq:Emax_result}, thus proving the second part of Result~\ref{prop:Emax_SD_bipartition}.
	
\begin{figure}[h!]
	\centering
	\includegraphics[width=0.45\textwidth]{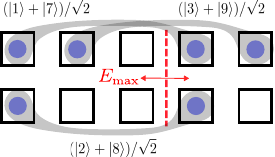}
	\caption{Representation of the orbitals involved in the construction of the optimal SD described in Result~\ref{prop:Emax_SD_bipartition} for three fermions in a lattice of $6 + 4$ sites.}
	\label{fig:max_SD_construction}
\end{figure}

The first key observation is as follows: for $\mu \in \{1, 2\}$, suppose that the number $N$ of fermions is strictly less than the number $\lambda_\mu$ of lattice sites in region $\mu$. Then this region must contain at least $\lambda_\mu - N$ holes, that is, there must be $\lambda_\mu - N$ \emph{empty orbitals} in region $\mu$. Similarly, a number of holes $d - N$ strictly less than $\lambda_\mu$ implies the existence of fully occupied orbitals. This reasoning is clarified by applying Lemma~\ref{lem:evalsubmatrix} to the diagonal blocks of the 1RDM of any SD, whose eigenvalues $0$ and $1$ are highly degenerate (see Appendix~\ref{app:normal_form}). These constraints on the distribution of particles and holes limit the achievable orbital entanglement, because completely empty or completely occupied orbitals are not entangled with the remaining ones. Taking these constraints into account, one aims to arrange the remaining unconstrained particles and holes so as to maximize orbital entanglement between the lattice sites of the first and second regions. One can try to do so by placing a particle in an orbital $\cos \theta_1 \ket{1} +  \sin \theta_1 \ket{\lambda_1 + 1}$ that is a superposition of a basis orbital from the first region and one from the second region; then another in $\cos \theta_2 \ket{2} + \sin \theta_2 \ket{\lambda_1 + 2}$; and so on (see Fig.~\ref{fig:max_SD_construction}). Under the assumptions $\lambda_1 \geq \lambda_2$ and $N \leq d - N$, this procedure continues until one runs out of either particles, if $N \leq  \lambda_2$, or second-region basis orbitals, if $N \geq  \lambda_2$. In the latter case, one simply places the remaining particles in empty orbitals in the first region. To maximize entanglement, it is convenient to choose balanced, or unbiased, superpositions of orbitals, i.e., $\cos \theta_j = \sin \theta_j = 1/\sqrt{2}$, up to irrelevant phases. The entanglement of this SD can be computed by treating each pair of orbitals $\ket{j}$ and $\ket{\lambda_1 + j}$ that share a fermion as a pair of maximally entangled qubits. Each such pair contributes one ebit of entanglement. Straightforward bookkeeping shows that the entanglement is given by the right-hand side of Eq.~\eqref{eq:Emax_result}. For chemical bond states, i.e., SDs of doubly occupied bonding orbitals, the analogous rotations of bonding and antibonding orbital pairs by $\pi/4$ were found to yield maximal orbital entanglement~\cite{DingKnecht2023bonding}.

Starting from an initial SD $\ket{SD_0}$ where the $N$ fermions are positioned on lattice sites, $\ket{SD_0}$ can be transformed into the optimal SD using a sequence of $\min \{ \lambda_1, \lambda_2, N, d-N \}$ elementary two-orbital unitaries. As discussed in greater detail in Sec.~\ref{subsec:2D_space}, these unitaries are straightforward to implement in experiments, which paves the way for entanglement-extraction protocols.

Finally, we remark that the above construction of SDs (with unconstrained angles $\theta_j$) actually recovers \emph{all} $N$-fermion SDs up to local (and thus entanglement-preserving) 1P unitaries, i.e., 1P unitaries that involve orbitals from one region only. This is the subject of Appendix~\ref{app:normal_form}, where a normal form of SDs relative to a bipartite 1P space is presented. As a consequence, the maximally entangled SD constructed above is the unique entanglement maximizer up to local 1P unitaries.	

\section{Mathematical structure of the basic problem} \label{sec:math_structure}
In this section, we first prove Result~\ref{prop:Emax_SD_bipartition} (Sec.~\ref{subsec:proof_of_main_theorem}), based on the description of SDs in terms of their 1RDMs. We then study (Sec.~\ref{subsec:towards_general}) the mathematical structure of our maximization problem in more detail. This prepares the ground for the generalized problem of maximizing the multipartite correlation of a free state over 1P unitaries (see Sec.~\ref{sec:generalizing}).

\subsection{Computation of $E_{\max}$} \label{subsec:proof_of_main_theorem}
Our maximization problem $E_{\max} = \max_{|SD\rangle} E(\ket{SD})$ can be broken down into three tasks:
\begin{enumerate} [ {(}1{)}]
	\item compute the orbital-reduced state $\rho_1$ (or, equivalently, $\rho_2$) of an arbitrary SD, relative to one of the two regions of the lattice;
	\item compute $E(\ket{SD}) = S(\rho_1)$;
	\item maximize $E(\ket{SD})$ over all SDs.
\end{enumerate}

We carry out these tasks using the notion of free states and the corresponding 1RDM formalism. We denote by $\L(\H)$ the set of linear operators on $\H$. We define the free states~\cite{Araki70, BLS94, Bach22, Gottlieboriginal, Gottliebfermbos} as the number-conserving Gaussian states. More precisely, the set of free states on the Fock space $\F_{\H_1}$ is the closure of the set of states of the form
\begin{equation}
	\rho = \frac{\e^{K}}{Z} \; , \;\; Z = \tr [ \e^{K} ] \; .
\end{equation}
Here, $K \in \L(\F_{\H_1})$ is any one-body Hermitian operator, i.e., $K$ can be written as $\sum_{j,k} \kappa_{jk} f^\dagger_j f_k$, with $\overline{\kappa_{kj}} = \kappa_{jk}$, in terms of the creation and annihilation operators $f^{(\dagger)}_j$ associated with the reference ONB of $\H_1$. Taking the closure ensures that the definition incorporates limit points, including SD states, which correspond to the eigenvalues of the matrix $\kappa$ diverging to $\pm \infty$. Free states are mathematically simple to deal with: by Wick's theorem~\cite{wick_gaudin, molinari_wick}, they are entirely specified by their 1P-reduced operator $\gamma \in \L(\H_1)$, i.e., the operator on $\H_1$ whose matrix representation is the 1RDM of $\rho$,
\begin{equation}
	\braket{j | \gamma | k} = \tr [ \rho f^\dagger_k f_j ] \; .
\end{equation}
In particular, SD states are characterized as the states whose 1P-reduced operators are idempotent, $\gamma = \gamma^2$, i.e., orthogonal projectors acting on $\H_1$. We refer to Refs.~\cite{Gottliebfermbos, Peschel_2003, peschel_eisler, Bach22,BLS94} for more on free states and their one-to-one correspondence with their 1P-reduced operators (equivalently, 1RDMs).

Task (1) can be accomplished using the following standard result on the orbital-reduced states of free states~\cite{BLS94,Peschel_2003}, of which SD states are a special case. We defer the proof to Appendix~\ref{app:proofs_for_basic_propositions}.
\begin{prop}[Refs.~\cite{BLS94,Peschel_2003}] \label{prop:marginalsoffreestates} %previously restatable
	Let $\rho \in \L(\F_{\H_1})$ be a fermionic free state, let $\gamma \in \L(\H_1)$ be its 1P-reduced operator, and let $V$ be a linear subspace of $\H_1$. Then the orbital-reduced state of $\rho$ relative to (the orbitals in) $V$ is the free state on the Fock space $\F_{V}$ whose 1P-reduced operator is the compression of $\gamma$ to $V$.
\end{prop}

Task (2), i.e., the computation of $E(\ket{SD}) = S(\rho_1)$, can then be simplified using the fact that $\rho_1$ is free. Indeed, since a free state is fully determined by its 1P-reduced operator, in particular its entropy is, too. More precisely, the latter is determined by its \emph{natural occupation numbers} (NONs), i.e., the eigenvalues of its 1P-reduced operator. This follows from the unitary invariance of entropy. This yields a simple formula for the entropy of any free state. To that end, we define the binary entropy~\cite{Gottlieboriginal,BLS94} as
\begin{equation} \label{eq:function_h}
	h(x) := -x \log(x)-(1-x)\log(1-x)
\end{equation}
for $x \in (0,1)$, and $h(x) := 0$ for $x \in \{ 0, 1\}$. Also, we define a modified Shannon entropy $H$ as follows:
\begin{equation} \label{eq:function_Hadamard}
	H(\boldsymbol{x}) := \sum_{j=1}^{d} h(x_j)
\end{equation}
for $\boldsymbol{x} \in [0,1]^d$.

The following expression for the entropy of a free state is a standard result~\cite{Gottlieboriginal,BLS94}. For completeness, we provide a proof in Appendix~\ref{app:proofs_for_basic_propositions}.
\begin{prop}[Refs.~\cite{Gottlieboriginal,BLS94}] \label{prop:entropyfreestates} % previously restatable
	Let $\rho \in \L(\F_{\H_1})$ be a fermionic free state. Let $\gamma$ be the 1P-reduced operator of $\rho$ and $n_1,\dots,n_d$ the eigenvalues of $\gamma$, i.e., the NONs of $\rho$, with $d = \emph{dim}(\H_1)$. Then
	\begin{equation} \label{eq:entropy_of_free_state}
		S(\rho) = \tr [h(\gamma)] = H(\bn) \,.
	\end{equation}
\end{prop}
The above propositions allow us to reduce the computation of $E_{\max}$ from a maximization task in the $\binom{d}{N}$-dimensional $N$-fermion Hilbert space to one in the $d^2$-dimensional 1P operator space $\L(\H_1)$. Given an SD of $N$ fermions $\ket{SD}$, one can compute $E(\ket{SD})$ as follows: write down the $d \times d$ 1RDM of $\ket{SD}$ in the basis $\{\ket{j}\}_{j=1}^d$; consider its upper-left $\lambda_1 \times \lambda_1$ diagonal block (one may equivalently choose the lower-right block); take the eigenvalues $b_1, \dots, b_{\lambda_1}$ of that block; compute
\begin{equation}
	E(\ket{SD}) = H(\bb) = \sum_{j = 1}^{\lambda_1} h(b_j) \; .
\end{equation}

Task (3) is then about maximizing the above expression. Taking into account that the 1RDMs of $N$-particle SDs are exactly the rank-$N$ orthogonal projection matrices (with complex entries), we can compactly write the maximum as
\begin{equation} \label{eq:Emax_rewritten}
	E_{\max} = \max_{\gamma} \sum_{j = 1}^{\lambda_1} h \left (\specdown_j(B_1^\gamma) \right ) \; ,
\end{equation}
where the optimization is over rank-$N$ orthogonal projection matrices $\gamma$ of size $d \times d$. Here and in the following, $B_\mu^\gamma$ denotes the $\mu$th diagonal block of $\gamma$ relative to an underlying partition of $d$, and $\specdown$ denotes the vector of possibly repeated eigenvalues, in non-increasing order. Notice that many-body states no longer appear. The maximization in Eq.~\eqref{eq:Emax_rewritten} requires a characterization of the possible spectra of a given diagonal block of a Hermitian matrix with given spectrum. Such a characterization is given in Ref.~\cite{Fan1957Imbedding} in terms of linear inequalities (see also Prop.~\ref{prop:bset_bipartitions_SD} in Appendix~\ref{app:main_prob_revisited}): any real numbers $b_1 \geq \ldots \geq b_{\lambda_1}$ constitute the spectrum $\specdown(B_1^\gamma)$ for some rank-$N$ orthogonal projection matrix $\gamma$ of size $d = \lambda_1 + \lambda_2$ if and only if
\begin{equation} \label{eq:block_spec_cond1}
	(b_1, \dots , b_{\lambda_1} ) = (\underbrace{1,\dots,1}_{T_1^P \emph{ times}}, \beta_1,\dots,\beta_R,\underbrace{0,\dots,0}_{T_1^H\emph{ times}} ) \; ,
\end{equation}
where
\begin{equation} \label{eq:block_spec_cond2}
1 \geq \beta_1\geq\dots\geq \beta_R \geq 0
\end{equation}
and
\begin{align} \label{eq:block_spec_cond3}
	T^P_1 &= \max \{ 0, \lambda_{1} - (d - N ) \} \\ \label{eq:block_spec_cond4}
	T^H_1 &= \max \{ 0, \lambda_{1} - N \} \\ \label{eq:block_spec_cond5}
	R &= \min \{ \lambda_1, \lambda_2, N, d - N \} \; . 
\end{align}
The $T^{P/H}_1$ components $b_j$ equal to $1$/$0$ correspond to fully occupied/empty orbitals in the first region of the lattice. These orbitals are necessarily present if $\lambda_1$ is large enough (see also the discussion following Result~\ref{prop:Emax_SD_bipartition}). The superscripts $P$ and $H$ stand for particles and holes, respectively.

Since $h$ has its maximum at input value $1/2$, we finally get
\begin{equation} \label{eq:maximal_sum_h}
	E_{\max} = \max_{\beta_1, \dots, \beta_R} \sum_{j = 1}^R h(\beta_j) = R \, h(1/2) = R \, \log(2) \; ,
\end{equation}
i.e., maximum entanglement is achieved by setting $\beta_1 = \ldots = \beta_R = 1/2$ and is indeed given by Eq.~\eqref{eq:Emax_result}. The choice $\beta_j = 1/2$ corresponds to the unbiased superpositions of orbital pairs described in Fig.~\ref{fig:max_SD_construction} and in the discussion of Sec.~\ref{sec:main_problem}.

This completes the proof of Result~\ref{prop:Emax_SD_bipartition}.

\subsection{Towards the general problem} \label{subsec:towards_general}

Props.~\ref{prop:marginalsoffreestates} and~\ref{prop:entropyfreestates} enable the computation of bipartite orbital entanglement in SD states. Likewise, they will simplify the more general problem of computing the \emph{multipartite} total (classical + quantum) orbital correlation in arbitrary \emph{free} states, the object of Sec.~\ref{sec:generalizing}. However, the \emph{maximization} will, in the general case, constitute a formidable and multifaceted challenge. In that light, we now examine and unpack further the maximization leading to $E_{\max}$ for bipartite entanglement and SDs.

From this point forward, we adopt the \emph{total} correlation $C$ (see Appendix~\ref{app:mutual_info_fermions}) rather than the entanglement entropy $E$ as the orbital correlation measure. Its definition
\begin{equation} \label{eq:total_corr}
	C(\rho) = \left (\sum_{\mu = 1}^k S(\rho_\mu) \right ) - S(\rho) \; , \;\;\; \rho_\mu = \tr_{\setminus \mu} \, \rho \, ,
\end{equation}
remains meaningful for mixed states of $k$-partite systems, as needed for the general problem of Sec.~\ref{sec:generalizing}.
%For bipartite systems, $C$ reduces to the quantum \emph{mutual information}.
%for any $k$, $C(\rho) = D(\rho , \rho_1 \otimes \dots \otimes \rho_k)$ is the minimal relative entropy of $\rho$ to the uncorrelated states (see Appendix~\ref{app:mutual_info_fermions}). 
Beyond pure bipartite states, $C$ no longer reduces to twice the entanglement: it also comprises quantum correlations beyond entanglement and classical correlations. By considering $C$ we thus assess \emph{all} correlations that non-interacting fermions, uncorrelated in the particle picture, exhibit in the orbital picture. Operationally, $C$ quantifies the cost of erasing correlations~\cite{GroismanPopescuWinter2005} and the work extractable from them by global operations~\cite{PerarnauLlobet2015}; moreover, unlike most entanglement measures, it is computable in closed form for all free states. 

As noted in Sec.~\ref{sec:main_problem}, for a pure SD $\ket{SD}$ with orbital-reduced states $\rho_1$ and $\rho_2$, $C$ equals $E$ up to a prefactor:
\begin{align} \nonumber
	C(\ket{SD}\bra{SD}) &= S(\rho_1) + S(\rho_2) - \underbrace{S(\ket{SD}\bra{SD})}_{= 0} \\ \label{eq:mutual_info} 
	&= 2 E(\ket{SD}) \; .
\end{align}
We define
\begin{equation} \label{eq:first_def_of_Imax}
	C_{\max} := \max_{|SD\rangle} C (\ket{SD}\bra{SD}) = 2\, E_{\max} \; ,
\end{equation}
where the maximization is over all $N$-fermion Slater determinants $\ket{SD}$. In close correspondence with Eq.~\eqref{eq:Emax_rewritten}, we recast the maximization as one over rank-$N$ orthogonal projection matrices $\gamma$ of size $d \times d$:
\begin{equation} \label{eq:I_max_rewritten}
	C_{\max} = \max_{\gamma} \left [ \sum_{j = 1}^{\lambda_1} h \left (\specdown_j (B_1^\gamma) \right ) + \sum_{j = 1}^{\lambda_2} h \left (\specdown_j (B_2^\gamma) \right )  \right ] \; .
\end{equation}
For any $\gamma$, the two sums are equal, because they equal the entropies of the two isospectral orbital-reduced states of the SD with 1RDM $\gamma$.

We now proceed a step further than in Sec.~\ref{subsec:proof_of_main_theorem}. The objective function of the above maximization clearly depends on the matrix $\gamma$ only through the eigenvalues of its diagonal blocks. This implies some redundancy and the non-uniqueness of maximizers. For example, given a maximizer $\gamma_{\max}$, we get infinitely many new maximizers by unitary conjugation $\gamma_{\max} \rightarrow u \gamma_{\max} u^\dagger$ with any $d \times d$ unitary $u = u_1 \oplus u_2$ that is block diagonal with respect to the bipartition $\blam$ of $d$ (see also the discussion following Result~\ref{prop:Emax_SD_bipartition}). As a further step toward later addressing the general problem, we get rid of this redundancy by replacing the variable $\gamma$ to be optimized over with its \emph{block spectrum} $\bb = (\specdown (B_1^\gamma) \sbs \specdown (B_2^\gamma) )  \in \R^{\lambda_1} \times \R^{\lambda_2} = \R^d$. We will use the vertical bar to separate the spectra of the blocks, and, with slight abuse of notation, denote with $(\bv, \boldsymbol{w})$ the concatenation of two vectors $\bv, \boldsymbol{w}$. We denote with $\bset \subset \R^d$ the compact set of all blockwise-ordered block spectra $\bb$ of arbitrary $d \times d$ rank-$N$ projection matrices. We finally rewrite Eq.~\eqref{eq:I_max_rewritten} as
\begin{equation} \label{eq:I_max_re_rewritten}
	C_{\max} = \max_{\bb \in \bset} H(\bb) \; .
\end{equation}
This formulation of the maximization highlights that the mathematical problem can be divided into two tasks:
\begin{enumerate} [ {(}1{)}]
	\item the set $\bset$ of block spectra must be, at least partially, characterized;
	\item the objective function, a modified Shannon entropy, must be maximized over $\bset$.		
\end{enumerate}
Task (1) is an adaptation of the quantum marginal problem~\cite{klyachko2004, bravyi2003q_marg_problem, Eisert2008_Gaussian_q_marg_prob} to fermionic states: roughly speaking, characterizing $\bset$ means characterizing the possible pairs of orbital-reduced states of an SD state. This perspective on the initial optimization problem is analogous to breaking down the ground state energy problem for systems with two-body Hamiltonians (e.g., chemical systems) into (1) the $N$-fermion two-body reduced density matrix representability problem~\cite{rep_problem} and (2) the subsequent energy minimization.

In Appendix~\ref{app:main_prob_revisited} we revisit the computation of $C_{\max} = 2E_{\max}$ from this new perspective, in view of generalizing our problem. Crucially, we show that there exists an element $\bb_{\max} \in \bset$ that is \emph{majorized} by all other elements $\bb \in \bset$, $\bb_{\max} \prec \bb$ (majorization is introduced in Appendix~\ref{app:majorization_schur_horn_cauchy}; see also Refs.~\cite{nielsenmajo,MarshallBook}). $\bb_{\max}$ follows as the \emph{unique maximizer} of the right-hand side of Eq.~\eqref{eq:I_max_re_rewritten} due to the strict \emph{Schur-concavity}~\cite{MarshallBook} of the objective function $H$, i.e., if $\bx \succ \by $ and $\bx $ is not a permutation of $\by$ then $  H(\bx) < H(\by)$.

\section{General problem: Maximal multipartite orbital correlation} \label{sec:generalizing}

In this section, we formulate the general problem of this work. It extends the basic problem of Sec.~\ref{sec:main_problem} in two directions: from bipartitions to \emph{arbitrary partitions} of the lattice, and from SDs to arbitrary, possibly mixed, \emph{free} states (see Sec.\til\ref{subsec:proof_of_main_theorem}). We focus on fermions. The analysis carries over to bosonic free states (Appendix~\ref{app:bosons}), which, however, are always separable~\cite{wolf_ent_power, ciancio_zanardi}, so that for bosons the maximized correlation contains no entanglement.

Two main reasons motivate the extension from SDs to arbitrary free states. First, by Prop.~\ref{prop:marginalsoffreestates}, free states include the orbital-reduced state of any orbital subsystem of a fermionic system whose pure state is (modeled as) an SD. A typical example is a fragment---a subset of nuclei and their associated electron clouds---of a molecule that is described via Hartree-Fock theory. Second, the grand canonical Gibbs states of non-interacting fermions at finite temperature are, by definition, free states.

We will use the total correlation of Eq.~\eqref{eq:total_corr} as the orbital correlation measure to be maximized.

As before, we consider identical non-interacting spinless fermions on a lattice of $d$ sites, described by the 1P space $\H_1 = \text{span}(\ket{1}, \dots , \ket{d})$, but we no longer assume the number of fermions to be fixed. The lattice is partitioned into $k \geq 2$ regions. For $1 \leq \mu \leq k$, region $\mu$ comprises $\lambda_\mu$ sites. Here, $(\lambda_\mu)_{\mu = 1}^k = \blam$ is an integer partition of the lattice size $d$, i.e., $\blam \in \mathbb{N}^k$, $\lambda_1 \geq \dots \geq \lambda_k \geq 1$ and $\sum_{\mu = 1}^k \lambda_\mu = d$. Region 1 consists of the sites $j = 1, \dots, \lambda_1$, region 2 consists of the sites $j = \lambda_1 + 1, \dots, \lambda_1 + \lambda_2$, and so on. We do not assume any particular lattice geometry or dimensionality, but we do assume that the $k$ regions are spatially separated and individually accessible by experiments, so that it is meaningful to study the orbital correlations across them. The lattice partition induces the splitting of $\H_1$ as the direct sum of the $k$ orthogonal subspaces $V_\mu$, defined by $V_1 = \text{span}(\ket{1}, \dots, \ket{\lambda_1}), V_2 = \text{span}(\ket{\lambda_1 + 1}, \dots, \ket{\lambda_1 + \lambda_2})$ and so on, i.e., $V_\mu$ is the 1P space for a particle in the lattice region $\mu$.

Let $\rho$ be the initial state of the system, acting on the Fock space $\F_{\H_1}$. For now we only assume that $\rho$ satisfies the number-parity SSR, i.e., that it commutes with the parity operator, as all physical fermionic states do~\cite{wick1952intrinsic,hegerfeldt1968proof,Friis_2016,johansson2016comment}. We consider the problem of maximizing the total correlation $C^{\blam}$---the superscript specifies the underlying splitting of the orbital space---over arbitrary 1P unitary evolutions. The latter are described by unitary 1P operators on the Fock space, now best represented as $U = \e^{- \i \sum_{j,k} \eta_{jk} f^\dag_j f_k}$ with $\eta_{jk} = \overline{\eta_{kj}}$. In other words, we are interested in the following quantity.
\begin{defn}
	Let $\rho$ be a parity-superselected state on $\F_{\H_1}$. For a given integer partition $\blam$ of the dimension $d$ of $\H_1$, we define the maximal orbital correlation
\begin{equation} \label{eq:Iblam_max_def}
	C^{\blam}_{\max}(\rho) := \max_U C^{\blam} ( U \rho U^\dagger ) \; ,
\end{equation}
where the maximization is over all 1P unitary operators $U$ on $\F_{\H_1}$.
\end{defn}
In the special case where $\rho = \ket{SD} \bra{SD}$ is a Slater determinant state and $k = 2$, i.e., $\blam$ is a bipartition of $d$, $C_{\max}^{\blam}(\rho)$ reduces to $C_{\max} = 2 E_{\max}$, as defined in Eqs.\til\eqref{eq:E_max} and\til\eqref{eq:first_def_of_Imax}.

We now restrict our attention to free states $\rho$. As seen in Sec.~\ref{sec:math_structure}, their orbital-reduced states and their entropy can be computed at the level of the 1P space. Let $\gamma \in \L(\H_1)$ be the 1P-reduced operator of $\rho$, and $\bn$ the decreasingly ordered spectrum (with repetitions) of $\gamma$, i.e., the vector of $\rho$'s NONs. Given a unitary $u$ on $\H_1$, let $U$ be the 1P unitary which implements $u$ on the Fock space $\F_{\H_1}$. Then $U \rho U^\dagger$ is the unique free state with 1P-reduced operator $u \gamma u^\dagger$, and thus has the same NONs as $\rho$. Vice versa, any free state with NONs $\bn$ can be represented as $U \rho U^\dagger$ with a suitable 1P unitary $U$. $C_{\max}^{\blam}(\rho)$ can thus be interpreted as the maximal orbital correlation of free states having the same NONs $\bn$ as $\rho$, and it can be expressed as a function of $\bn$ only, as follows. A generalization of the central Eq.\til\eqref{eq:I_max_rewritten} leads to the following compact formulation of the general goal of this work. We denote the $\mu$th diagonal block of a matrix $\gamma$ by $B^\gamma_\mu$, which implicitly depends on the partition $\blam$ of the matrix size $d$.
\begin{tcolorbox}[
colframe=blue!70!black, % Line color
colback=blue!17,         % Background color
boxrule=1.2pt,          % Line width
arc=6pt,                % Corner radius (increase for more "roundness")
boxsep=0pt,
left=10pt,              % Inner margins
right=10pt,
top=7pt,
bottom=7pt,
enhanced            
]
\textbf{General Scientific Problem.} Determine the maximal orbital total correlation of a fermionic free state $\rho$ with given NONs $n_1, \ldots, n_d$ on a $k$-partite lattice of $ d = \sum_{\mu=1}^k \lambda_\mu$ sites, i.e.,
\begin{equation} \label{eq:Iblam_max_rewritten}
	C^{\blam}_{\max}(\rho) =  \max_{\substack{\gamma = \gamma^\dagger\\ \specdown(\gamma) = \bn}} \left [ \sum_{\mu = 1}^k  H(\specdown(B^\gamma_{\mu})) \right ] \; - H(\bn) \; .
\end{equation}
\end{tcolorbox}
The maximization is over all 1RDMs of arbitrary free states $U \rho U^\dagger$ with NONs $\bn$, which clarifies the connection with Eq.~\eqref{eq:Iblam_max_def}. The term $H(\bn)$ in Eq.\til\eqref{eq:Iblam_max_rewritten}, a rewrite of the entropy $S(U \rho U^\dagger) = S(\rho)$, is a constant and was thus moved outside the maximization. The quantity to be maximized depends on $\gamma$ only through its $\blam$-block spectrum, i.e., $\bb = ( \specdown(B^\gamma_{1}) \sbs \specdown(B^\gamma_{2}) \sbs \dots \sbs \specdown(B^\gamma_{k}) )$. The latter can thus be considered the variable to be optimized over, instead of $\gamma$, which leads to the following definition and proposition.
\begin{defn} \label{def:b_set}
	Let $\bn \in \R^d$ be a non-increasingly ordered vector. For a given integer partition $\blam = (\lambda_1, \dots , \lambda_k)$ of $d$, we define
	\begin{align} \nonumber
		\bset^{\blam}_{\bn} := \big \{ & \bb = \left ( \specdown ( B^\gamma_1 ) \sbs \dots \sbs \specdown ( B^\gamma_k ) \right ) \sbs  \\ \label{eq:def_bset}
		& \gamma = \gamma^\dagger \in \textup{Mat}_{d \times d}(\C) \, , \; \specdown(\gamma) = \bn \big \} \; ,
	\end{align}
	where, for $\mu = 1, \dots, k$, $B^\gamma_\mu$ is the $\mu$th diagonal block of $\gamma$ and has size $\lambda_\mu \times \lambda_\mu$.
\end{defn}
\begin{prop} \label{prop:Imax_max_over_bset}
Let $\rho$ be a free state on $\F_{\H_1}$ and let $\bn$ be the vector of its non-increasingly ordered NONs. Then
\begin{equation} \label{eq:Imax_max_over_bset}
	C^{\blam}_{\max} (\rho) = \max_{\bb \in \bset^{\blam}_{\bn}} \left ( H(\bb) \right ) - H(\bn) \; .
\end{equation}
\end{prop}
Prop.~\ref{prop:Imax_max_over_bset}, a generalization of Eq.~\eqref{eq:I_max_re_rewritten}, is trivially proved by rewriting Eq.\til\eqref{eq:Iblam_max_rewritten} in light of Def.\til\ref{def:b_set}. It implies that $C^{\blam}_{\max} (\rho)$ can be computed by (partially) characterizing the set $\bset^{\blam}_{\bn}$ and maximizing $H$ over it, as in Sec.~\ref{subsec:towards_general}.

The algebraic problem of characterizing $\bset^{\blam}_{\bn}$ has been successfully addressed in Ref.\til\cite{lipoon}: in particular, $\bset^{\blam}_{\bn}$ is a \emph{convex polytope}. This is presumably related to the representability of $\bset^{\blam}_{\bn}$ as the moment polytope of a suitable moment map~\cite{mike_on_kirwan,Kirwan, onKirwan}. In Ref.~\cite{Eisert2008_Gaussian_q_marg_prob}, the characterization of the analogue of $\bset^{\blam}_{\bn}$ for bosonic Gaussian states is discussed, and is referred to as Gaussian quantum marginal problem. Interestingly, both the fermionic and the bosonic problems can be reinterpreted as the problem of characterizing the 1RDMs of the possible orbital-reduced states of an \emph{arbitrary} (i.e., not necessarily free) many-body state whose total 1RDM has a fixed spectrum.

As a side comment, we recall that $C_{\max}^{\blam}(\rho)$ can be reinterpreted as the total correlation of $\rho$ between orthogonal \emph{subspaces} $W_1, \ldots, W_k$ of $\H_1$ maximized over the choice of these subspaces, constrained by dim$(W_\mu) = \lambda_\mu$ for $\mu = 1, \ldots, k$ and $\oplus_{\mu =1}^k W_\mu = \H_1$ (see Sec.~\ref{sec:main_problem}).

\section{Mathematical structure of the general problem} \label{sec:math_struc_gen_prob}

In this section, we establish the mathematical structure underlying all solutions of Sec.~\ref{sec:solutions_general}: existence of maximizers, general bounds, the central role of majorization, and our solution strategy (Fig.~\ref{fig:strategy}). We also collect all closed-form maximizers we found (Tables~\ref{tab:tabledgen}--\ref{tab:tabled5}).

Our approach to tackle the general problem starts from Prop.~\ref{prop:Imax_max_over_bset}: for given $d, \blam, \bn$, we need to find a maximizer $\bb_{\max} \in \bset^{\blam}_{\bn}$ for the right-hand side of Eq.~\eqref{eq:Imax_max_over_bset}. If such a maximizer is given, the maximal orbital correlation of any free state $\rho$ with NONs $\bn$, relative to $\blam$, simply follows as
\begin{equation} \label{eq:obtain_Imax_from_bmax}
	C^{\blam}_{\max}(\rho) = H(\bb_{\max}) - H(\bn) \; .
\end{equation}
Notice that a maximizer always (i.e., $\forall$ $\blam$ and $\bn$) exists, since the objective function $H$ is continuous and $\bset_{\bn}^{\blam}$ is compact---it is closed (i.e., it contains its limit points) and is contained in the bounded set $[n_d, n_1]^d$.

The convex polytope $\bset_{\bn}^{\blam}$ is characterized in Ref.~\cite{lipoon} in terms of facet-defining inequalities. These are complicated to construct and highly redundant. Therefore, we will mostly rely on alternative---and possibly partial---characterizations of $\bset^{\blam}_{\bn}$ that are less general but simpler to handle. In particular, the Schur-Horn theorem~\cite{schur1923klasse, Horn1954DoublySM} and the Cauchy interlacing inequalities~\cite{Fan1957Imbedding} provide a simple characterization of $\bset^{\blam}_{\bn}$ for two special types of partitions $\blam$: trivial partitions $\blam = (1, \ldots, 1)$ and bipartitions of the type $\blam = (d - 1, 1)$, respectively. We recall these well-known results in Appendix~\ref{app:majorization_schur_horn_cauchy}. In particular, the Schur-Horn theorem leads to the following proposition, also proved in Appendix~\ref{app:majorization_schur_horn_cauchy}.
\begin{prop} \label{prop:basic_bounds} %previosuly restatable
	Let $\rho$ be a fermionic free state with NONs $n_1, \ldots, n_d$ and filling fraction $\eta = \sum_{j=1}^d n_j/d$. Then
	\begin{equation} \label{eq:basic_bounds}
		0 \leq C_{\max}^{\blam}(\rho) \leq d \cdot h ( \eta ) - H(\bn) \; .
	\end{equation}
	In particular, the upper bound is saturated if and only if $\bset^{\blam}_{\bn}$ contains the maximally mixed element $(\eta, \ldots, \eta)$.
\end{prop}
Indeed, there exist both instances $d, \blam, \bn$ such that $\bset^{\blam}_{\bn}$ does contain the maximally mixed vector $(\eta, \ldots, \eta)$, and instances such that it does not (see Appendix\til\ref{app:3D_space}). If $(\eta, \ldots, \eta) \in \bset^{\blam}_{\bn}$, then this is the unique maximizer $\bb_{\max}$ of our objective function.

In more detail, the function $H$ defined on $\{ \bb \in \R^d | \sum_{j} b_j = \sum_{j} n_j \}$ has a unique stationary point, namely the maximally mixed vector. This can be checked by the method of Lagrange multipliers. In the generic case where the polytope $\bset^{\blam}_{\bn}$ has codimension 1 in $\R^d$, this implies that any (local or global) maximizer for our problem is either maximally mixed or lies on the relative \emph{boundary} of $\bset^{\blam}_{\bn}$. Furthermore, convexity of $\bset^{\blam}_{\bn}$ and the strict concavity of $H$ imply that any local maximizer is a \emph{global} maximizer, too; and that there is a \emph{unique} (global) maximizer. We can then put forward the following definition.

% apparently some empty  lines are needed here, theres a latex bug. If whole doc does not compile, try removing the above itemizer and below defn...

\begin{defn}
	We denote by $\boldsymbol{b}^{\boldsymbol{\lambda}}_{\max}(\boldsymbol{n})$ the unique maximizer of $H$ over $\mathfrak{B}^{\boldsymbol{\lambda}}_{\boldsymbol{n}}$.
\end{defn}
We will simplify the notation $\bb^{\blam}_{\max}(\bn)$ to $\bb_{\max}$ when the integer partition $\blam$ and the vector $\bn$ are clear from the context.

It may happen that there is an element in $\bset^{\blam}_{\bn}$ that is majorized by every other $\bb \in \bset^{\blam}_{\bn}$, whether it is the maximally mixed vector or not. If so, such an element is the maximizer $\bb^{\blam}_{\max}(\bn)$, because our objective function is Schur-concave (see Sec.~\ref{subsec:towards_general}). Crucially, this was the case for \emph{every} instance of $d, \blam$ and $\bn$ we studied: we were always able to perform the maximization task by exhibiting a vector that is least in $\bset^{\blam}_{\bn}$ with respect to the majorization preorder, and thus is the maximizer $\bb^{\blam}_{\max}(\bn)$.

\afterpage{
\renewcommand{\arraystretch}{1.3}

%%%%%%%%%%%%%% starting table for arbitrary d %%%%%%%%%%%%%%%%%%%%%%%%%%%%
\begin{table*}
\begin{center}

		\centering
		\begin{tabular}{|lll|}
			\hline
			\multicolumn{3}{|l|}{Arbitrary $d$}
			\\ \hline
			\multicolumn{1}{|l|}{$\blam$} & \multicolumn{1}{l|}{Regime} & $\bb^{\blam}_{\max}(\bn)$ \\ \hline
			\multicolumn{1}{|l|}{$(\lambda_1, \lambda_2)$}         & \multicolumn{1}{l|}{Any $\bn$}       & $(\frac{n_1 + n_d}{2}, \frac{n_2 + n_{d-1}}{2}, \ldots , \frac{n_{\lambda_2} + n_{\lambda_1 + 1}}{2}, n_{\lambda_2+1}, \ldots, n_{\lambda_1} \sbs \frac{n_1 + n_d}{2}, \frac{n_2 + n_{d-1}}{2}, \ldots , \frac{n_{\lambda_2} + n_{\lambda_1 + 1}}{2})$         \\ \hline
			\multicolumn{1}{|l|}{(1, \ldots ,1)}         & \multicolumn{1}{l|}{Any $\bn$}       & $(\eta \sbs \eta \sbs \ldots \sbs \eta)$         \\ \hline
		\end{tabular}
		\caption{Formulas for $\bb^{\blam}_{\max}(\bn)$ (up to a reordering of the entries) for arbitrary $d$ and bipartitions $\blam = (\lambda_1, \lambda_2)$ or trivial partitions $\blam = (1, \ldots , 1)$. Here $\eta = \sum_j n_j / d $ is the filling fraction of the underlying free states.}
		\label{tab:tabledgen} 
\end{center}

\vspace{1.5cm}

	\centering
	\begin{minipage}{0.48\textwidth}
%%%%%%%%%%%% starting table d=2 %%%%%%%%%%%%%%%%%%
	\begin{tabular}{|lll|}
	\hline
	\multicolumn{3}{|l|}{$d = 2$}
	\\ \hline
	\multicolumn{1}{|l|}{$\blam$} & \multicolumn{1}{l|}{Regime} & $\bb^{\blam}_{\max}(\bn)$ \\ \hline
	
	\multicolumn{1}{|l|}{(1,1)}         & \multicolumn{1}{l|}{Any $\bn$}       & $(\frac{n_1+n_2}{2} \sbs \frac{n_1+n_2}{2})$         \\ \hline
	\end{tabular}
	\caption{Formula for $\bb^{\blam}_{\max}(\bn)$ for $d=2$ and the partition $\blam = (1,1)$.}
	\label{tab:tabled2}
	\vspace{.5cm}
	
%%%%%%%%%%%% starting table d=3 %%%%%%%%%%%%%%%%%%
	\begin{tabular}{|lll|}
		\hline
		\multicolumn{3}{|l|}{$d = 3$}
		\\ \hline
		\multicolumn{1}{|l|}{$\blam$} & \multicolumn{1}{l|}{Regime} & $\bb^{\blam}_{\max}(\bn)$ \\ \hline
		
		\multicolumn{1}{|l|}{(2,1)}         & \multicolumn{1}{l|}{Any $\bn$}       & $(\frac{n_1+n_3}{2}, n_2 \sbs \frac{n_1+n_3}{2})$         \\ \hline
		\multicolumn{1}{|l|}{(1,1,1)}         & \multicolumn{1}{l|}{Any $\bn$}       & $(\eta \sbs \eta \sbs \eta)$         \\ \hline
	\end{tabular}
	\caption{Formulas for $\bb^{\blam}_{\max}(\bn)$ (up to a reordering of the entries) for $d=3$ and all partitions $\blam$. Here $\eta = \sum_j n_j / 3 $.}
	\label{tab:tabled3}
	\end{minipage}
	\hfill
	\begin{minipage}{0.48\textwidth}
%%%%%%%%%%%% starting table d=4 %%%%%%%%%%%%%%%%%%
	\begin{tabular}{|lll|}
		\hline
		\multicolumn{3}{|l|}{$d= 4$}                                                                \\ \hline
		\multicolumn{1}{|l|}{$\blam$}                & \multicolumn{1}{l|}{Regime} & $\bb^{\blam}_{\max}(\bn)$ \\ \hline
		\multicolumn{1}{|l|}{(3,1)}                    & \multicolumn{1}{l|}{Any $\bn$}       & $(\frac{n_1+n_4}{2},n_2,n_3 \sbs \frac{n_1+n_4}{2})$        \\ \hline
		\multicolumn{1}{|l|}{(2,2)}                    & \multicolumn{1}{l|}{Any $\bn$}       & $(\frac{n_1+n_4}{2},\frac{n_2+n_3}{2} \sbs \frac{n_1+n_4}{2},\frac{n_2+n_3}{2})$         \\ \hline
		
		\multicolumn{1}{|l|}{(2,1,1)} & \multicolumn{1}{l|}{$n_1\geq \eta > n_2$}       & $(\frac{n_1+n_3+n_4}{3}, n_2 \sbs \frac{n_1+n_3+n_4}{3} \sbs \frac{n_1+n_3+n_4}{3})$        \\ \cline{2-3} 
		\multicolumn{1}{|l|}{}                         & \multicolumn{1}{l|}{$n_2 \geq \eta \geq n_3$}       & $(\eta,\eta \sbs \eta \sbs \eta)$        \\ \cline{2-3} 
		\multicolumn{1}{|l|}{}                         & \multicolumn{1}{l|}{$n_3 > \eta \geq n_4$}       & $(n_3,\frac{n_1+n_2+n_4}{3} \sbs \frac{n_1+n_2+n_4}{3} \sbs \frac{n_1+n_2+n_4}{3})$         \\ \hline
		
		\multicolumn{1}{|l|}{(1,1,1,1)}                & \multicolumn{1}{l|}{Any $\bn$}       & $(\eta \sbs \eta \sbs \eta \sbs \eta)$         \\ \hline
	\end{tabular}
	\caption{Formulas for $\bb^{\blam}_{\max}(\bn)$ (up to a reordering of the entries) for $d=4$ and all partitions $\blam$. Here $\eta = \sum_j n_j / 4 $. For $\blam = (2,1,1)$ we notice the emergence of different regimes (see Appendix~\ref{app:regimes}).}
	\label{tab:tabled4}
	\end{minipage}

	\vspace{1.5cm}

%%%%%%%%%%%% starting table d=5 %%%%%%%%%%%%%%%%%%

	\centering
	\begin{tabular}{|lll|}
		\hline
		\multicolumn{3}{|l|}{$d= 5$}                                                                  \\ \hline
		\multicolumn{1}{|l|}{$\blam$}                  & \multicolumn{1}{l|}{Regime} & $\bb^{\blam}_{\max}(\bn)$ \\ \hline
		
		\multicolumn{1}{|l|}{(4,1)}                      & \multicolumn{1}{l|}{Any $\bn$}       & $(\frac{n_1+n_5}{2},n_2,n_3,n_4 \sbs \frac{n_1+n_5}{2})$         \\ \hline
		
		\multicolumn{1}{|l|}{(3,2)}                      & \multicolumn{1}{l|}{Any $\bn$}       & $(\frac{n_1+n_5}{2},\frac{n_2+n_4}{2},n_3 \sbs \frac{n_1+n_5}{2},\frac{n_2+n_4}{2})$         \\ \hline
		
		\multicolumn{1}{|l|}{(3,1,1)}   & \multicolumn{1}{l|}{$n_1 \geq \frac{n_1+n_2+n_4+n_5}{4} > n_2$}       & $(\frac{n_1+n_4+n_5}{3}, n_2, n_3 \sbs \frac{n_1+n_4+n_5}{3} \sbs \frac{n_1+n_4+n_5}{3})$        \\ \cline{2-3} 
		\multicolumn{1}{|l|}{}                           & \multicolumn{1}{l|}{$n_2 \geq \frac{n_1+n_2+n_4+n_5}{4} \geq n_4$}       & $(\frac{n_1+n_2+n_4+n_5}{4},\frac{n_1+n_2+n_4+n_5}{4},n_3 \sbs \frac{n_1+n_2+n_4+n_5}{4} \sbs \frac{n_1+n_2+n_4+n_5}{4})$        \\ \cline{2-3} 
		\multicolumn{1}{|l|}{}                           & \multicolumn{1}{l|}{$n_4 > \frac{n_1+n_2+n_4+n_5}{4} \geq n_5$}       & $(n_3,n_4,\frac{n_1+n_2+n_5}{3} \sbs \frac{n_1+n_2+n_5}{3} \sbs \frac{n_1+n_2+n_5}{3})$         \\ \hline
		
		\multicolumn{1}{|l|}{(2,2,1)}   & \multicolumn{1}{l|}{$n_1 \geq \eta > \frac{n_2+n_3}{2}$}       & $(\frac{n_1+n_4+n_5}{3},\frac{n_2+n_3}{2} \sbs \frac{n_1+n_4+n_5}{3},\frac{n_2+n_3}{2} \sbs \frac{n_1+n_4+n_5}{3})$        \\ \cline{2-3} 
		\multicolumn{1}{|l|}{}                           & \multicolumn{1}{l|}{$\frac{n_2+n_3}{2} \geq \eta > \frac{n_1+n_4}{2}$}       & $(\frac{n_2+n_3+n_5}{3},\frac{n_1+n_4}{2} \sbs \frac{n_2+n_3+n_5}{3}, \frac{n_1+n_4}{2} \sbs\frac{n_2+n_3+n_5}{3})$        \\ \cline{2-3} 
		\multicolumn{1}{|l|}{}                           & \multicolumn{1}{l|}{$\frac{n_1+n_4}{2}, \frac{n_2 + n_3}{2} \geq \eta \geq \frac{n_3 + n_4}{2}, \frac{n_2+n_5}{2}$}       & $(\eta,\eta \sbs \eta,\eta\sbs\eta)$          \\ \cline{2-3} 
		\multicolumn{1}{|l|}{}                           
		& \multicolumn{1}{l|}{$\frac{n_2+n_5}{2} > \eta \geq \frac{n_3+n_4}{2}$}       
		& $(\frac{n_2+n_5}{2}, \frac{n_1+n_3+n_4}{3} \sbs \frac{n_2+n_5}{2}, \frac{n_1+n_3+n_4}{3} \sbs \frac{n_1+n_3+n_4}{3})$          
		\\ \cline{2-3} 
		\multicolumn{1}{|l|}{}                      
		& \multicolumn{1}{l|}{$\frac{n_3+n_4}{2} > \eta \geq n_5$}       
		& $(\frac{n_3+n_4}{2}, \frac{n_1+n_2+n_5}{3} \sbs \frac{n_3+n_4}{2}, \frac{n_1+n_2+n_5}{3} \sbs \frac{n_1+n_2+n_5}{3})$        
		\\ \hline
		
		\multicolumn{1}{|l|}{(2,1,1,1)} & \multicolumn{1}{l|}{$n_1 \geq \eta > n_2$}       & $(\frac{n_1+n_3+n_4+n_5}{4},n_2 \sbs \frac{n_1+n_3+n_4+n_5}{4} \sbs \frac{n_1+n_3+n_4+n_5}{4} \sbs \frac{n_1+n_3+n_4+n_5}{4})$         \\ \cline{2-3} 
		\multicolumn{1}{|l|}{}                           & \multicolumn{1}{l|}{$n_2 \geq \eta \geq n_4$}       & $(\eta, \eta\sbs\eta\sbs\eta\sbs\eta)$         \\ \cline{2-3} 
		\multicolumn{1}{|l|}{}                           & \multicolumn{1}{l|}{$n_4 > \eta \geq n_5$}       & $(n_4,\frac{n_1+n_2+n_3+n_5}{4}\sbs\frac{n_1+n_2+n_3+n_5}{4}\sbs\frac{n_1+n_2+n_3+n_5}{4} \sbs \frac{n_1+n_2+n_3+n_5}{4})$         \\ \hline
		
		\multicolumn{1}{|l|}{(1,1,1,1,1)}                & \multicolumn{1}{l|}{Any $\bn$}       & $(\eta \sbs \eta \sbs \eta \sbs \eta \sbs \eta)$        \\ \hline
	\end{tabular}
	\caption{Formulas for $\bb^{\blam}_{\max}(\bn)$ (up to a reordering of the entries) for $d=5$ and all partitions $\blam$. Here $\eta = \sum_j n_j / 5 $.} 
	\label{tab:tabled5}
\end{table*}
}

We summarize all formulas we found for $\bb_{\max} = \bb^{\blam}_{\max}(\bn)$ in Tables\til\ref{tab:tabledgen},\til\ref{tab:tabled2},\til\ref{tab:tabled3},\til\ref{tab:tabled4}, and\til\ref{tab:tabled5}. These closed-form maximizers constitute the central results of this work; they are stated as Results~\ref{res:bipartitions}, \ref{res:finest}, and~\ref{res:small_d} and derived in Sec.\til\ref{sec:solutions_general}. These tables show that the entries of $\bb_{\max}$ are always \emph{averages} of a suitable subset of components $n_j$. This fact reflects the specific discrete-Fourier-transform-like structure of the \bt optimal' unitary $u$ that conjugates $\text{diag}(\bn)$ into the matrix $\gamma = u \, \text{diag}(\bn) u^\dagger$ maximizing the right-hand side of Eq.\til\eqref{eq:Iblam_max_rewritten}, i.e., into the matrix $\gamma$ whose block spectrum is precisely $\bb_{\max}$ (see, e.g., Sec.\til\ref{subsec:2D_space}).

Tables\til\ref{tab:tabled4} and\til\ref{tab:tabled5} also show that, for \emph{fixed} $d$ and $\blam$, the \emph{functional} form of the maximizer may depend on $\bn$. This induces a partition of the set of spectra $\bn$ into what we call \bt regimes' (see Appendix~\ref{app:regimes}). Our results display two further features: a manifestation of the particle-hole symmetry of fermionic systems, and the fact that equal-size blocks of $\bb_{\max}$ coincide. We briefly discuss these symmetries of the general problem in Appendix~\ref{app:symmetries}.

Finally, for many $d$ and $\blam$, tackling our maximization numerically in the form of Eq.\til\eqref{eq:Iblam_max_rewritten} was extremely useful as a tool for conjecturing formulas for $\bb_{\max}$. We performed numerical maximization using the Adam optimization algorithm~\cite{adam}. The absence of nonglobal local maximizers allows the numerical computations to converge fast enough for our purposes.

%To summarize, our general strategy to compute $C_{\max}^{\blam}(\rho)$ for given $d$, $\blam$ and NONs $\bn$ of the free state $\rho$ can be broken down as follows:
%\begin{enumerate} [ {(}1{)}]
%	\item partially or completely characterize $\bset^{\blam}_{\bn}$;
%	\item exhibit simple formulas for the components of a vector in $ \bset^{\blam}_{\bn}$ that is majorized by every other $\bb \in \bset^{\blam}_{\bn}$---such a vector is then the maximizer $\bb^{\blam}_{\max}(\bn)$;
%	\item compute $C_{\max}^{\blam}(\rho)$ using Eq.\til\eqref{eq:obtain_Imax_from_bmax}.
%\end{enumerate}

Our general strategy to compute $C_{\max}^{\blam}(\rho)$ for given $d$, $\blam$ and NONs $\bn$ is summarized in Fig.~\ref{fig:strategy}.

\usetikzlibrary{arrows.meta,positioning,calc}
%\usepackage[labelsep=period,font=normalsize]{caption}

% ---- colours of the flowchart -------------------------------------------
\definecolor{boxfill}{RGB}{237,237,252}   % very light lavender
\definecolor{boxline}{RGB}{35,35,205}     % saturated blue

% ---- node / arrow styles -------------------------------------------------
\tikzset{
  stage/.style={
    rectangle, rounded corners=3.5pt,
    draw=boxline, line width=0.9pt,
    fill=boxfill,
    align=center,
    inner xsep=6pt, inner ysep=7pt,
    text width=66mm
  },
  flow/.style={
    -{Stealth[length=3.2mm,width=2.6mm]},
    draw=boxline, line width=1.1pt
  },
  note/.style={align=left, anchor=west, inner sep=0pt}
}

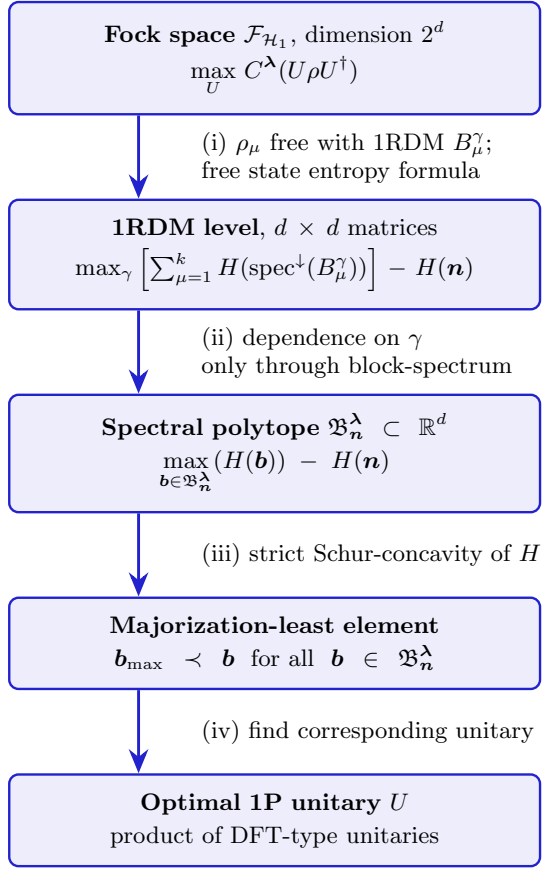
\begin{figure}[t]
\centering
\begin{tikzpicture}[node distance=11mm]

% ------------------------------- the five stages -------------------------
\node[stage] (fock) {%
  \textbf{Fock space} $\mathcal{F}_{\mathcal{H}_1}$, dimension $2^{d}$\\[2pt]
  $\displaystyle\max_{U}\, C^{\bm{\lambda}}(U\rho U^{\dagger})$};

\node[stage, below=of fock] (rdm) {%
  \textbf{1RDM level}, $d\times d$ matrices\\[2pt]
  $ \max_{\gamma} \left [ \sum_{\mu = 1}^k  H(\specdown(B^\gamma_{\mu})) \right ] - H(\bn)$};

\node[stage, below=of rdm] (poly) {%
  \textbf{Spectral polytope} $\mathfrak{B}^{\bm{\lambda}}_{\bm{n}}\subset\mathbb{R}^{d}$\\[2pt]
  $\displaystyle\max_{\bm{b}\in\mathfrak{B}^{\bm{\lambda}}_{\bm{n}}} (H(\bm{b})) - H(\bm{n})$};

\node[stage, below=of poly] (least) {%
  \textbf{Majorization-least element}\\[2pt]
  $\bm{b}_{\max}\prec\bm{b}$ \ for all \ $\bm{b}\in\mathfrak{B}^{\bm{\lambda}}_{\bm{n}}$};

\node[stage, below=of least] (rot) {%
  \textbf{Optimal 1P unitary $U$}\\[2pt]
  product of DFT-type unitaries};

% ------------------------------- the arrows -------------------------------
\newcommand{\arrowx}{-18mm}   % negative = leftwards

\draw[flow] ([xshift=\arrowx]fock.south)  -- ([xshift=\arrowx]rdm.north);
\draw[flow] ([xshift=\arrowx]rdm.south)   -- ([xshift=\arrowx]poly.north);
\draw[flow] ([xshift=\arrowx]poly.south)  -- ([xshift=\arrowx]least.north);
\draw[flow] ([xshift=\arrowx]least.south) -- ([xshift=\arrowx]rot.north);

% ------------------------- the side annotations ---------------------------
\node[note] at ($(fock.south east)!0.5!(rdm.north east)+(-45mm,0)$)
  {(i) $\rho_{\mu}$ free with 1RDM $B^{\gamma}_{\mu}$;\\
   free state entropy formula};

\node[note] at ($(rdm.south east)!0.5!(poly.north east)+(-45mm,0)$)
  {(ii) dependence on $\gamma$ \\
  only through block-spectrum};

\node[note] at ($(poly.south east)!0.5!(least.north east)+(-45mm,0)$)
  {(iii) strict Schur-concavity of $H$};

\node[note] at ($(least.south east)!0.5!(rot.north east)+(-45mm,0)$)
  {(iv) find corresponding unitary};

\end{tikzpicture}

\caption{Strategy for maximizing the orbital total correlation $C^{\bm{\lambda}}$
of a free state $\rho$ over 1P unitaries $U$. (i) Restricting to free states allows us to map the
problem from the $2^{d}$-dimensional Fock space to $d\times d$ 1RDMs $\gamma$ whose eigenvalues equal the NONs $\bn$ of $\rho$.
(ii) Invariance under orbital rotations within each lattice region further reduces the problem to the 1RDMs' block spectra $\bb$, which form a convex polytope
$\mathfrak{B}^{\bm{\lambda}}_{\bm{n}}$. (iii) The strictly Schur-concave objective function $H$ is maximized by the majorization-least element of
$\mathfrak{B}^{\bm{\lambda}}_{\bm{n}}$. (iv) This element corresponds to an optimal 1P unitary built with discrete Fourier transform-type unitaries.}
\label{fig:strategy}
\end{figure}

\section{Solutions of the general problem} \label{sec:solutions_general}

We will now solve several instances of the general problem of computing $C^{\blam}_{\max}(\rho)$ for free states $\rho$. An instance of the problem is specified by the \emph{dimension} of the 1P space $\H_1$, $d$, an \emph{integer partition} of $d$, $\blam = (\lambda_1, \ldots, \lambda_k)$, and the \emph{natural occupation numbers} of the free state $\rho$, $ 1 \geq n_1 \geq \ldots \geq n_d \geq 0$.

The main result of this section, and of the article, is the complete solution for the problem instances where $\blam = (\lambda_1, \lambda_2)$ is a \emph{bipartition} of $d$ (see Result~\ref{res:bipartitions} in Sec.~\ref{subsec:bipartitions}). This is a generalization of Result~\ref{prop:Emax_SD_bipartition}. Furthermore, we provide the solution for the trivial multipartitions $\blam = (1, \ldots, 1)$ (Result~\ref{res:finest}); we solve all instances with $d \leq 5$ (Result~\ref{res:small_d}); and we present a systematic simplification for instances with $\bn = (1, \ldots, 1, 0, \ldots, 0)$ (with arbitrary $d$ and $\blam$), i.e., for SDs. All results are summarized in Tables\til\ref{tab:tabledgen},\til\ref{tab:tabled2},\til\ref{tab:tabled3},\til\ref{tab:tabled4}, and\til\ref{tab:tabled5}; technical details are deferred to the appendices.

We conclude this section with two Conjectures (Conjs.~\ref{conj:1} and~\ref{conj:2}) on the form of a potential general solution, which we support with numerical investigations.

\subsection{1P space dimension $d = 2$} \label{subsec:2D_space}
To illustrate the core concepts, we first examine the simple case $d=2$, with its trivial bipartition $\blam = (1,1)$. Fix NONs $1 \geq n_1 \geq n_2 \geq 0$. The set $\bset^{\blam}_{\bn}$ is the set of all possible diagonals $\bb $ of $2 \times 2$ Hermitian matrices with spectrum $\bn$. It is characterized by the Schur-Horn Theorem\til\ref{thrm:schur-horn} (or by the Cauchy interlacing inequalities of Thm.\til\ref{thrm:interlacing}), as follows: $\bb = (b_1 \sbs b_2) \in \bset^{\blam}_{\bn}$ if and only if $\bn \succ \bb$. The polytope $\bset^{\blam}_{\bn}$ is thus the \emph{segment} in $\R^2$ with endpoints $(n_1, n_2)$ and $(n_2, n_1)$.
% We represent it in Fig.\til\ref{fig:bset_d=2}.
It contains the maximally mixed vector
\begin{equation}
	\boxed{\bb_{\max} = \left ( \frac{n_1 + n_2}{2} \sbs \frac{n_1 + n_2}{2} \right )} \;
\end{equation}
whose two equal components are the average of the entries of $\bn$. According to our discussion in Sec.~\ref{sec:math_struc_gen_prob}, since $\bb_{\max}$ is majorized by every other $\bb \in \bset^{\blam}_{\bn}$, it is indeed the unique maximizer of the objective function of Eq.~\eqref{eq:Imax_max_over_bset}. For any free state $\rho$ with NONs $n_1, n_2$,
\begin{equation}
	C^{\blam}_{\max}(\rho) = 2 h \left ( \frac{n_1 + n_2}{2}  \right ) - h(n_1) - h(n_2)
\end{equation}
follows from Prop.~\ref{prop:Imax_max_over_bset}.
%\begin{figure}
	
%	\input{figures/bset_d=2}
%	\caption{The polytope $\bset^{\blam}_{\bn}$ of block spectra for $\blam = (1,1)$, with the maximally mixed element $\bb_{\max}$.}
%	\label{fig:bset_d=2}
	
%\end{figure}

Next, we determine the two natural orbitals of an optimally rotated fermionic state with NONs $\bn$. Let us restate this question from the experimental perspective of Sec.~\ref{sec:main_problem}. Suppose we have prepared an initial fermionic free state $\rho$ over two spatially separated orbitals $\ket{1}$ and $\ket{2}$, with NONs $\bn$. Suppose that $\ket{1}$ and $\ket{2}$ are precisely the natural orbitals of $\rho$---this is the case if, e.g., the states of the two orbitals have been prepared independently. The 1RDM of $\rho$, relative to the basis $\{ \ket{1}, \ket{2} \}$, is then $\gamma= \text{diag}(n_1, n_2)$. Our question becomes: what 1P unitary (i.e., orbital transformation) should we apply to $\rho$ in order to maximize the correlation between the orbitals $\ket{1}$ and $\ket{2}$? The answer, up to phases, is a two-dimensional \emph{discrete Fourier transform} (or Hadamard transform) of the orbitals. That is, we should implement the 1P unitary $U$ defined by the following transformation of the fermionic creation operators $f^\dagger_{1/2}$ associated with the basis orbitals:
\begin{align} \nonumber
	f^\dagger_j \rightarrow U^\dagger f^\dagger_j U &= \sum_{i=1}^2 f_i^\dagger (\mathrm{DFT}_2)_{ij} \; ,\\
	\mathrm{DFT}_2 & = \frac{1}{\sqrt{2}} \begin{pmatrix} \label{eq:DFT2}
		1 & 1 \\
		1 & -1
	\end{pmatrix} \; .
\end{align}
If the system state $\rho$ transforms accordingly, i.e., $\rho \rightarrow U \rho U^\dagger = \rho'$, then the 1RDM of the output free state $\rho'$ reads
\begin{equation} %YES, THE (here unrelevant) DAGGER IS IN THE RIGHT PLACE...
	\gamma' = \mathrm{DFT}_2^{\dagger} \cdot \gamma \cdot \mathrm{DFT}_2 = 
	\frac{1}{2}
	\begin{pmatrix}
		n_1 + n_2 & n_1 - n_2 \\
		n_1 - n_2 & n_1 + n_2
	\end{pmatrix} \; .
\end{equation}
The diagonal of $\gamma'$ is maximally mixed, i.e., $\text{diag}(\gamma') = \bb_{\max}$, so that the correlation between $\ket{1}$ and $\ket{2}$ contained in $\rho'$ indeed reaches its maximal attainable value. An alternative to $\mathrm{DFT}_2$ is given by
\begin{equation}
	\text{e}^{\text{i} \pi \sigma_x/4} = \frac{1}{\sqrt{2}}
	\begin{pmatrix}
		1 & \text{i} \\
		\text{i} & 1
	\end{pmatrix}
\end{equation}
or by any other unitary $u$ such that $|u_{ij}| = 1/\sqrt{2}$. Under the transformations described by these unitaries, the natural orbitals of the output free state $\rho'$ are \bt unbiased' superpositions of $\ket{1}$ and $\ket{2}$, which indeed ensures maximal correlation. This idea underpins the analysis for all the cases $d, \blam$ that we will examine.

Crucially, the mode-mixing transformations described above can be realized experimentally. In particular, the last matrix describes the unitary evolution induced by the tunneling Hamiltonian $t(f^\dagger_1 f_2 + f^\dagger_2 f_1)$ for a specific time. This evolution can be implemented for, e.g., cold fermionic atoms in optical lattices (with two sites only)\til\cite{bloch_review}, double quantum dots\til\cite{1998divincenzo_quantum_dots,kim_double_qdot} or, for a bosonic analogue, 50:50 beam splitters\til\cite{  zeilinger_unitary_breakdown}.
	
Our analysis of free states with $d=2$ and $\blam = (1,1)$ can be extended to \emph{arbitrary} two-orbital fermionic states that commute with the number parity operator. In particular, the orbital transformations described above maximize the correlation for any initial parity-superselected state with natural orbitals $\ket{1}$ and $\ket{2}$. This extension beyond free states is a general feature of all trivial partitions of the form $\blam = (1, \ldots, 1)$ (see Sec.~\ref{subsec:trivial_partitions}), and is discussed in Appendix\til\ref{app:finest_partition_arbitrary_states}.

\subsection{Bipartitions} \label{subsec:bipartitions}

We now solve the problem of maximal orbital correlation of free states for arbitrary bipartitions $\blam$ of any space dimension $d$, and arbitrary NONs $\bn$. This solution generalizes Result~\ref{prop:Emax_SD_bipartition}, and constitutes the main result of this work. We obtain it by generalizing the analysis for $d=2$ and $\blam=(1,1)$ of Sec.~\ref{subsec:2D_space}. As an intermediate step, the illuminating case $d=3$, $\blam = (2,1)$ is examined in Appendix~\ref{app:3D_space}.

Our solution is stated in the following result, which we prove below.
\begin{result}[Free states, arbitrary bipartitions] \label{res:bipartitions}
Consider a lattice of $d \geq 2$ sites split into two regions of $\lambda_1 \geq \lambda_2$ sites, i.e., a bipartition $\blam = (\lambda_1, \lambda_2)$ of $d$, and let $1 \geq n_1 \geq \ldots \geq n_d \geq 0$. Then the maximizer of $H$ [Eq.~\eqref{eq:function_Hadamard}] over $\bset^{\blam}_{\bn}$ [Def.~\ref{def:b_set}] is
%\begin{figure}[H]
\begin{mdframed}
\begin{align} \nonumber
&\bb_{\max} = \\ \nonumber
& \Bigg(
\frac{n_1 + n_d}{2},\ \frac{n_2 + n_{d-1}}{2},\ \ldots,\ \frac{n_{\lambda_2} + n_{\lambda_1 + 1}}{2}, n_{\lambda_2+1}, \, \ldots, \\ \label{eq:bbest_bip}
& n_{\lambda_1}\ \Big |  \ \frac{n_1 + n_d}{2},\ \frac{n_2 + n_{d-1}}{2},\ \ldots,\ \frac{n_{\lambda_2} + n_{\lambda_1 + 1}}{2}
\Bigg)
\end{align}
\end{mdframed}
%\end{figure}
up to a non-increasing reordering of the entries within each block. This vector is majorized by every other element of $\bset^{\blam}_{\bn}$. Consequently, the maximal orbital total correlation [Eq.~\eqref{eq:Iblam_max_def}] between the two regions of any fermionic free state $\rho$ with NONs $\bn$ is $C^{\blam}_{\max}(\rho) = H(\bb_{\max}) - H(\bn)$. It is attained by applying $\lambda_2$ two-orbital transformations $\mathrm{DFT}_2$ [Eq.~\eqref{eq:DFT2}], the $j$th one mixing the orbitals $\ket{j}$ and $\ket{d+1-j}$, to the free state with 1RDM $\text{diag}(n_1, \ldots, n_d)$ in the lattice-site basis.
\end{result}

First, the relation $\bb_{\max} \in \bset^{\blam}_{\bn}$ follows from the existence of the following Hermitian matrix $\gamma$ with spectrum $\bn$ and block spectrum $\bb_{\max}$:
\begin{widetext}
\begin{equation} \label{eq:big_kaustavs_matrix}
	\gamma =
	\begin{pmatrix}
	\begin{array}{cccccc|ccc}
		\frac{n_1 + n_d}{2} &   &   &   &   &   &   &   & \frac{n_1 - n_d}{2} \\
			& \ddots &   &   &   &   &   &   &   \\
			&   & \frac{n_{\lambda_2} + n_{\lambda_1 + 1}}{2} &   &   &   & \frac{n_{\lambda_2} - n_{\lambda_1 + 1}}{2} &   &   \\
			&   &   & n_{\lambda_2 + 1} &   &   &   &   &   \\
			&   &   &   & \ddots &   &   &   &   \\
			&   &   &   &   & n_{\lambda_1} &   &   &   \\
		\hline
			&   & \frac{n_{\lambda_2} - n_{\lambda_1 + 1}}{2} &   &   &   & \frac{n_{\lambda_2} + n_{\lambda_1 + 1}}{2}  &   &   \\
			&   &   &   &   &   &   & \ddots &   \\
		\frac{n_1 - n_d}{2} &   &   &   &   &   &   &   & \frac{n_1 + n_d}{2}
	\end{array} 
	\end{pmatrix} \; .
\end{equation}
\end{widetext}
Empty spaces indicate zero entries. $\gamma$ can be obtained from diag$(n_1, \ldots, n_d)$ by successive conjugations with unitary matrices $u_1, \ldots, u_{\lambda_2}$, where $u_j$ acts as $\mathrm{DFT}_2$ on $\text{span}(\be_j, \be_{d+1-j})$ and leaves the remaining canonical basis vectors unchanged---here, $\be_j$ is the $j$th element of the canonical basis of $\C^d$. Correspondingly, one can transform an initial free state with 1RDM diag$(n_1, \ldots, n_d)$ (relative to the lattice-site basis) to the free state with 1RDM $\gamma$ via a simple sequence of $\lambda_2$ elementary two-orbital transformations (see Sec.~\ref{subsec:2D_space}).

Second, we prove that $\bb_{\max}$ indeed maximizes $H$ over $\bset^{\blam}_{\bn}$. To that end, it suffices to show that $\bb_{\max} \prec \bb$ $\forall \, \bb \in \bset^{\blam}_{\bn}$. For clarity, we focus on \emph{balanced} bipartitions, $\lambda_1 = \lambda_2 = k$, deferring the general case to Appendix~\ref{app:proof_bipartition_lemma}. We begin with a theorem of independent interest.
\begin{thrm} \label{thrm:mikes_maj_rel}
	Let $A$ and $B$ be $k \times k$ Hermitian matrices and
	\begin{equation} \label{eq:ABM_matrix}
	\gamma =
	\begin{pmatrix}
	\begin{array}{c|c}
		A & M \\
		\hline
		M^{\dagger} & B
	\end{array} 
	\end{pmatrix} \; .
	\end{equation}
	Let $\balpha, \bbeta$ and $\boldsymbol{n}$ be the decreasingly ordered spectra of $A, B$ and $\gamma$, respectively. Let $\tilde{{\boldsymbol{n}}} = (n_1 + n_{2k}, n_2 + n_{2k-1}, \ldots, n_{k} + n_{k+1})$. Then
	\begin{equation} \label{eq:mikes_maj_rel}
		\tilde{\boldsymbol{n}} \prec \balpha + \bbeta \; .
	\end{equation}
\end{thrm}
In Appendix~\ref{app:proof_bipartition_lemma}, we prove both Thm.~\ref{thrm:mikes_maj_rel} directly, and a remarkable generalization of it, namely Thm.~\ref{thm:mike_main}. To the best of our knowledge, neither theorem is present in the standard literature, although Ref.~\cite{Thompson_Therianos_1972} presents a related result.
%We remark that the inequalities underlying the majorization relation of Eq.~\eqref{eq:mikes_maj_rel} are tight. For instance, there are nontrivial examples of matrices $\gamma, A, B$ s.t. $\tilde{\bn} = \balpha + \bbeta$, such as the matrix of Eq.~\eqref{eq:big_kaustavs_matrix}, if we set $\lambda_1 = \lambda_2 = k$ and assume that $n_1 + n_{2k} \geq n_2 + n_{2k - 1} \geq \ldots \geq n_k + n_{k+1}$.

Under the assumption $\lambda_1 = \lambda_2 = k$, Thm.~\ref{thrm:mikes_maj_rel} directly implies that $\bb_{\max} \prec \bb$ $\forall \, \bb \in \bset^{\blam}_{\bn}$. To verify this, consider a $2k \times 2k $ Hermitian matrix $\gamma$ with spectrum $\bn$, partitioned as in Eq.~\eqref{eq:ABM_matrix}, and let $\balpha, \bbeta$ and $\tilde{\bn}$ be defined as in Thm.~\ref{thrm:mikes_maj_rel}. Then
\begin{equation} \label{eq:maj_eq_even_bip}
	\bb_{\max} \ = \left (\frac{\tilde{\bn}}{2}, \frac{\tilde{\bn}}{2} \right) \prec \left ( \frac{\balpha + \bbeta}{2}, \frac{\balpha + \bbeta}{2} \right ) \prec (\balpha, \bbeta) = \bb \; .
\end{equation}
This implies that $\bb_{\max}$ is the maximizer of $H$ over $\bset^{\blam}_{\bn}$. For any free state $\rho$ with NONs $\bn$, Prop.~\ref{prop:Imax_max_over_bset} then gives
\begin{equation}
	C^{\blam}_{\max}(\rho) = H(\bb_{\max}) - H(\bn) \; .
\end{equation}
This proves Result~\ref{res:bipartitions} for balanced bipartitions; the unbalanced case is proven in Appendix~\ref{app:proof_bipartition_lemma}.

Our analysis also yields the following upper bound on the bipartite orbital entanglement entropy of an arbitrary parity-superselected~\cite{Friis_2016} fermionic pure state, not necessarily free and hence including ground states of interacting systems, in terms of its natural occupation numbers.
\begin{result}[Entanglement bound for arbitrary states] \label{res:entanglement_bound}
    Let $\ket{\Psi}$ be a parity-superselected fermionic pure state over a bipartite lattice of $\lambda_1 + \lambda_2 = d$ sites. Let $n_1 \geq \dots \geq n_d$ be the NONs of $\ket{\Psi}$. Then the orbital entanglement entropy of $\ket{\Psi}$ between the two regions of $\lambda_1$ and $\lambda_2$ sites is bounded as
    \begin{equation} \label{eq:arbitrary_state_ent_bound}
        E(\ket{\Psi}) \leq \frac{1}{2}H(\bb_{\max}) \; ,
    \end{equation}
    where $\bb_{\max}$ is given by Eq.~\eqref{eq:bbest_bip} and $H$ by Eq.~\eqref{eq:function_Hadamard}. In particular, for $\lambda_1 = \lambda_2 = d/2$, this bound reads
    \begin{equation}
        E(\ket{\Psi}) \leq \sum_{j=1}^{d/2} h \left ( \frac{n_j + n_{d+1-j}}{2} \right )
    \end{equation}
    and is stronger than the standard bound $E(\ket{\Psi}) \leq (d/2) \log 2$.
\end{result}
\begin{proof}
    Let $\gamma$ be the 1RDM of $\ket{\Psi}$ and $\rho_{1/2}$ its orbital-reduced 
	states; the 1RDM of $\rho_\mu$ is the diagonal 
	block $B^\gamma_\mu$. Then
	\begin{align} \nonumber
		E(\ket{\Psi}) &\overset{(a)}{=} \; \frac{1}{2} \left[ S(\rho_1) + S(\rho_2) \right] \\ \nonumber
		&\overset{(b)}{\leq} \; \frac{1}{2} \sum_{\mu = 1}^{2} 
			H \! \left( \spec^\downarrow ( B^\gamma_\mu ) \right) \\ \nonumber
		&= \; \frac{1}{2} H(\bb) \, , \quad 
			\bb := \left( \spec^\downarrow ( B^\gamma_1 ) \, \middle| \, 
			\spec^\downarrow ( B^\gamma_2 ) \right) \in \bset^{\blam}_{\boldsymbol{n}} \\
		&\overset{(c)}{\leq} \; \frac{1}{2} H(\bb_{\max}) \; .
	\end{align}
	In $(a)$ we used that $\rho_1$ and $\rho_2$ are isospectral~\cite{Friis_2016}; 
	in $(b)$, that the entropy of a fermionic state is at most that of the free 
	state with the same 1RDM~\cite{gottliebfermions}, i.e.\ $H$ evaluated at its NONs; and $(c)$ is $\bb_{\max} \prec \bb$ 
	(see Result~\ref{res:bipartitions}) with the Schur-concavity of $H$.
\end{proof}

\subsection{Trivial partitions \texorpdfstring{$\boldsymbol{\lambda}$}{lambda} $ = (1,\ldots,1)$} \label{subsec:trivial_partitions}

We now solve our optimization problem for an arbitrary 1P space dimension $d$ that is partitioned trivially, i.e., $\blam = (1, \ldots, 1)$. We build on the case of $d=2$, solved in Sec.~\ref{subsec:2D_space}. The solution for the intermediate case $d=3$, $\blam = (1,1,1)$ is presented in Appendix~\ref{app:3D_space}.

Our solution is stated in the following result.
\begin{result}[Free states, finest partition] \label{res:finest}
Let $\blam = (1, \ldots, 1)$, i.e., each of the $d$ lattice sites forms its own region, and, for given $n_1 \geq \ldots \geq n_d$, let $\eta = (\sum_{j=1}^d n_j)/d$ be the average of all $\bn$ entries. Then the maximally mixed vector
\begin{equation} \label{eq:bbest_trivial_parts}
	\boxed{\bb_{\max} = (\eta | \eta | \ldots | \eta)}
\end{equation}
is in $\bset^{\blam}_{\bn}$, is majorized by every other $\bb \in \bset^{\blam}_{\bn}$, and is thus the maximizer of $H$ over that set. Consequently, the maximal orbital total correlation [Eq.~\eqref{eq:Iblam_max_def}] among the $d$ sites of any fermionic free state $\rho$ with NONs $\bn$ is $C^{\blam}_{\max}(\rho) = d \, h(\eta) - H(\bn)$; it is attained by a free state whose natural orbitals are the discrete plane waves $\propto \sum_{j=1}^d \e^{\i 2\pi k j / d} \ket{j}$, $k = 0, \ldots, d-1$.
\end{result}
For completeness, note that $\bset^{\blam}_{\bn}$ is the set of all possible diagonals of Hermitian matrices with spectrum $\bn$, and is characterized by the Schur-Horn Theorem~\ref{thrm:schur-horn}.

An example of an optimal matrix $\gamma$ with spectrum $\bn$ and diagonal $\bb_{\max}$ is given by
\begin{equation}
	\gamma = \mathrm{DFT}^\dagger_d \text{diag}(\bn) \mathrm{DFT}_d \; ,
\end{equation}
where $\mathrm{DFT}_d$ is the $d$-dimensional discrete Fourier transform matrix. We can interpret this result through the lens of our initial physical problem. Consider a 1D chain of lattice sites partitioned trivially, i.e., each lattice site represents a subsystem on its own (as in Ref.~\cite{Eisert2008_Gaussian_q_marg_prob} on bosonic Gaussian states). Then, for fixed NONs $n_j$, a maximally orbital-correlated free state is one whose natural orbitals are discrete plane waves $\propto \sum_{j=1}^d \e^{\i 2\pi k j / d} \ket{j}$. Here, $k = 0, \ldots, d-1$ and $\ket{j}$ represents a lattice site orbital, as before.

The maximal orbital correlation of any free state $\rho$ with NONs $n_j$ follows from Prop.~\ref{prop:Imax_max_over_bset} as 
\begin{equation} \label{eq:I_max_finest_partition}
	C^{\blam}_{\max}(\rho) = d \, h(\eta) - H(\bn) \; .
\end{equation}

As announced in Sec.\til\ref{subsec:2D_space}, our analysis for trivial partitions can be extended to \emph{arbitrary} states that commute with the number parity operator. We present this extension in Appendix~\ref{app:finest_partition_arbitrary_states}.

\subsection{General multipartitions}

We will now discuss the solutions to a few further simple instances of our optimization problem, and put forward two conjectures regarding the potential form of a general solution. Our solution for all small lattices is stated in the following result.
\begin{result}[Free states, all partitions with $d \leq 5$] \label{res:small_d}
For every splitting of a lattice of $d \leq 5$ sites into at least two regions of $\lambda_1 \geq \lambda_2 \geq \ldots$ sites, i.e., every integer partition $\blam$ of $d$ with at least two parts, and all NONs $1 \geq n_1 \geq \ldots \geq n_d \geq 0$, the maximizer $\bb^{\blam}_{\max}(\bn)$ of $H$ over $\bset^{\blam}_{\bn}$ is majorized by every other element of $\bset^{\blam}_{\bn}$, and its entries are averages of subsets of the NONs. Its closed form is given in Tables~\ref{tab:tabledgen}--\ref{tab:tabled5}; in general, its functional form depends on $\bn$, i.e., it differs between regimes. Consequently, the maximal orbital total correlation [Eq.~\eqref{eq:Iblam_max_def}] among these regions is $C^{\blam}_{\max}(\rho) = H(\bb^{\blam}_{\max}(\bn)) - H(\bn)$ for any fermionic free state $\rho$ with NONs $\bn$.
\end{result}

The first nontrivial multipartition occurs at $d=4$, namely $\blam = (2,1,1)$. The derivation of $\bb_{\max}$ for this case is presented in full detail in Appendix~\ref{app:d=4}. For given $\bn$, we find that
\begin{widetext}
\begin{equation} \label{eq:bbest_2_1_1}
	\boxed{
	\bb_{\max} = 
	\begin{cases} 
		\left( \frac{n_1+n_3+n_4}{3}, n_2 \sbs \frac{n_1+n_3+n_4}{3} \sbs \frac{n_1+n_3+n_4}{3} \right) & \text{if } n_1 \geq \eta > n_2 \\[10pt]
		(\eta, \eta \sbs \eta \sbs \eta) & \text{if } n_2 \geq \eta \geq n_3 \\[10pt]
		\left( n_3, \frac{n_1+n_2+n_4}{3} \sbs \frac{n_1+n_2+n_4}{3} \sbs \frac{n_1+n_2+n_4}{3} \right) & \text{if } n_3 > \eta \geq n_4
	\end{cases}
	}
\end{equation}
\end{widetext}
is in $\bset^{\blam}_{\bn}$, is majorized by every other $\bb \in \bset^{\blam}_{\bn}$, and is thus the maximizer of $H$ over that set. This is the first and simplest scenario where the functional form of $\bb_{\max}$ depends on $\bn$. In other words, the function $\bn \mapsto \bb_{\max}^{\blam}(\bn) = \bb_{\max}$ is piecewise-defined. We refer to the different---here, three---pieces as regimes, discussed in Appendix~\ref{app:regimes}.

The analysis continues in Appendix~\ref{app:d=5} with the full solution for $d=5$ and all of its nontrivial multipartitions, namely $\blam = (3,1,1), (2,2,1)$ and $(2,1,1,1)$. The corresponding formulas for $\bb_{\max}$ are summarized in Table~\ref{tab:tabled5}.

Furthermore, in Appendix~\ref{app:SDs} we show that restricting to SDs, i.e., constraining the free state NONs $n_j$ to be in $\{0,1\}$, enables a systematic simplification of our optimization problem.

Interestingly, our derivations vary substantially across different space dimensions $d$ and partitions $\blam$, highlighting the intricate structure of the problem. Nevertheless, in all cases (i) the optimality of $\bb_{\max}$ stems from it being majorized by every other vector $\bb \in \bset^{\blam}_{\bn}$, and (ii) the entries of $\bb_{\max}$ are averages of a subset of entries of $\bn$. In Conjs.~\ref{conj:1} and~\ref{conj:2}, we propose that these features characterize the general solution. We then provide supporting arguments based on numerical results.

\begin{conj} \label{conj:1}
	For any admissible choice of $d, \blam, \bn$, $\bb_{\max}^{\blam}(\bn)$ is least in $\bset^{\blam}_{\bn}$ relative to majorization. 
\end{conj}
For more on majorization-least vectors, we refer to the book~\cite{MarshallBook}, to the more recent Ref.~\cite{uiterkamp2023leastmajorized}, and to the references therein.

Before putting forward the second conjecture, we introduce some notation. We denote by $|X|$ the cardinality of the set $X$ or the length of a tuple $X$. Also, a (finite) ordered partition $\bP$ of a set $X$ is a tuple $(P_{\alpha})_{\alpha = 1}^{|\bP|}$ such that $P_\alpha \neq \varnothing \ \forall \alpha, \quad \bigcup_\alpha P_\alpha = X$ and $P_\alpha \cap P_{\alpha'} = \varnothing \ \text{for } \alpha \neq \alpha'$.
\begin{defn}
Let $\bP = (P_\alpha)_{\alpha = 1}^{|\bP|}$ be an ordered partition of the set $\{1, \ldots, d\}$. We define the linear map $\textup{avg}_{\bP}: \R^d \rightarrow \R^d$ via
\begin{equation}
	\textup{avg}_{\bP} (\bn) := (\ba_{P_1}(\bn), \ldots, \ba_{P_{|\bP|}}(\bn)) \; , \; \; \bn \in \mathbb{R}^d,
\end{equation}
where
\begin{equation}
	\ba_{P_\alpha}(\bn) = \frac{\sum_{j \in P_\alpha} n_j}{|P_\alpha|}
(
\underbrace{
	1,
	\ldots,1
}_{|P_\alpha| \textup{ times}}
) \; .
\end{equation}
\end{defn}
For example, for $d = 4$, $\bP = ((1,3,4),(2))$ and $\bn \in \R^4$, the above definition reads
\begin{align} \nonumber
	&\textup{avg}_{\bP} (\bn) = \\
	& \;\;\; \left ( \frac{n_1 + n_3 + n_4}{3},\frac{n_1 + n_3 + n_4}{3},\frac{n_1 + n_3 + n_4}{3}, n_2  \right ) \; .
\end{align}
\begin{conj} \label{conj:2}
	For any admissible choice of $d, \blam, \bn$, there exists an ordered partition $\bP$ of $\{ 1, \ldots, d \}$ such that
	\begin{equation}
		\bb_{\max}^{\blam}(\bn) = \textup{avg}_{\bP} (\bn) \;\;\; \textup{   up to a permutation.}
	\end{equation}
\end{conj}
If this conjecture holds true, the partition $\bP$ would depend not only on $\blam$, but also on $\bn$---this is the concept of regimes discussed in Appendix~\ref{app:regimes}. 

In more detail, we expect that one can always find simple procedures, based on generalized discrete Fourier transform matrices, to build optimal unitaries---i.e., unitaries that conjugate diag$(\bn)$ into a matrix with block spectrum $\bb_{\max}^{\blam}(\bn)$. Such a simple structure of the unitaries would be advantageous for experimental applications of our results to maximal entanglement extraction, especially when the optimal unitary can be decomposed as a sequence of two-orbital unitaries~\cite{bloch_review,1998divincenzo_quantum_dots,kim_double_qdot,zeilinger_unitary_breakdown}.

%Suppose Conjecture~\ref{conj:2} holds true. Then, since the number of ordered partitions of $\{ 1, \ldots, d \}$ is finite for any integer $d$, the number of regimes (see Sec.~\ref{app:regimes}) would also be finite for any given $d$ and $\blam$.

We wrote a program that takes as input $d, \blam$ and $\bn$, and returns as output (a numerical approximation of) $\bb_{\max}^{\blam}(\bn)$. The program uses a gradient ascent method that converges fast thanks to the concavity of $H$. For every integer partition $\blam$ with $d = \sum_\mu \lambda_\mu \leq 8$, we ran the program for $\sim$ 1000 different vectors $\bn$. For each instance, we then sampled vectors from $ \bset^{\blam}_{\bn}$ and verified that $\bb_{\max}^{\blam}(\bn)$ was majorized by all of them, thus supporting Conj.~\ref{conj:1}. We also checked that the entries of $\bb_{\max}^{\blam}(\bn)$ are always averages of the entries of $\bn$ within a certain convergence threshold, thus supporting Conj.~\ref{conj:2}. We extended our analysis to $d=9, \blam = (3,3,3)$; $d=10, \blam = (2,2,2,2,2)$; $d=12, \blam = (3,3,3,3)$ and $\blam = (4,4,4)$. All results align with our conjectures.

As an example of the complexity of the problem, the integer $d=8$ has 16 different nontrivial multipartitions $\blam$, and for $\blam = (4,3,1)$ alone we found 31 different regimes numerically.

Setting aside the optimization itself, providing a general and efficient characterization of the polytopes $\bset^{\blam}_{\bn}$ would be a significant step towards finding a general solution to our problem. Ref.~\cite{lipoon} provides such a characterization in terms of linear inequalities, but these are highly redundant. One possible approach would be to generalize the methods in Appendix~\ref{app:proof_bipartition_lemma}.

\section{Extensions} \label{sec:extensions}

We discuss two extensions of the problem of maximal orbital entanglement or total correlation of free states. In Sec.\til\ref{subsec:extensions_SSR}, we analyze the correlation \emph{reduction} due to the local particle-number SSR and the necessary adaptations of our problem. In Sec.\til\ref{subsec:gaussian}, we then obtain the maximal orbital entanglement of a general pure fermionic \emph{Gaussian} state relative to a bipartition of the 1P space.

\subsection{Enforcing the local particle-number superselection rule} \label{subsec:extensions_SSR}
The problem we dealt with so far is not optimally formulated from a physical point of view, as we did not take into account the local particle-number SSR. Physically viable operations on a given set of fermionic orbitals are described by operators that commute with the particle-number operator relative to those orbitals\til\cite{ent_with_ssr_schuch,wisemanvaccaro,wiseman_vaccaro_on_general_SSRs_and_entanglement, wiseman_vaccaro_ferreting, ent_ferm_sys_banuls,Friis_2016}. This operational constraint results in a reduction of the orbital entanglement and total correlation that can be extracted from a fermionic system. Specifically, the entanglement entropy $E$ and the total correlation $C$ do not take the local particle-number SSR into account, and thus overestimate the amount of operationally accessible correlations.

We focus on total correlation, $C$; we recall that $E = C/2$ for pure states. We introduce two versions of an SSR-adapted total correlation, $C^{SSR}$. For the sake of clarity, here we restrict to a bipartite set of fermionic orbitals, $A \cup B$; we refer to Appendix\til\ref{app:SSR} for a thorough analysis, including orbital multipartitions. Denote by $\Pi_{A/B}^{N_{A/B}}$ the projector onto the subspace of the Fock space of states with $N_{A/B}$ particles in the orbitals $A/B$. For a given state $\rho$ on the Fock space, $p_{N_A, N_B} = \tr[\rho \Pi_A^{N_A} \Pi_B^{N_B}]$ is the probability of measuring $N_A$ particles in the orbitals $A$ and $N_B$ in $B$. For $p_{N_A, N_B} \neq 0$, the corresponding post-measurement state is $\rho_{N_A, N_B} = \Pi_A^{N_A} \Pi_B^{N_B}\rho \Pi_A^{N_A} \Pi_B^{N_B} / p_{N_A, N_B}$. The first modified total correlation measure is given by the standard total correlation of the \bt projected' state $\rho^{SSR} = \sum_{N_A, N_B} p_{N_A, N_B} \rho_{N_A, N_B}$~\cite{wiseman_vaccaro_ferreting, gesa_lexin_2024}:
\begin{equation}
	C^{SSR}_{proj} (\rho) : = C (\rho^{SSR})  \; .
\end{equation}
The second SSR-adapted measure is the averaged total correlation of the \bt blocks' $\rho_{N_A, N_B}$~\cite{wisemanvaccaro}:
\begin{equation}
	C^{SSR}_{block} (\rho) : = \sum_{N_A, N_B} p_{N_A, N_B} C (\rho_{N_A, N_B}) \; .
\end{equation}

In Appendix\til\ref{app:SSR}, we motivate these definitions, address the \emph{maximization} of these quantities over free states with fixed NONs, and discuss similarities and differences with the SSR-agnostic problem of Sec.\til\ref{sec:generalizing}. We also solve some simple cases of the maximization problem: there, the optimal unitary rotation obtained by ignoring the SSR (see Sec.\til\ref{sec:generalizing}) is still effective, if not optimal, for extracting SSR-reduced correlations. Furthermore, in these example cases the loss of correlations due to the SSR is negligible in the limit of large number of particles and of orbitals (see Sec.~\ref{subsec:SSR_asymptotics} and Refs.~\cite{wisemanvaccaro,botero_normal_form_ferm_Gaussian_states}). Consequently, \emph{the solution to the SSR-agnostic problem may serve as a simpler yet reliable proxy for the SSR-constrained solution}.

Depending on the context, the less restrictive \emph{local-parity} SSR can be considered. The analysis of Appendix\til\ref{app:SSR} can be readily adapted to the local-parity SSR case. Furthermore, the two SSRs imply exactly the same modifications to our problem for small 1P space dimension $d \leq 3$.

\subsection{General pure fermionic Gaussian states} \label{subsec:gaussian}

We now extend the optimization of entanglement over SDs (see Sec.\til\ref{sec:main_problem}) to general pure fermionic Gaussian states\til\cite{2021Hackl_Gaussian_states} which need not have a definite number of particles. An analogous entanglement optimization problem for bosonic Gaussian states is addressed in Ref.\til\cite{ciancio_zanardi}, where entanglement is quantified using logarithmic negativity.

For the sake of clarity, we restrict to balanced bipartitions, and ignore the reduction of the physically available entanglement due to the relevant local-parity SSR\til\cite{ent_ferm_sys_banuls}. Consider $d = 2k$ orbitals of spinless fermions, evenly shared between two parties. Let $f^\dagger_{A(B),1}, \ldots, f^\dagger_{A(B),k}$ be the creation operators for the orbital subset $A$ ($B$). For any definite-parity $\ket{\Psi}$, we define $E(\ket{\Psi})$ as in Eq.\til\eqref{eq:def_entanglement}. Gaussian states are eigenstates of parity\til\cite{2021Hackl_Gaussian_states}, so their entanglement is well-defined\til\cite{Friis_2016}. We define and calculate
\begin{equation} \label{eq:Emax_gauss}
	E_{\max}^G := \max_{|\Psi \rangle \; \text{Gaussian}} E(\ket{\Psi}) \; .
\end{equation}
We can rewrite this quantity as
\begin{equation}
	E_{\max}^G = \max_{U \; \text{B.T.}} E(U \ket{0}) \; ,
\end{equation}
where the maximization is over general Bogoliubov transformations $U$\til\cite{fetter_walecka, Robinson65-metric,BLS94,Bach22}, and $\ket{0}$ is the particle vacuum or any other reference Gaussian state vector.

For SDs, the maximally entangled state was obtained by pairing each orbital from the first subset with one orbital from the second subset (see Fig.\til\ref{fig:max_SD_construction}). This suggests the following BCS-like guess\til\cite{fetter_walecka} for a Gaussian state vector that is a maximizer for the right-hand side of Eq.\til\eqref{eq:Emax_gauss}:
\begin{equation} \label{eq:Gaussian_states_normal_form}
	\ket{\Psi_{\btheta}} = \prod_{j=1}^k (\cos \theta_j + \sin \theta_j f_{A,j}^\dagger f_{B,j}^\dagger) \ket{0} \; .
\end{equation}
As shown in Ref.\til\cite{botero_normal_form_ferm_Gaussian_states}, \emph{any} Gaussian state vector can be brought to the above normal form via two local Bogoliubov unitaries acting separately on $A$ and $B$. See Appendix\til\ref{app:normal_form} for the analogous result for SDs. Since local unitaries preserve entanglement, we may restrict the optimization to state vectors of the form\til\eqref{eq:Gaussian_states_normal_form}. The entanglement of such states is easily computed, using the orbital-pairing structure:
\begin{equation}
	E(\ket{\Psi_{\btheta}}) = \sum_{j = 1}^k h(\cos^2 \theta_j) \; .
\end{equation}
It is maximal for $\cos^2 \theta_j = 1/2 \; \forall \, j$, which leads to $k$ ebits of entanglement. The following proposition, a straightforward consequence of Ref.\til\cite{botero_normal_form_ferm_Gaussian_states}, summarizes our discussion.
\begin{prop}[cf.~Ref.~\cite{botero_normal_form_ferm_Gaussian_states}] \label{prop:gaus}
	The fermionic Gaussian state vector
	\begin{equation}
	\ket{\Psi_{\btheta = (\frac{\pi}{4}, \ldots, \frac{\pi}{4})}} = 2^{-k/2} \prod_{j=1}^k (1 + f_{A,j}^\dagger f_{B,j}^\dagger) \ket{0}
\end{equation}
	is a maximizer for the right-hand side of Eq.\til\eqref{eq:Emax_gauss}. The maximal entanglement of pure fermionic Gaussian states follows as
	\begin{equation}
		E_{\max}^G = k \log 2 \; .
	\end{equation}
\end{prop}

% One can further generalize our analysis by considering arbitrary \emph{mixed} fermionic Gaussian states and orbital \emph{multipartitions}, as was done in Sec.\til\ref{sec:generalizing} for free states (i.e., number-conserving Gaussian states). For a systematic approach, one describes Gaussian states with their $2d \times 2d$ covariance matrix\til\cite{botero_normal_form_ferm_Gaussian_states}. This generalizes the $d\times d$ 1RDM $\gamma$ to non-nonumber-conserving states by including information on the expectation values of the operators $ f_j f_k $ and $f^\dagger_j f^\dagger_k$. The analysis then follows and generalizes Sec.\til\ref{sec:math_struc_gen_prob}.

% To avoid overestimating the entanglement that can physically be extracted from fermionic Gaussian states, one needs to consider the \emph{local-parity SSR}\til\cite{ent_ferm_sys_banuls}.

\section{Summary and outlook} \label{sec:conclusions}

Orbital rotations are commonly used to minimize orbital correlation and thereby reduce the complexity of many-body states. In this work, we addressed the opposite task, motivated by quantum information processing: maximizing orbital entanglement and correlation over all orbital rotations for non-interacting fermions. The resulting optimized state then serves as a resource: for example, its entanglement can be transferred to qubit registers~\cite{wiseman_vaccaro_ferreting, gesa_lexin_2024} and used for quantum information processing~\cite{horodecki_entanglement,acin2018roadmap}. We solved this maximization problem in several key settings, with analytical closed-form results throughout.

As a first key result, we proved that the maximal orbital entanglement of Slater determinants of $N$ fermions on a lattice of $d_A + d_B = d$ sites is $\min\{d_A, d_B, N, d-N\}$ ebits (Result~\ref{prop:Emax_SD_bipartition}). It is attained by simple beam-splitter-type rotations of orbital pairs and, given enough particles and holes, equals the maximum over \emph{all} fermionic states. We then determined the maximal orbital total correlation of arbitrary number-conserving fermionic Gaussian states (free states) for every bipartition and arbitrary natural occupation numbers (Result~\ref{res:bipartitions}), for the finest partition in any dimension (Result~\ref{res:finest}), and for every partition of up to five sites (Result~\ref{res:small_d}). In all these cases, the optimal block spectrum of the corresponding 1RDM-level maximization problem is majorization-least.
%so the same optimal states maximize every Schur-concave function of the block spectra, not only the total correlation.
Its entries are averages of natural occupation numbers. Numerics for all partitions with $d \leq 8$ and selected ones up to $d = 12$ support our conjectures that this structure holds in general (Conjs.~\ref{conj:1} and~\ref{conj:2}). The underlying majorization relations between the spectrum of a Hermitian matrix and the spectra of its diagonal blocks (Thms.~\ref{thrm:mikes_maj_rel} and~\ref{thm:mike_main}) are new results of independent interest in matrix analysis. Beyond free states, natural occupation numbers bound the orbital entanglement of any fermionic pure state, e.g., the ground state of an interacting system (Result~\ref{res:entanglement_bound}). Finally, in the cases we analyzed under the local particle-number SSR, the rotations that are optimal without SSRs remain effective, if not optimal. Table~\ref{tab:roadmap} on page two summarizes all results.

These results answer the question posed in the introduction: states without any \emph{particle} correlation can carry substantial \emph{orbital} entanglement and \emph{orbital} correlation, and for enough particles and holes even as much as any fermionic state. This has consequences for several fields. For quantum information processing, mean-field states become a readily available entanglement resource: they are ubiquitous, efficiently prepared, and their optimal rotations are elementary two-orbital operations, realizable with beam splitters in electron or photon optics, tunneling in optical lattices and quantum dots, or two-qubit gates on superconducting quantum processors. For quantum chemistry and many-body physics, our results provide the maximization counterpart of orbital optimization, completing the picture of Fig.~\ref{fig:yinyang}, together with entanglement bounds that follow directly from natural occupation numbers. On a conceptual level, they make entanglement relativity~\cite{zanardi, 2004zanardi_ent_relativity_observables} quantitative for identical particles and identify orbital entanglement, among free states, as a genuinely fermionic resource, since bosonic free states are separable.

Natural next steps include proving our conjectures for arbitrary multipartitions, using proper entanglement measures such as the negativity~\cite{vol_set_sep_states, eisert_bounds_on_negativity} beyond pure bipartite states, and extending our analysis to interacting systems, with free states as analytical references~\cite{bellomia2024quasilocal_ent, bellomia2026classical_correlations}.

Beyond these questions, our results point to concrete applications. Slater determinants are routinely prepared on quantum simulators and quantum processors~\cite{Arute2020HF, bloch_review}.
Combined with our optimal rotations, which mostly act on disjoint orbital pairs and can thus be applied in parallel, they offer a low-depth route to maximal orbital entanglement.
Finally, in quantum chemistry, the maximal orbital entanglement of a molecular mean-field state could serve as a quantitative descriptor of electronic structure and chemical bonding~\cite{DingKnecht2023bonding, ding2026bonding}. 

\begin{acknowledgments}
We thank Kaustav Chatterjee for discussions during the early stages of this work, Chih-Chun Wang for helpful input on the general result for lattice bipartitions, and Chi-Kwong Li for insightful discussions on spectral sets. We used AI (Claude Opus 5) to recheck the derivations, to help polish the language (primarily in the introduction and conclusion), for Fig.~\ref{fig:strategy}, and to conceive the bound of Eq.~\eqref{eq:arbitrary_state_ent_bound}. We acknowledge financial support from the German Research Foundation (Grant SCHI 1476/1-1), the Munich Center for Quantum Science and Technology, the ERC-2025-COG under grant agreement No.~101231616 beyondDFT (C.S.), and the International Max Planck Research School for Quantum Science and Technology (IMPRS-QST) (J.L.). This research is also part of the Munich Quantum Valley, which is supported by the Bavarian state government with funds from the Hightech Agenda Bayern Plus.
\end{acknowledgments}

\bibliography{refs}

\appendix

\section{Orbital total correlation and entanglement in fermionic states} \label{app:mutual_info_fermions}

In this Appendix we motivate in the fermionic context the definition of the orbital total correlation $C$. For a mixed state $\rho$ over an orbital space partitioned into $k$ orthogonal orbital subspaces, it is given by Eq.~\eqref{eq:total_corr}. We focus on bipartitions $A + B$ of the underlying reference orbital basis, and only consider parity-superselected states~\cite{wisemanvaccaro, ent_ferm_sys_banuls,Friis_2016}.

For bipartite systems, $C$ reduces to the quantum mutual information
\begin{equation} \label{eq:app_total_bip_corr}
	C(\rho) = S(\rho_A) + S(\rho_B) - S(\rho) \; ,
\end{equation}
and to the entanglement entropy $E = C/2$ for pure states, see Eq.~\eqref{eq:def_entanglement}. We provide a representation of the orbital total correlation $C(\rho)$ of a state $\rho$ on the Fock space as the minimal relative entropy of $\rho$ with respect to arbitrary uncorrelated, or product, states. The minimizer is the product state constructed with the orbital-reduced states $\rho_{A/B}$ of $\rho$. In particular, this representation of $C(\rho)$ implies $C(\rho)\geq 0$. Our analysis adapts the formalism of Ref.~\cite{ModiVedralrecapgeom} from distinguishable particles to fermionic orbitals (see also Ref.~\cite{lexin_orb_ent}).

We work with creation and annihilation operators relative to orthogonal orbitals, which we here denote with $a^{\dagger}_{j}$, $j = 1, \ldots, \lambda_1$ and $b^{\dagger}_{k}$, $k = 1, \ldots, \lambda_2$. The former refer to the first subset $A$ of orbitals, the latter to the subset $B$ of the remaining orbitals.

We first need to introduce some notation. In the following, the local parity operators $P_{A/B}$ are given by
\begin{align}
	P_A = \e^{\i \pi N_A} 	\; ,
	P_B = \e^{\i \pi N_B} \; ,
\end{align}
where $N_{A} = \sum_{j=1}^{\lambda_1} a_j^\dagger a_j$ and $N_{B} = \sum_{k=1}^{\lambda_2} b_k^\dagger b_k$ are the local number operators. The total parity operator follows as
\begin{equation}
	P = P_A P_B = P_B P_A = \e^{\i \pi (N_A + N_B)} \; .
\end{equation}
Also, we denote with $\F_{A/B}$ the local Fock spaces, and with $\F$ the total Fock space. Any given operator $\O_{A/B}$ on $\F_{A/B}$ can be lifted unambiguously to an operator on $\F$ (by representing it in terms of creation and annihilation operators), which we still denote with $\O_{A/B}$.

We recall that any operator $\O$ on $\F$ can be decomposed into an `even' part $\O^e$ and an `odd' part $\O^o$, as follows. Let $\Pi^{e}$ ($\Pi^{o}$) the projector onto the even (odd)-parity subspace of $\F$. Then
\begin{align} \nonumber
	\O &= (\Pi^{e} + \Pi^{o}) \O (\Pi^{e} + \Pi^{o}) \\ \nonumber
	& = \underbrace{\Pi^{e} \O \Pi^{e} + \Pi^{o} \O \Pi^{o}}_{\O^e} \\
	& + \underbrace{\Pi^{e} \O \Pi^{o} + \Pi^{o} \O \Pi^{e}}_{\O^o} \; .
\end{align}
$\O^e$ ($\O^o$) commutes (anticommutes) with $P$ and is a linear combination of products of an even (odd) number of creation and annihilation operators. The same decomposition can be done for local operators $\O_{A/B}$ on $\F_{A/B}$ by referring to the local parity operator $P_{A/B}$.

The above decomposition allows us to easily prove the following result.
\begin{prop} \label{prop:rho_ssr_implies_rhoA_ssr}
	Let $\rho$ be a state on $\F$ and $\rho_{A}$ its orbital-reduced states relative a subset $A$ of orbitals. If $[P, \rho]=0$ then $[P_{A}, \rho_{A}]=0$, i.e., if $\rho$ is parity-superselected, then $\rho_{A}$ is, too.
\end{prop}
\begin{proof}
	The equality $[P, \rho]=0 \Rightarrow [P_{A}, \rho_{A}]=0$ is equivalent to $\Pi_A^e \rho_A \Pi_A^o = 0$, which is turn equivalent to
	\begin{equation}
		\tr_{\F_A}[\Pi_A^e \rho_A \Pi_A^o \O_A] = 0 \;\;\; \forall \, \O_A \text{ on } \F_A \; .
	\end{equation}
	This equality is proven as follows, using the definition of the orbital partial trace given in Ref.~\cite{Friis_2016}:
	\begin{align} \nonumber
		\tr_{\F_A}[\Pi_A^e \rho_A \Pi_A^o \O_A] &= \tr_{\F_A}[\rho_A \Pi_A^o \O_A \Pi_A^e ] \\ 
		&= \tr_{\F}[\rho \Pi_A^o \O_A \Pi_A^e ] = 0 \; .
	\end{align}
	The last equality follows from the fact that $\Pi_A^o \O_A \Pi_A^e$ is an odd operator on $\F$, i.e., it changes the global parity, while $\rho$ is even, i.e., it preserves it.
\end{proof}
We are now in a position to introduce and analyze states with no correlations between $A$ and $B$.
\begin{defn}
	An \emph{uncorrelated state} is a state $\sigma$ on $\F$ such that $[P, \sigma]=0$ and such that for any operators $\O_{A/B}$ on $\F_{A/B}$
	\begin{equation} \label{eq:def_uncorrelated}
		\tr_{\F}[\sigma \O_A \O_B] = \tr_{\F}[\sigma \O_B \O_A] = \tr_{\F}[\sigma \O_A] \tr_{\F}[\sigma \O_B] \; .
	\end{equation}
\end{defn}
The parity-superselection requirement in the above definition, i.e., $[P, \sigma]=0$, is dictated by both physical meaningfulness, and mathematical convenience. The last expression in Eq.~\eqref{eq:def_uncorrelated} can also be written in terms of the orbital-reduced states of $\sigma$ as $\tr_{\F_A}[\sigma_A \O_A] \tr_{\F_B}[\sigma_B \O_B]$.

A less physically meaningful but more mathematically practical characterization of the uncorrelated states is given by the following definition and proposition.
\begin{defn}
A \emph{product state} is a state $\sigma$ on $\F$ that can be written as
\begin{equation}
\sigma = \sigma_A \sigma_B = \sigma_B \sigma_A
\end{equation}
for some states $\sigma_{A/B}$ on $\F_{A/B}$ that satisfy $[P_{A/B}, \sigma_{A/B}]=0$ (and therefore commute when lifted to operators on $\F$).
\end{defn}
We emphasize that for \emph{any} two local states $\sigma_{A/B}$ on $\F_{A/B}$ satisfying $[P_{A/B}, \sigma_{A/B}]=0$ the operator $\sigma_A \sigma_B = \sigma_B \sigma_A$ on $\F$ is a valid product state on $\F$.

We also remark that leaving out the requirement $[P_{A/B}, \sigma_{A/B}]=0$ may result in $\sigma_A \sigma_B \neq \sigma_B \sigma_A$, which makes the definition of the corresponding product state ambiguous.
\begin{prop}
	A state on $\F$ is uncorrelated if and only if it is a product state. More in detail, any uncorrelated state $\sigma$ with orbital-reduced states $\sigma_{A/B}$ satisfies $\sigma = \sigma_A \sigma_B = \sigma_B \sigma_A$.
\end{prop}
\begin{proof}
	\textbf{`If' direction.} Let $\sigma = \sigma_A \sigma_B=\sigma_B \sigma_A$ be a product state. It commutes with $P = P_A P_B$ because $[ \sigma_{A/B}, P_{A/B} ] = 0$. Next, we show Eq.~\eqref{eq:def_uncorrelated} holds true for any local operators $\O_{A/B}$. While that equation seems obvious for product states, its rigorous derivation requires extra care. Let $\{ \ket{\Phi_{j}} \}_{j=1}^{2^{\lambda_{1}}}$/$\{ \ket{\Psi_{k}} \}_{k=1}^{2^{\lambda_{2}}}$ be an ONB of $\F_{A/B}$ of eigenstates of $\sigma_{A/B}$, with corresponding eigenvalues $\alpha_j$/$\beta_k$. Since $\sigma_{A/B}$ commutes with parity, i.e., is an even operator, we can and do choose these eigenstates to be eigenstates of $P_{A/B}$, too. We need to lift them to eigenstates in $\F$ of $\sigma$. We can do that via the unambiguous inclusion of operators on $\F_{A/B}$ into the set of operators on $\F$, as follows. We introduce the operators on $\F_{A/B}$
	\begin{align}
		A_j^\dagger &= \ket { \Phi_{ j} } {}_A \bra{0} \\
		B_k^\dagger &= \ket { \Psi_{k} } {}_B \bra{0} \; ,
	\end{align}
	where $\ket{0}_{A/B}$ is the vacuum in $\F_{A/B}$ (unique up to a phase). These operators are all even or odd, because the vacua, the $\ket { \Phi_{ j} }$'s and the $\ket { \Psi_{ k} }$'s are eigenstates of the corresponding local parity operator. It is then easy to check that
	\begin{equation}
		\sigma A_j^\dagger B_k^\dagger \ket{0} = \alpha_j \beta_k A_j^\dagger B_k^\dagger \ket{0} \; ,
	\end{equation}
	where we have lifted the $A_j^\dagger$'s and $B_k^\dagger$'s to operators on $\F$; $\ket{0}$ is the vacuum in $\F$. We have thus found an orthonormal eigenbasis of $\sigma$. We now use it to compute the traces needed to prove Eq.~\eqref{eq:def_uncorrelated} for given local operators $\O_{A/B}$, as follows.
	\begin{align} \nonumber
		& \tr_{\F} [\sigma \O_A \O_B] = \nonumber \\ \nonumber
		&=\tr_{\F} [\sigma_A \sigma_B \O_A \O_B] \\ \nonumber
		&= \tr_{\F} [\sigma_A \sigma_B \O_A^e \O_B^e] \\ \nonumber
		&= \sum_{j,k} \braket{0| B_k A_j \sigma_A \sigma_B \O_A^e \O_B^e A^\dagger_j B^\dagger_k |0 } \\ \nonumber
		&= \sum_{j,k} \alpha_j \beta_k \braket{0| B_k A_j \O_A^e \O_B^e A^\dagger_j B^\dagger_k |0 } \\ \nonumber
		&= \sum_{j,k} \alpha_j \beta_k \braket{0| A_j \O_A^e A^\dagger_j B_k \O_B^e B^\dagger_k |0 }  \\ \nonumber
		&= \sum_{j,k} \alpha_j \beta_k \braket{0| \Pi_B^{\ket{0}} A_j \O_A^e A^\dagger_j B_k \O_B^e B^\dagger_k \Pi_A^{\ket{0}} |0 }  \\ \nonumber
		&= \sum_{j,k} \alpha_j \beta_k \braket{0| A_j \O_A^e A^\dagger_j \ket{0} \bra{0} B_k \O_B^e B^\dagger_k | 0} \\ \label{eq:longest_align}
		& = \underbrace{\left ( \sum_{j} \alpha_j  \braket{0| A_j \O_A^e A^\dagger_j | 0} \right )}_{t_A} \underbrace{\left ( \sum_{k} \beta_k \langle 0| B_k \O_B^e B^\dagger_k | 0 \rangle \right )}_{t_B} \; .
	\end{align}
The second equality follows from $[\sigma_{A/B}, P_{A/B}] = 0$. In the fifth equality, we used the fact that the operators $A_j^{(\dagger)}$ and $B_k^{(\dagger)}$ are even or odd, and thus for given indices $j,k$ $A_j^{(\dagger)}$ and $B_k^{(\dagger)}$ exactly commute or anticommute; also, $\O^e_A$ is even and thus commutes with $B_k$. In the sixth equality, we introduced the projectors $\Pi_A^{\ket{0}} = \prod_{j=1}^{\lambda_1} a_j a^\dagger_j$ and $\Pi_B^{\ket{0}} = \prod_{k=1}^{\lambda_2} b_k b^\dagger_k$ onto local vacua. For the seventh equality, we commuted $\Pi_{A/B}^{\ket{0}}$ with operators relative to modes in $B/A$, and used $\Pi_{A}^{\ket{0}}\Pi_{B}^{\ket{0}} = \Pi_{B}^{\ket{0}}\Pi_{A}^{\ket{0}} = \ket{0} \bra{0}$. Next, we compute $t_A$:
\begin{align} \nonumber
	t_A &=  \sum_{j} \alpha_j  \braket{0| A_j \O_A^e A^\dagger_j | 0} \\ \nonumber
	&=  \sum_{j,k} \alpha_j \beta_k  \braket{0| A_j \O_A^e A^\dagger_j | 0}   \\ \nonumber
	&=  \sum_{j,k} \alpha_j \beta_k  \braket{0| A_j \O_A^e A^\dagger_j \Pi^{\ket{0}}_B | 0}   \\ \nonumber
	&=  \sum_{j,k} \alpha_j \beta_k  \braket{0| A_j \O_A^e A^\dagger_j B_k B_k^\dagger | 0}  \\ \nonumber
	&=  \sum_{j,k} \alpha_j \beta_k  \braket{0| B_k A_j \O_A^e A^\dagger_j B_k^\dagger | 0}  \\ \nonumber
	&=  \sum_{j,k}  \braket{0| B_k A_j \sigma \O_A^e A^\dagger_j B_k^\dagger | 0}   \\ \nonumber
	&=  \tr_{\F} [ \sigma \O_A^e ]   \\ 
	&= \tr_{\F} [ \sigma \O_A ] \; .
\end{align}
The second equality follows from $\sum_k \beta_k = 1$. For the fourth equality, we used $B_k B_k^\dagger = \ket { 0 }_B \braket{\Psi_{k} | \Psi_{k} } {}_B \bra{0} = \Pi^{\ket{0}}_B$, valid for $k = 1, \ldots, 2^{\lambda_2}$. In the fifth and eight equalities, we used similar lines of reasoning as in the derivation~\eqref{eq:longest_align}. A similar calculation yields
\begin{equation}
	t_B = \tr_{\F} [ \sigma \O_B ] \; .
\end{equation}
Combining these results yields
\begin{equation} \label{eq:result_first_part}
	\tr_{\F} [\sigma \O_A \O_B] = t_A t_B = \tr_{\F} [ \sigma \O_A ] \tr_{\F} [ \sigma \O_B ] \; .
\end{equation}
To finish proving Eq.~\eqref{eq:def_uncorrelated}, we observe that
\begin{align} \nonumber
	\tr_{\F}[\sigma \O_A \O_B] &= \tr_{\F}[\sigma \O^e_A\O^e_B] \\ \nonumber
	&= \tr_{\F}[\sigma \O^e_B \O^e_A]\\
	&= \tr_{\F}[\sigma \O_B \O_A] \; .
\end{align}
We remark that $t_{A/B}$, by its definition, also equals $\tr_{\F_{A/B}}[\sigma_{A/B} \O_{A/B}]$, so that
\begin{equation} \label{eq:marginals_of_prod_state}
	\tr_{\F_{A/B}}[\sigma_{A/B} \O_{A/B}] = \tr_{\F} [\sigma \O_{A/B}] \; .
\end{equation}
This confirms that $\sigma_{A}$ and $\sigma_B$ are the orbital-reduced states of $\sigma = \sigma_A \sigma_B$.

\textbf{`Only if' direction.} Let $\sigma$ be an uncorrelated state, and denote with $\sigma_{A/B}$ its orbital-reduced state on $\F_{A/B}$. By Prop.~\ref{prop:rho_ssr_implies_rhoA_ssr}, $[\sigma_{A/B}, P_{A/B}] = 0$ (and thus $\sigma_A$ and $\sigma_B$ commute when lifted to operators on $\F$). Therefore, if
\begin{equation} \label{eq:sigma=sigmaAsigmaB}
\sigma = \sigma_A \sigma_B
\end{equation}
then $\sigma$ satisfies the definition of product states. We prove Eq.~\eqref{eq:sigma=sigmaAsigmaB} by showing that, $\forall \; \O$ on $\F$,
\begin{equation} \label{eq:target_onlyif}
	\tr_{\F} [\sigma \O] = \tr_{\F} [\sigma_A \sigma_B \O] \; .
\end{equation}
We may assume
\begin{equation}
\O = \O_A \O_B \; ;
\end{equation}
the general case follows by linearity. On one hand, the fact that $\sigma$ is uncorrelated implies
\begin{align} \nonumber
	\tr_{\F} [\sigma \O] &= \tr_{\F} [\sigma \O_A] \tr_{\F} [\sigma \O_B]\\ \label{eq:tek1}
	& = \tr_{\F_A} [\sigma_A \O_A] \tr_{\F_B} [\sigma_B \O_B] \; .
\end{align}
On the other hand, based on Eqs.~\eqref{eq:result_first_part} and~\eqref{eq:marginals_of_prod_state},
\begin{align} \nonumber
	\tr_{\F} [\sigma_A \sigma_B \O] &= \tr_{\F} [\sigma_A \sigma_B \O_A] \tr_{\F} [\sigma_A \sigma_B \O_B]\\ \label{eq:tek2}
	& = \tr_{\F_A} [\sigma_A \O_A] \tr_{\F_B} [\sigma_B \O_B] \; .
\end{align}
Eqs.~\eqref{eq:tek1} and~\eqref{eq:tek2} imply Eq.~\eqref{eq:target_onlyif}, and thus Eq.~\eqref{eq:sigma=sigmaAsigmaB}.
\end{proof}
We are now ready to connect the above construction to the bipartite total correlation of a parity-superselected state $\rho$ on $\F$, defined in Eq.~\eqref{eq:app_total_bip_corr}. We denote the quantum relative entropy of a state $\rho$ with respect to another state $\sigma$ as
\begin{equation}
	D(\rho, \sigma) := \tr[\rho (\log \rho - \log \sigma)] = - S(\rho) - \tr[\rho \log \sigma] \;.
\end{equation}
\begin{prop}
	Given a state $\rho$ on $\F$ s.t. $[\rho, P] = 0$,
	\begin{equation}
		C(\rho) = \min_{\sigma} D(\rho , \sigma) \; ,
	\end{equation}
	where the minimization is over all uncorrelated states, i.e., product states $\sigma = \sigma_A \sigma_B$. Furthermore, the unique minimizer is
	\begin{equation}
		\sigma^* = \rho_A \rho_B = \rho_B \rho_A \; ,
	\end{equation}
	where $\rho_{A}$ and $\rho_B$ are the orbital-reduced states of $\rho$. Therefore,
	\begin{equation}
		C(\rho) = D(\rho , \rho_A \rho_B) \; .
	\end{equation}
\end{prop}
\begin{proof}
	We have
	\begin{equation} \label{eq:min_rel_entr}
		\min_{\sigma \text{ product state}} D(\rho , \sigma) = -S(\rho) - \max_{\sigma_A, \sigma_B}  \tr_{\F} [\rho \log ( \sigma_A \sigma_B)] \; ,
	\end{equation}
	where the maximization is over local states $\sigma_{A/B}$ satisfying $[\sigma_{A/B}, P_{A/B}] = 0$. Since this implies $[\sigma_A , \sigma_B] = 0$, we can break down the logarithm and let the orbital-reduced states of $\rho$ come into play:
	\begin{align} \nonumber
		\tr_{\F} [\rho \log ( \sigma_A \sigma_B)] &= \tr_{\F} [\rho ( \log  \sigma_A + \log  \sigma_B)] \\
		&= \tr_{\F_A} [\rho_A \log  \sigma_A ] + \tr_{\F_B} [\rho_B \log  \sigma_B ] \; .
	\end{align}
	Therefore the maximization can be split:
	\begin{align} \nonumber
		\max_{\sigma_A, \sigma_B} & \tr_{\F} [\rho \log ( \sigma_A \sigma_B)] \\
		&= \max_{\sigma_A}  \tr_{\F_A} [\rho_A \log ( \sigma_A)] + \max_{\sigma_B}  \tr_{\F_B} [\rho_B \log ( \sigma_B)] \; .
	\end{align}
	A standard argument (see, e.g., Ref.~\cite{ModiVedralrecapgeom}) shows that the unique extremizers are $\sigma^*_A = \rho_A$ and $\sigma^*_B = \rho_B$, respectively, with respective maxima $-S(\rho_A)$ and $- S(\rho_B)$. Plugging this result into Eq.~\eqref{eq:min_rel_entr} yields
	\begin{align} \nonumber
		\min_{\sigma \text{ product state}} D(\rho , \sigma) &= D(\rho, \rho_A \rho_B) \\
		&= -S(\rho) + S(\rho_A) + S(\rho_B) = C(\rho) \; .
	\end{align}
\end{proof}
The following corollary follows from the non-negativity of the quantum relative entropy.
\begin{coro} \label{cor:mutual_info_geq_0}
	For any state $\rho$ on $\F$ s.t. $[\rho, P] = 0$,
	\begin{equation}
		C(\rho) \geq 0 \; .
	\end{equation}
\end{coro}

\section{Normal form of bipartite Slater determinants} \label{app:normal_form}

We now adapt the normal form of pure fermionic Gaussian states relative to a direct-sum bipartition of the underlying 1P space (presented in Ref.~\cite{botero_normal_form_ferm_Gaussian_states}; see also Ref.~\cite{botero_normal_form_bos_Gaussian_states} and Sec.~\ref{subsec:gaussian}) to the particular case of \bt bipartite' Slater determinants.

Once the 1P space $\H_1 \cong \C^d$ has been fixed, there is a one-to-one correspondence between all $N$-fermion SD states $\ket{SD}\bra{SD}$ and their 1P reduced operators, i.e., rank-$N$ orthogonal projectors $\gamma$ on $\H_1$. In that light, we will first introduce in Prop.~\ref{prop:normal_form_gamma} a normal form of orthogonal projectors relative to a direct-sum bipartition of $\H_1$ into orthogonal subspaces, $\H_1 = V_1 \oplus V_2$. Roughly speaking, the point of this normal form is twofold: (1) the $N$-dimensional image of any orthogonal projector $\gamma$ and $V_{1/2}$ have an intersection that has a known minimal dimension $T^P_{1/2} \geq 0$, and the same applies to the kernel of $\gamma$; (2) after reducing the 1P space dimension by discarding these intersections, one can always find an ONB of the remaining subspace of the image of $\gamma$ (or of its kernel) with elements $\cos\theta_j \ket{\phi_{1,j}} + \sin\theta_j \ket{\phi_{2,j}}$ that can be expressed in terms of \emph{orthonormal} vectors $\ket{\phi_{1,j}} \in V_1$ and $\ket{\phi_{2,j}} \in V_2$. 

In Corollary~\ref{coro:normal_form_SD} we will then translate this result into a normal form of bipartite SDs, which resembles very much the construction presented at the end of Sec.~\ref{sec:main_problem} and in Fig.~\ref{fig:max_SD_construction}.

We start with a simple and interesting result in linear algebra.
\begin{lemm} \label{lem:evalsubmatrix}
	Let $A$ be a linear operator on a (complex or real) vector space $V$ of dimension $d$, and let $V'$ be a subspace of $V$ of dimension $d'$. Suppose $A$ has an eigenvalue $\alpha$ with geometric multiplicity $m$, i.e., $A$ has an $\alpha$-eigenspace $E_{\alpha}$ of dimension $m$, and $m > d-d'$. Then
	\begin{equation} \label{eq:dimensionseigenspaces}
		\emph{dim}(V' \cap E_{\alpha}) \geq m - (d-d') >0 \; .
	\end{equation}
	In particular the compression of $A$ to $V'$ has the eigenvalue $\alpha$ with geometric multiplicity $m' \geq m - (d-d')$, and we say it has $m - (d-d')$ \emph{trivial eigenvalues} $\alpha$.
\end{lemm}
\begin{proof}
	$\text{dim}(V' \cap E_{\alpha}) = \text{dim}(V') + \text{dim}(E_{\alpha}) - \text{dim}(V' +E_{\alpha} )  \geq m - (d-d')$. The fact that the compression of $A$ to $V'$ has at least $m-(d-d')$ eigenvalues $\alpha$ can also be shown from Cauchy's interlacing inequalities~\cite{Fan1957Imbedding}.
\end{proof}
We now present the normal form of orthogonal projectors, relative to a bipartition of the underlying space as a direct sum.
\begin{prop} \label{prop:normal_form_gamma}
	Let $\gamma$ be a rank-$N$ orthogonal projector on $\H_1 \cong \C^d$. Let $V_1 \cong \C^{\lambda_1}$ and $V_2 \cong \C^{\lambda_2}$ be subspaces of $\H_1$ such that $V_1 \perp V_2$ and $V_1 \oplus V_2 = \H_1$. For $\mu = 1, 2$, set
	\begin{align} \label{eq:TPmu}
		T^P_{\mu} &= \max \{ 0, \lambda_{\mu} - (d - N ) \} \\
		T^H_{\mu} &= \max \{ 0, \lambda_{\mu} - N \} \\ \label{eq:R_bis}
		R &= \min \{ \lambda_1, \lambda_2, N, d - N \} \; .
	\end{align}
	Then there is an ordered ONB of $\H_1$ whose first (last) $\lambda_{1}$ ($\lambda_2$) elements span $V_{1}$ ($V_2$) with respect to which $\gamma$ is represented by a matrix with the following form,
		\begin{equation} \label{eq:gamma_6_blocks}
		\gamma \rightarrow \begin{pmatrix}
			\begin{array}{ccc|ccc}
				\I_{T^P_1} & & & & & \\
				& \mathbf{0}_{T^H_1} & & & & \\
				& & C^2 & & & C \cdot S \\
				\hline
				& & & \I_{T^P_2} & & \\
				& & & & \mathbf{0}_{T^H_2} & \\
				& & C \cdot S & & & S^2
			\end{array}
		\end{pmatrix} ~ ,
	\end{equation}
	where the $2 \times 2$-block structure refers to the direct sum $\H_1 = V_1 \oplus V_2$, empty spaces indicate zero submatrices, and
	\begin{align} \label{eq:diag_C}
		C &= \textup{diag}(\cos \theta_1, \ldots, \cos \theta_R) \\ \label{eq:diag_S}
		S &= \textup{diag}(\sin \theta_1, \ldots, \sin \theta_R) \; ,
	\end{align} 
	with $\theta_1, \ldots, \theta_R \in [0, \pi/2]$.
\end{prop}
\begin{proof}
	\textbf{Part 1.} We first consider the particular case where
	\begin{equation} \label{eq:simple_case_constraints}
		\lambda_1 = \lambda_2 = N = d - N  = R \til.
	\end{equation}
	This case is simpler than the general case in that it involves no \bt trivial eigenvalues' $0$ and $1$ of the compression of $\gamma$ to either $V_1$ or $V_2$, i.e.,
	\begin{equation}
		T^{P}_{1} =T^{H}_{1} =T^{P}_{2} =T^{H}_{2} = 0 \til .
	\end{equation}
	By fixing an arbitrary ONB of $V_1$ and one of $V_2$, we can represent $\gamma$ by highlighting the block structure relative to the bipartition of $\H_1$, as follows:
	\begin{equation}
		\gamma \rightarrow
		\begin{pmatrix}
			\begin{array}{c|c}
				\gamma_{11} & \gamma_{12}  \\
				\hline
				\gamma_{21} & \gamma_{22}
			\end{array}
		\end{pmatrix} \til .
	\end{equation}
	The blocks $\gamma_{jk}$ are $R \times R$ matrices. We now focus on the two diagonal ones, $\gamma_{11}$ and $\gamma_{22}$, which represent the compressions $\evalat{\gamma}{V_{1}}$ and $\evalat{\gamma}{V_{2}}$ of $\gamma$ to $V_1$ and $V_2$, respectively. We aim to show that the eigenvalues of $\gamma_{11}$ and $\gamma_{22}$ come in pairs that add up to $1$. To that end, we consider the diagonalization equation
	\begin{equation}
		\begin{pmatrix}
			\begin{array}{c|c}
				\gamma_{11} & \gamma_{12}  \\
				\hline
				\gamma_{21} & \gamma_{22}
			\end{array}
		\end{pmatrix}
		=
		\underbrace{\begin{pmatrix}
			\begin{array}{c|c}
				U_{11} & U_{12}  \\
				\hline
				U_{21} & U_{22}
			\end{array}
		\end{pmatrix}}_{U}
		\cdot
		\begin{pmatrix}
			\begin{array}{c|c}
				\I_{R} & 0  \\
				\hline
				0 & 0
			\end{array}
		\end{pmatrix}
		\cdot
		\underbrace{\begin{pmatrix}
			\begin{array}{c|c}
				U_{11}^\dagger & U_{21}^\dagger  \\
				\hline
				U_{12}^\dagger & U_{22}^\dagger
			\end{array}
		\end{pmatrix}}_{U^\dagger} \; ,
	\end{equation}
	where $U$ is a $d \times d$ unitary matrix. At this stage, a cosine-sine decomposition of $U$ would quickly lead to the desired representation of $\gamma$. We follow a different, more insightful approach that directly targets $\gamma$. We get
		\begin{align}
		\gamma_{11} &= U_{11} U_{11}^{\dag} \\
		\gamma_{22} &= U_{21} U_{21}^{\dag} \til.
	\end{align}
	The unitarity of $U$ implies
	\begin{equation}
		U_{11}^\dagger U_{11} + U_{21}^\dagger U_{21} = \I_R \; ,
	\end{equation}
	which implies that eigenvalues of $U_{11}^\dagger U_{11}$ and of $U_{21}^\dagger U_{21}$ come in pairs of the form $\beta_j, 1 - \beta_j$ (with possible repetitions). A standard argument then shows that the eigenvalues of $U_{11}^\dagger U_{11}\;(U_{21}^\dagger U_{21})$ are the same as those of $ U_{11} U_{11}^{\dag} = \gamma_{11}\;(U_{21} U_{21}^{\dag} = \gamma_{22})$, implying a pairing of the eigenvalues of $\gamma_{11}$ and $\gamma_{22}$. Finally, given any eigenvalue $\beta_j$ of $\gamma_{11}$, we have the standard bounds
	\begin{equation}
		0 = \min(\text{spec}(\gamma)) \leq \beta_j \leq \max(\text{spec}(\gamma))=1 \; .
	\end{equation}
	Therefore, we can indicate the $j$th pair $\beta_j, 1 - \beta_j$ of eigenvalues of, respectively, $\gamma_{11}$ and $\gamma_{22}$ in terms of an angle $\theta_j \in [0, \pi/2]$, i.e., $\beta_j = \cos^2 \theta_j$ and $1 - \beta_j = \sin^2 \theta_j$.
	
	The pairing of the eigenvalues of $\gamma_{11}$ and $\gamma_{22}$, i.e., of the compressions $\evalat{\gamma}{V_1}$ and $\evalat{\gamma}{V_2}$, corresponds to an analogous pairing of the corresponding eigenvectors. Representing $\gamma$ in the basis of these eigenvectors will result in the matrix structure of Eq.\til\eqref{eq:gamma_6_blocks}. Let $\ket{\alpha} \in V_1$ be an eigenvector of $\evalat{\gamma}{V_1}$ with eigenvalue $\cos^2 \theta_j$ and $\ket{\alpha'} \in V_2$ be an eigenvector of $\evalat{\gamma}{V_2}$ with eigenvalue $\sin^2 \theta_k$, with $1 \leq j, k \leq R$. Also, we introduce the projectors $\Pi_{1/2}$ onto $V_{1/2}$. We get
	\begin{align} \nonumber
		\braket{\alpha | \gamma | \alpha'} &= \braket{\alpha | \gamma^2 | \alpha'}\\ \nonumber
		&= (\bra{\alpha} \gamma) (\gamma \ket{\alpha'})\\ \nonumber
		&= (\cos^2 \theta_j \bra{\alpha} + \bra{\alpha} \gamma \Pi_2) (\Pi_1 \gamma \ket{\alpha'} + \sin^2 \theta_k \ket{\alpha'})\\ \nonumber
		&= \cos^2 \theta_j \bra{\alpha}\Pi_1 \gamma \ket{\alpha'} + \sin^2 \theta_k \bra{\alpha} \gamma \Pi_2 \ket{\alpha'} \\
		&= (\cos^2 \theta_j + \sin^2 \theta_k) \braket{\alpha | \gamma | \alpha'} \; ,
	\end{align} 
implying that either $\theta_j = \theta_k$, or $\braket{\alpha | \gamma | \alpha'} = 0$. That is, $\gamma$ has zero matrix elements between eigenvectors of $\evalat{\gamma}{V_1}$ and eigenvectors of $\evalat{\gamma}{V_2}$ corresponding to different angles. Therefore, if we pick any ONB of $\H_1$ that consists of eigenvectors of $\evalat{\gamma}{V_1}$ and $\evalat{\gamma}{V_2}$, the corresponding matrix representation of $\gamma$ has a block-diagonal structure: there is one block for each angle $\theta \in \{\theta_j\}_{j=1}^R$, and the size of the block equals the \bt degeneracy' of the corresponding $\theta$, i.e., the number of indices $j \in \{ 1, \ldots, R \}$ s.t. $\theta_j = \theta$. For a given angle $\theta$ with degeneracy $\delta$, let us denote the corresponding block with
\begin{equation} \label{eq:theta_block}
	B_\theta = \begin{pmatrix}
		\begin{array}{c|c}
			\cos^2 \theta \I_\delta & X \\
			\hline
			X^\dagger & \sin^2 \theta \I_\delta
		\end{array}
	\end{pmatrix} \; ,
\end{equation}
where the block structure is inherited from the initial bipartition of $\H_1$, and $X$ is a $\delta \times \delta$ matrix. Consider the SVD decomposition of $X$ in terms of unitaries $V_1$ and $V_2$ and a diagonal matrix $\Delta$ with real, nonnegative entries, $X = V_1 \Delta V_2^\dagger$. Then
\begin{equation}
	B_\theta = \begin{pmatrix}
		\begin{array}{c|c}
			V_1 & 0  \\
			\hline
			0 & V_2
		\end{array}
	\end{pmatrix}
	\cdot
		\begin{pmatrix}
			\begin{array}{c|c}
				\cos^2 \theta \I_\delta & \Delta  \\
				\hline
				\Delta & \sin^2 \theta \I_\delta
			\end{array}
		\end{pmatrix}
		\cdot
	\begin{pmatrix}
		\begin{array}{c|c}
			V_1^\dagger & 0  \\
			\hline
			0 & V_2^\dagger
		\end{array}
	\end{pmatrix} \; .
\end{equation}
This implies that, up to separate basis transformations within the eigenspaces of $\evalat{\gamma}{V_1}$ and $\evalat{\gamma}{V_2}$, we may take $X$ to be diagonal with nonnegative entries in Eq.\til\eqref{eq:theta_block}. This further breaks down the block $B_\theta$ into $\delta$ smaller $2 \times 2$ blocks, one per $\theta_j$ (or better, one per label value $j \in \{ 1, \ldots, R \}$, to account for repetitions). The $j$th $2 \times 2$ block,
\begin{equation}
	B'_j = \begin{pmatrix}
		\begin{array}{c|c}
			\cos^2 \theta & x  \\
			\hline
			x & \sin^2 \theta
		\end{array}
	\end{pmatrix} \; ,
\end{equation}
must have eigenvalues in $\{0,1\}$ because these are eigenvalues of $\gamma$, as well. In view of $\tr(B'_j) = 1$, this implies
\begin{equation}
	x = \cos \theta \sin \theta \; .
\end{equation}

To sum up, we discovered that there exists an ONB of $\H_1$ that consists of elements of $V_1$ and $V_2$ (namely, appropriately chosen eigenvectors of $\evalat{\gamma}{V_1}$ and $\evalat{\gamma}{V_2}$) such that the corresponding matrix representation of $\gamma$ has the following form,
\begin{equation}
	\gamma \rightarrow 
	\begin{pmatrix}
		\begin{array}{c|c}
			C^2 & C \cdot S  \\
			\hline
			C \cdot S & S^2
		\end{array}
	\end{pmatrix} \; ,
\end{equation}
where the diagonal matrices $C$ and $S$ are given by Eqs.\til\eqref{eq:diag_C} and\til\eqref{eq:diag_S}.

\textbf{Part 2.} We now lift the assumption of Eq.\til\eqref{eq:simple_case_constraints} and extend our proof from the particular to the general case. The strategy is to use the high degeneracies of the eigenvalues $0$ and $1$ of the projector $\gamma$ and Lemma\til\ref{lem:evalsubmatrix} to reduce to the particular case. Lemma\til\ref{lem:evalsubmatrix} implies that, for $\mu = 1, 2$, the compression $\evalat{\gamma}{V_\mu}$ has (at least) $T^P_{\mu}$ eigenvalues 1 and (at least) $T^H_{\mu}$ eigenvalues 0. Furthermore, the corresponding eigenvectors are not only eigenvectors of $\evalat{\gamma}{V_\mu}$, \emph{but also of $\gamma$}. For example, given $\ket{\phi} \in V_1$ of $\evalat{\gamma}{V_1}$ s.t. $\evalat{\gamma}{V_1} \ket{\phi} = \ket{\phi}$, we have
\begin{align} \nonumber
	|| \gamma \ket{\phi} ||^2 &= \bra{\phi}\gamma^2 \ket{\phi} \\ \nonumber
	&= \bra{\phi}\gamma \ket{\phi} \\ \nonumber
	&= \bra{\phi}\evalat{\gamma}{V_1} \ket{\phi} \\ \nonumber
	&= \braket{\phi | \phi} \\ \nonumber
	&= || \evalat{\gamma}{V_1} \ket{\phi} ||^2 \; ,
\end{align}
which implies that $\gamma \ket{\phi}$ has no component in $V_2$, and thus
\begin{equation}
	\gamma \ket{\phi} = \evalat{\gamma}{V_1} \ket{\phi} = \ket{\phi} \; .
\end{equation}
This leads us to consider any ONB $\B$ of $\H_1$ that includes $T^P_1$ 1-eigenvectors of $\evalat{\gamma}{V_1}$, i.e., orthonormal vectors in $V_1$ that are eigenvectors of $\evalat{\gamma}{V_1}$ relative to the eigenvalue 1, $T^H_1$ 0-eigenvectors of $\evalat{\gamma}{V_1}$, $T^P_2$ 1-eigenvectors of $\evalat{\gamma}{V_2}$, and $T^H_2$ 0-eigenvectors of $\evalat{\gamma}{V_2}$. Up to reordering the ONB $\B$, the matrix representation of $\gamma$ relative to $\B$ has the same first, second, fourth, and fifth \bt block rows' and \bt block columns' as the matrix in Eq.\til\eqref{eq:gamma_6_blocks}.

One can then discard these trivial blocks of the matrix and focus on the nontrivial ones. These blocks still need to be brought to normal form by appropriately choosing the remaining elements of the ONB $\B$. In other words, one has to analyze the remaining eigenvalues and eigenvectors of $\gamma$. Some bookkeeping shows that there are
\begin{equation}
	N' = N - T^P_1 - T^P_2 = R
\end{equation}
remaining eigenvalues 1, and
\begin{equation}
	d' - N' = (d - N) - T^H_1 - T^H_2 = R
\end{equation}
remaining eigenvalues 0. The corresponding remaining eigenvectors span a reduced one-particle space $\H_1' \subseteq \H_1$ of dimension $d' = 2R$, on which the rank-$N'$ orthogonal projector $\gamma' = \evalat{\gamma}{\H_1'}$ acts. $\H_1'$ comes with an inherited partition, $\H_1' = V_1' \oplus V_2'$, into subspaces $V_1' \subseteq V_1$ and $V_2' \subseteq V_2$ of dimension
\begin{equation}
	\lambda_1' = \lambda_1 - T_1^P - T_1^H = R
\end{equation}
and
\begin{equation}
	\lambda_2' = \lambda_2 - T_2^P - T_2^H = R \; ,
\end{equation}
respectively. Our \bt primed' notation highlights the following: using Lemma\til\ref{lem:evalsubmatrix} to get rid of the trivial eigenvalues 0 and 1, we reduced the general case of arbitrary $N, d$ and $\lambda_{1/2}$ to the particular case discussed in part (1) of this proof. Applying the result of part (1) to the projector $\gamma'$ acting on $\H_1'$, with respect to the partition $\H_1' = V'_1 \oplus V'_2$, finally allows us to complete the ONB $\B$ of $\H_1$ such that the corresponding matrix representation of $\gamma$ takes the form of Eq.\til\eqref{eq:gamma_6_blocks}.
\end{proof}
Going from the level of 1RDMs back to the level of SDs yields a normal form of bipartite SDs:
\begin{coro} \label{coro:normal_form_SD}
	Let $\ket{SD} \in \F_{\H_1}$ be an $N$-fermion Slater determinant. Let $V_1 \cong \C^{\lambda_1}$ and $V_2 \cong \C^{\lambda_2}$ be subspaces of the 1P space $\H_1 \cong \C^d$ such that $V_1 \perp V_2$ and $V_1 \oplus V_2 = \H_1$. Let $T^P_1, T^P_2$ and $R$ be defined as in Eqs.~\eqref{eq:TPmu} and~\eqref{eq:R_bis}. Then there exist orthonormal orbitals
	\begin{align} \nonumber
		& \ket{\phi_1}, \ldots , \ket{\phi_{T^P_1 + R}} \in V_1 \\ \nonumber
		& \ket{\psi_1}, \ldots , \ket{\psi_{T^P_2 + R}} \in V_2
	\end{align}
	and angles
	\begin{equation}
		\theta_1 , \ldots , \theta_R \in [0, \pi/2] 
	\end{equation}
	such that
	\begin{equation} \label{eq:SD_rep_normal}
		\ket{SD} = f^\dagger_{\phi_1} \ldots f^\dagger_{\phi_{T^P_1}} f^\dagger_{\psi_1} \ldots f^\dagger_{\psi_{T^P_2}} f^\dagger_{\chi_1} \ldots f^\dagger_{\chi_R} \ket{0} \; ,
	\end{equation}
	where
	\begin{equation}
		\ket{\chi_j} = \cos \theta_j \ket{\phi_{T^P_1 + j}} + \sin \theta_j \ket{\psi_{T^P_2 + j}}, \;\;\; j = 1, \ldots , R \; .
	\end{equation}
\end{coro}
\begin{proof}
	Let $\gamma$ be the 1P-reduced operator of $\ket{SD}$. $\gamma$ is a rank-$N$ orthogonal projector on $\H_1$. This allows us to use Prop.~\ref{prop:normal_form_gamma} and find an ordered ONB $\B = ( \ket{j} )_{j=1}^d$ of $\H_1$, made up of an ONB $(\ket{j})_{j=1}^{\lambda_1}$ of $V_1$ and an ONB $(\ket{j})_{j=\lambda_1 + 1}^{d}$ of $V_2$, with respect to which the matrix representation of $\gamma$ takes the form given by Eqs.~\eqref{eq:gamma_6_blocks},~\eqref{eq:diag_C} and~\eqref{eq:diag_S}. From that matrix form, we can read off the elements of an ONB of the range of $\gamma$, namely
	\begin{align} \nonumber
		\ket{\phi_j} &= \ket{j} \in V_1 , \;\;\; j= 1, \ldots, T^P_1  \\ \nonumber
		\ket{\psi_j} &= \ket{\lambda_1 + j} \in V_2, \;\;\; j= 1, \ldots, T^P_2 \\
		\ket{\chi_j} &= \cos \theta_j \ket{\phi_{T^P_1 + j}} + \sin \theta_j \ket{\psi_{T^P_2 + j}}, \;\;\; j = 1, \ldots , R \; ,
	\end{align}
	where in the last equation we simply renamed
	\begin{align} \nonumber
	\ket{\phi_{T^P_1 + j}} &= \ket{\lambda_1 - R +  j} \in V_1 \\
	\ket{\psi_{T^P_2 + j}} &= \ket{d - R + j} \in V_2 \; ,
	\end{align}
	to match the notation of the corollary that we are proving. This ONB - as any other ONB of the range of $\gamma$ - can be used to represent the Slater determinant $\ket{SD}$ in terms of creation operators applied to the vacuum vector, leading to Eq.~\eqref{eq:SD_rep_normal} (up to a phase factor that can be absorbed in, e.g., $\ket{\phi_1}$).
\end{proof}
The above Corollary can be used to show that the two orbital-reduced states of any SD are isospectral, which is actually true for arbitrary parity-superselected fermionic pure states~\cite{Friis_2016}.

We now prove another direct consequence of Prop.~\ref{prop:normal_form_gamma}, namely Prop.~\ref{prop:bset_bipartitions_SD}, stated and discussed in Appendix~\ref{app:main_prob_revisited}. We refer to Sec.~\ref{subsec:towards_general} for the notation. Pick
\begin{equation}
	\bb = (b_1, \dots , b_{\lambda_1} \sbs b_{\lambda_1 + 1}, \dots , b_{\lambda_1 + \lambda_2}) \in \bset
\end{equation}
and a $d \times d$ rank-$N$ projection matrix $\gamma$ with block-spectrum $\bb$. Then Prop.~\ref{prop:normal_form_gamma} implies that $\gamma$ can be brought to the form given by Eqs.~\eqref{eq:gamma_6_blocks},~\eqref{eq:diag_C} and~\eqref{eq:diag_S} by conjugation with a block-diagonal unitary, which does not change the spectrum of the two diagonal blocks. Since the diagonal blocks of the transformed matrix are diagonal, their spectra can be easily read off, and are indeed of the form explained in Eqs.~\eqref{eq:constraint1} and~\eqref{eq:constraint2}. Viceversa, given a vector $\bb$ with the form of Eqs.~\eqref{eq:constraint1} and~\eqref{eq:constraint2}, one can construct the rank-$N$ projection matrix of Eq.~\eqref{eq:gamma_6_blocks}, observe that the decreasingly-ordered spectra $\sigma_1$ and $\sigma_2$ of its two diagonal blocks indeed satisfy $\bb = (\sigma_1 | \sigma_2)$, and conclude that $\bb \in \bset$.

\section{Proofs of the propositions on free states} \label{app:proofs_for_basic_propositions}
We report from Sec.~\ref{subsec:proof_of_main_theorem} and prove Props.~\ref{prop:marginalsoffreestates} and~\ref{prop:entropyfreestates}.
% previously restatable \marginalfreestate*
\begin{prop*}
	Let $\rho \in \L(\F_{\H_1})$ be a fermionic free state, let $\gamma \in \L(\H_1)$ be its 1P-reduced operator, and let $V$ be a linear subspace of $\H_1$. Then the orbital-reduced state of $\rho$ relative to (the orbitals in) $V$ is the free state on the Fock space $\F_{V}$ whose 1P-reduced operator is the compression of $\gamma$ to $V$.
\end{prop*}
\begin{proof}
	Let $\rho_V \in \L(\F_{V})$ denote the orbital-reduced state of $\rho$ relative to $V$ (as introduced in Ref.~\cite{Friis_2016}), and let $\rho_{\gamma_V} \in \L(\F_{V})$ denote the unique~\cite{Gottliebfermbos} free state whose 1P-reduced operator is the compression $\gamma_V$ of $\gamma$ to $V$. We prove $\rho_V = \rho_{\gamma_V}$ by showing that
	\begin{equation} \label{eq:target_average_eq}
		\tr_{\F_V}[\rho_V O] = \tr_{\F_V}[\rho_{\gamma_V} O]
	\end{equation}
	holds true for any operator $O$ on $\F_{V}$. By linearity, it suffices to prove the above equality for operators $O$ which are products of creation/annihilation operators relative to orbitals in $V$. Given such a product $O$, by definition, $\tr_{\F_V}[\rho_V O] = \tr_{\F_{\H_1}}[\rho O]$ (with slight abuse of notation, we regard the operator $O \in \L (\F_V)$ as an operator on $\F_{\H_1}$, too). Since $\rho$ is free, Wick's theorem~\cite{wick_gaudin, molinari_wick} allows one to express $\tr_{\F_{\H_1}}[\rho O]$ in terms of matrix elements $\braket{\phi | \gamma | \psi} = \tr_{\F_{\H_1}} [\rho f^\dagger_{\psi} f_\phi]$ relative to orbitals $\ket{\phi},\ket{\psi} \in V$. Similarly, freeness of $\rho_{\gamma_V}$ allows us to express $\tr_{\F_V}[\rho_{\gamma_V} O]$ in the same way and in terms of the same matrix elements, hence Eq.~\eqref{eq:target_average_eq}.
\end{proof}
% previously restatable \entropyfreestate*
\begin{prop*}
	Let $\rho \in \L(\F_{\H_1})$ be a fermionic free state. Let $\gamma$ be the 1P-reduced operator of $\rho$ and $n_1,\dots,n_d$ the eigenvalues of $\gamma$, i.e., the NONs of $\rho$, with $d = \emph{dim}(\H_1)$. Then
	\begin{equation}
		S(\rho) = \tr [h(\gamma)] = H(\bn) \,.
	\end{equation}
\end{prop*}
\begin{proof}
	We recall the main passages of the calculation of $S(\rho)$ as presented in Ref.~\cite{gottliebfermions}. Assuming $\rho$ is not on the boundary of the set of free states, it can be expressed in terms of fermionic operators $f^{(\dagger)}_j$ relative to its natural orbitals as $\rho = \e^{-\sum_{j=1}^d \epsilon_j f^\dagger_j f_j}/Z$, with $d = \text{dim} (\H_1)$, $\epsilon_j \in \R$ and $Z$ the normalization constant given by $Z = \prod_{j=1}^d Z_j $, with $Z_j = 1 + \e^{-\epsilon_j}$. The $2^d$ eigenvalues of $\rho$, indexed by $d$ occupation numbers $\boldsymbol{m} = (m_1, \dots, m_d) \in \{0,1\}^d$, are given by
	\begin{equation}
		\lambda_{\boldsymbol{m}} = \prod_{j=1}^d \frac{\e^{- m_j \epsilon_j}}{Z_j} = \prod_{j=1}^d n_j^{m_j} (1 - n_j)^{1 - m_j} \; .
	\end{equation}
	In the last equality, we rewrote each factor in terms of the NONs $n_j = 1/(\e^{\epsilon_j} + 1)$ of $\rho$~\cite{Peschel_2003}, which are well-defined also in the case that $\rho$ lies on the boundary of the set of free states (that is, $\epsilon_j \rightarrow \pm \infty$, or $n_j \in \{ 0, 1 \}$). Plugging the last equation into the entropy
	\begin{equation}
		S(\rho) = - \sum_{\boldsymbol{m} \in \{0,1\}^d} \lambda_{\boldsymbol{m}} \log(\lambda_{\boldsymbol{m}})
	\end{equation} 
	and a short calculation yield Eq.~\eqref{eq:entropy_of_free_state}.
\end{proof}

\section{Alternative computation of $E_{\max}$} \label{app:main_prob_revisited}

We revisit the computation of $C_{\max} = 2 E_{\max}$ following the discussion of Sec.~\ref{subsec:towards_general}. Task (1), i.e., the characterization of $\bset$, is the aim of the following proposition:
\begin{prop} \label{prop:bset_bipartitions_SD}
	Let $\bb = (b_1, \dots , b_{\lambda_1} \sbs b_{\lambda_1 + 1}, \dots , b_{\lambda_1 + \lambda_2})$ $ \in \R^d$. Then $\bb \in \bset$ if and only if
	\begin{align} \label{eq:constraint1}
		(b_1, \dots , b_{\lambda_1} ) &= (\underbrace{1,\dots,1}_{T^P_1 \emph{ times}}, \beta_1,\dots,\beta_R,\underbrace{0,\dots,0}_{T^H_1\emph{ times}} ) \\ \label{eq:constraint2}
		(b_{\lambda_1 + 1}, \dots , b_{\lambda_1 + \lambda_2} ) &=  (\underbrace{1,\dots,1}_{T^P_2\emph{ times}}, 1-\beta_R,\dots,1 - \beta_1,\underbrace{0,\dots,0}_{T^H_2\emph{ times}}) \til ,
	\end{align}
	where $1 \geq \beta_1\geq\dots\geq \beta_R \geq 0$ and, for $\mu = 1, 2$,
	\begin{align}
		T^P_{\mu} &= \max \{ 0, \lambda_{\mu} - (d - N ) \} \\
		T^H_{\mu} &= \max \{ 0, \lambda_{\mu} - N \} \\ \label{eq:R}
		R &= \min \{ \lambda_1, \lambda_2, N, d - N \} \; .
	\end{align}
\end{prop}
The characterization of $\bset$ presented above emerges naturally from our discussion in Appendix~\ref{app:normal_form}, where it is derived in the last lines.

Notice that the four coefficients $T^{P/H}_{\mu}$ are related: at least two of them must vanish, for any choice of $N, d$ and $\lambda_{1/2}$. The block spectra $\bb \in \bset$ are thus characterized in terms of three conditions: (i) $\bb$ should contain a certain amount $T^{P/H}_{1} + T^{P/H}_{2}$ of `trivial' block eigenvalues 1 and 0, corresponding, respectively, to the occupied and empty orbitals (as previously stated in Eqs.~\eqref{eq:block_spec_cond1},~\eqref{eq:block_spec_cond3} and~\eqref{eq:block_spec_cond4}); (ii) the remaining $R + R$ `nontrivial' block eigenvalues $b_j$ should be in the interval $[0,1]$, which can be seen as a consequence of the Pauli principle (as stated in Eqs.~\eqref{eq:block_spec_cond2} and~\eqref{eq:block_spec_cond5}); (iii) each nontrivial eigenvalue $b_j$ associated with the first block, i.e., $j \leq \lambda_1$, should correspond to a nontrivial eigenvalue $b_{j'}$ associated with the second block, i.e., $j' > \lambda_1$, such that the pair sums to 1. This characterization of $\bset$, and, in particular, the pairing of nontrivial eigenvalues, is best understood in terms of the normal form of SDs and their 1RDMs relative to a splitting of the 1P space (see Ref.~\cite{botero_normal_form_ferm_Gaussian_states} and Prop.~\ref{prop:normal_form_gamma} in Appendix~\ref{app:normal_form}). Incidentally, note that the particular structure of the block spectrum of the 1RDM of any SD ensures that the two orbital-reduced states of any SD are isospectral, which is known to be true for arbitrary parity-superselected fermionic pure states~\cite{Friis_2016}.

Task (2), i.e., the maximization of Eq.~\eqref{eq:I_max_re_rewritten}, is now straightforward. The maximum value of $H(\bb)$ over $\bset$ is attained for the block spectrum $\bb$ whose nontrivial eigenvalues all equal $1/2$ (compare with Eq.~\eqref{eq:maximal_sum_h}). However, with a view to addressing the general problem of maximal multipartite correlation in general free states, we state this result from the point of view of \emph{majorization}~\cite{nielsenmajo}, as follows.	
\begin{coro} \label{cor:bbets_SDs}
	Let $\bb_{\max} \in \bset $ have the correct numbers $T^{P}_\mu$ and $T^H_\mu$ of trivial entries 1 and 0 for $\mu = 1, 2$, as explained in Prop.~\ref{prop:bset_bipartitions_SD}, and with all nontrivial entries equal to $1/2$. Then $\bb_{\max} \prec \bb$ $\forall \bb \in \bset$.
\end{coro}
\begin{proof}
	The largest $T^P_1 + T^P_2$ entries of $\bb_{\max}$ and those of $\bb$ all equal $1$; analogously, the smallest $T^H_1 + T^H_2$ entries of the two vectors all equal $0$. Therefore, the majorization condition $\bb_{\max} \prec \bb$ can be simplified to $\bb_{\max}' \prec \bb'$, where
	\begin{equation}
		\bb_{\max}' = \underbrace{(1/2, \ldots, 1/2)}_{2R \text{ times}}
	\end{equation}
	and $\bb'$ contains the $2R$ nontrivial entries of $\bb$ (with $R$ given by Eq.~\eqref{eq:R}). The simplified majorization condition holds because, by Prop.~\ref{prop:bset_bipartitions_SD},
	\begin{equation}
		\sum_{j=1}^{2R} b'_{\max,j} = \sum_{j=1}^{2R} b'_{j} = R \; ,
	\end{equation}
	and because $\bb_{\max}'$ is maximally mixed.
\end{proof}
The strict Schur-concavity~\cite{MarshallBook} of the function $H$ implies the following, by definition:
\begin{coro} \label{cor:bbest_SD_bipartitions}
	The vector $\bb_{\max}$ introduced in Cor.~\ref{cor:bbets_SDs} is the unique maximizer for Eq.~\eqref{eq:I_max_re_rewritten}. $C_{\max}$ follows as $2 R \log(2)$, with $R = \min \{ \lambda_1, \lambda_2, N, d - N \}$.
\end{coro}
We have thus retrieved Result~\ref{prop:Emax_SD_bipartition}. To fully complete the connection to that result, we now show that the maximally entangled SD introduced therein has as 1RDM a rank-$N$ projection matrix $\gamma$ whose block spectrum is indeed $\bb_{\max}$. For the sake of clarity, we consider the case $N = 2, d = 5, \lambda_1 = 3$ and $\lambda_2 = 2$. As prescribed by Result~\ref{prop:Emax_SD_bipartition}, the $N = 2$ filled orbitals $\ket{\phi_1}, \ket{\phi_2}$ constituting a maximally entangled SD $f^\dagger_{\phi_1}f^\dagger_{\phi_2} \ket{0}$ are then all balanced superpositions of a left and a right basis orbital, e.g., $\ket{\phi_j} = (\ket{j} + \ket{\lambda_1 + j})/\sqrt{2}$, for $j = 1, 2$. As a consequence, the 1P-reduced operator $\sum_{j=1}^2 \ket{\phi_j}\bra{\phi_j}$ of the SD is represented, in the reference basis, by the rank-$2$ projection matrix
\begin{equation} \label{eq:gamma_block_1/2_evals}
	\gamma = \begin{pmatrix}
		\begin{array}{ccc|cc}
			1/2 & 0 & 0 & 1/2 & 0 \\
			0 & 1/2 & 0 & 0 & 1/2  \\
			0 & 0 & 0 & 0 & 0 \\
			\hline
			1/2 & 0  & 0 & 1/2 & 0  \\
			0 & 1/2  & 0 & 0 & 1/2  \\
		\end{array}
	\end{pmatrix} ~ .
\end{equation}
The highlighted block structure refers to the splitting of the reference basis of $\H_1$ into left and right basis orbitals. For the chosen values of $N, d, \lambda_1$ and $\lambda_2$, the numbers of trivial eigenvalues of the two diagonal blocks are $T_1^P = T_2^P = T_2^H = 0$ and $T_1^H = 1$, and indeed the upper-left block $B_1^\gamma$ has one eigenvalue $0$. All other $R + R = 2 + 2$ block eigenvalues equal $1/2$. In other words, the block spectrum of $\gamma$ is given by $\bb = (1/2, 1/2, 0 \sbs 1/2, 1/2) = \bb_{\max}$, as expected. This simple calculation can be extended in a straightforward manner to arbitrary values of $N, d, \lambda_1$ and $\lambda_2$, thus establishing the equivalence of Result~\ref{prop:Emax_SD_bipartition} and Cor.~\ref{cor:bbest_SD_bipartitions}. Similarly, at the level of matrices, rank-$N$ projection matrices $\gamma$ given by a simple generalization of Eq.\til\eqref{eq:gamma_block_1/2_evals} are (non-unique) maximizers for Eq.~\eqref{eq:I_max_rewritten}.

\section{Schur-Horn theorem and Cauchy interlacing inequalities} \label{app:majorization_schur_horn_cauchy}
We recall two well-known theorems and their consequences on the characterization of the set $\bset^{\blam}_{\bn}$ discussed in Sec.~\ref{sec:math_struc_gen_prob}.

First, we introduce the concept of majorization~\cite{MarshallBook,nielsenmajo}. Given $x, y \in \R^d$, $x$ majorizes $y$, written $x \succ y$, if and only if
	\begin{equation} \label{eq:majorization_eq}
		\sum_{j = 1}^d x_j = \sum_{j = 1}^d y_j
	\end{equation}
	and
	\begin{equation} \label{eq:majorization_ineq}
		\sum_{j = 1}^k x^{\downarrow}_j \geq \sum_{j = 1}^k y^{\downarrow}_j \, , \;\;\; k = 1, \ldots, d - 1 \; .
	\end{equation}
\begin{thrm} \label{thrm:schur-horn}
	\emph{(Schur-Horn~\cite{schur1923klasse, Horn1954DoublySM}).} Let $d \geq 2$, $\bn = \bn^\downarrow \in \R^d, \bb \in \R^d$ and $\blam = (1, \ldots, 1)$. Then $\bb$ is the diagonal of a Hermitian matrix with spectrum $\bn$, i.e., $\bb \in \bset^{\blam}_{\bn}$, if and only if $\bn \succ \bb$.
\end{thrm}
\begin{thrm} \label{thrm:interlacing}
	\emph{(Cauchy interlacing inequalities~\cite{Fan1957Imbedding}).} Let $d \geq 2$, $\bn = \bn^\downarrow \in \R^d$, $\tilde{\bb} = \tilde{\bb}^\downarrow \in \R^{d-1}$ and $\blam = (d-1, 1)$. Then $\tilde{\bb}$ is the spectrum of a $(d-1) \times (d-1)$ principal submatrix of a Hermitian matrix with spectrum $\bn$, i.e., $\bb = (\tilde{\bb} \sbs \sum_{j=1}^d n_j - \sum_{j=1}^{d-1} \tilde{b}_j ) \in \bset^{\blam}_{\bn}$, if and only if $\tilde{\bb}$ interlaces $\bn$, i.e.,
	\begin{equation} \label{eq:interlacing_ineq}
		n_1 \geq \tilde{b}_1 \geq n_2 \geq \tilde{b}_2 \geq \ldots \geq \tilde{b}_{d-1} \geq n_d \; .
	\end{equation}
\end{thrm}
Thm.~\ref{thrm:schur-horn} has the following quite general implication.
\begin{prop}
	Let $d \geq 2$, $ \blam$ a partition of $d$ and $\bn = \bn^\downarrow \in \R^d$. Let $\eta = \sum_{j=1}^d n_j / d$. Pick $\bb \in \bset^{\blam}_{\bn}$. Then
	\begin{equation} \label{eq:basic_maj_relations}
	\bn \succ \bb \succ (\underbrace{\eta, \ldots, \eta}_{\textup{$d$ times}}) \; .
	\end{equation}
\end{prop}
\begin{proof}
	Let $\gamma$ be a $d \times d$ Hermitian matrix with spectrum $\bn$ and $\blam$-block spectrum $\bb$. Conjugating $\gamma$ with a block-diagonal unitary $u$ yields a matrix $\gamma' = u \gamma u^\dagger$ with the same spectrum $\bn$ and block spectrum $\bb$ as $\gamma$. Also, a suitable choice of the diagonal blocks of $u$ ensures that the diagonal blocks of $\gamma'$ are diagonal, implying that $\bb$ is the diagonal of $\gamma'$ (up to a permutation of the components). By Thm.~\ref{thrm:schur-horn}, we get $\bn \succ \bb$. In particular, we get $\sum_{j=1}^d b_j = \sum_{j=1}^d n_j = \sum_{j=1}^d \eta$. The relation $\bb \succ (\eta, \ldots, \eta)$ then follows from the fact that the average of the $k$ largest components of $\bb$ is higher than or equal to the average $\eta$ of all $d$ components of $\bb$.
\end{proof}
The above proposition implies the basic bounds on the maximal orbital correlation $C_{\max}^{\blam}(\rho)$ of free states $\rho$ stated in Prop.~\ref{prop:basic_bounds}, which we copy here.
% previously restatable \propbasicbounds*
\begin{prop*}
	Let $\rho$ be a fermionic free state with NONs $n_1, \ldots, n_d$ and filling fraction $\eta = \sum_{j=1}^d n_j/d$. Then
	\begin{equation}
		0 \leq C_{\max}^{\blam}(\rho) \leq d \cdot h ( \eta ) - H(\bn) \; .
	\end{equation}
	In particular, the upper bound is saturated if and only if $\bset^{\blam}_{\bn}$ contains the maximally mixed element $(\eta, \ldots, \eta)$.
\end{prop*}
\begin{proof}
	We use the representation of $C_{\max}^{\blam}(\rho)$ from Prop.\til\ref{prop:Imax_max_over_bset}. The first (second) bound of Eq.\til\eqref{eq:basic_bounds} follows from the first (second) majorization relation in Eq.\til\eqref{eq:basic_maj_relations}, because the function $H$ is Schur-concave.
	
	If $(\eta, \ldots, \eta) \in \bset^{\blam}_{\bn}$, the upper bound is clearly saturated. The opposite implication follows from the \emph{strict} concavity of $h$. In fact, the function $H$ with domain $\tilde{\bset} = \{ \bb \in [ 0, 1 ]^d | \sum_{j} b_j = \sum_{j} n_j \}$ strictly increases under nontrivial convex combinations of its input variables. This implies that the maximally mixed vector $(\eta, \ldots, \eta)$ is its unique maximizer over $\tilde{\bset}$, of which $\bset^{\blam}_{\bn}$ is a subset.
\end{proof}

\section{Bosonic free states} \label{app:bosons}
We show that the general analysis of Sec.~\ref{sec:generalizing} and~\ref{sec:math_struc_gen_prob} and our results of Sec.~\ref{sec:solutions_general} for fermionic free states can be directly adapted to bosonic free states, i.e., (mixed) bosonic number-conserving Gaussian states. We stress that, in contrast to fermionic free states, such bosonic states are separable with respect to any partition of the mode space~\cite{wolf_ent_power, ciancio_zanardi}. For bosons, the maximal orbital total correlation computed below therefore contains no entanglement.

We use the same notation as for fermions. The key observation is the following: maximizing the orbital multipartite total correlation over bosonic free states with fixed NONs $\bn = (n_1, \ldots, n_d)$, $\infty > n_1 \geq \ldots \geq n_d \geq 0$, boils down to the maximization over the set $\bset^{\blam}_{\bn}$ (as defined in Def.~\ref{def:b_set}) of the function 
\begin{equation}
	\tilde H(\bb) := \sum_{j=1}^d [ (b_j + 1) \log (b_j + 1) - b_j \log b_j ]\; .
\end{equation}
More precisely, Prop.~\ref{prop:Imax_max_over_bset} holds true for bosons, up to modifying $H \rightarrow \tilde H$. Crucially, $\tilde H$ is strictly Schur-concave, just as its fermionic counterpart $H$. Also, for each 1P space dimension $d$ and partition $\blam$ that we investigated, there is $\forall \, \bn \in \R^{d, \downarrow}$ a unique element $\bb_{\max} \in \bset^{\blam}_{\bn}$ that is majorized by any other $\bb \in \bset^{\blam}_{\bn}$ and is thus the unique maximizer of any Schur-concave function defined on $\bset^{\blam}_{\bn}$---see Tables\til\ref{tab:tabledgen},\til\ref{tab:tabled2},\til\ref{tab:tabled3},\til\ref{tab:tabled4}, and\til\ref{tab:tabled5}. The bottom line is that the results of those tables apply to bosons as well as to fermions. The only difference is the numerical value of the maximal total correlation: $H(\bb_{\max}) - H(\bn)$ for a fermionic $\bn$ (see Eq.~\eqref{eq:obtain_Imax_from_bmax}), and $\tilde H(\bb_{\max}) - \tilde H(\bn)$ for a bosonic $\bn$.

We now briefly revisit the derivation of Prop.~\ref{prop:Imax_max_over_bset}, adapted to bosons. We omit technical details. The definition of $C^{\blam}$ and $C^{\blam}_{\max}$ are unchanged for bosons, and so is Prop.~\ref{prop:marginalsoffreestates} for computing the orbital-reduced states of a free state. Instead, an adaptation of Prop.~\ref{prop:entropyfreestates} for computing the entropy of a free state is in order. Given a bosonic free state $\rho$ over $d$ modes with NONs (i.e., 1P-reduced operator eigenvalues) $n_1, \ldots, n_d$, the following standard calculation yields $S(\rho)$.
Let $b^{\dagger}_1, \ldots, b^{\dagger}_d$ be bosonic creation operators relative to an ONB of `normal' modes in the 1P space such that $\rho \propto \exp ( - \sum_{j = 1}^d \epsilon_j b^\dagger_j b_j)$. Then $\rho$ can be decomposed as a product state, where the tensor product structure refers to the normal modes~\cite{zanardi,ciancio_zanardi}:
\begin{equation}
	\rho = \bigotimes_{j = 1}^d \frac{\exp ( - \epsilon_j b^\dagger_j b_j)}{\tr [ \exp ( - \epsilon_j b^\dagger_j b_j) ]} = \bigotimes_{j = 1}^d \rho_j \; .
\end{equation}
One can show that the NON $n_j = \tr [ \rho b^\dagger_j b_j ] $ is related to the 1P energy $\epsilon_j$ via the Bose-Einstein formula
\begin{equation}
	n_j = \frac{1}{\exp(\epsilon_j) - 1} \; .
\end{equation}
By inverting this formula and using the additivity of the entropy of product states, one eventually gets
\begin{equation}
	S(\rho) = \sum_{j = 1}^d S(\rho_j) = \sum_{j = 1}^d \left [ (n_j+1) \log (n_j+1) - n_j \log n_j \right ] \; ,
\end{equation}
i.e.,
\begin{equation}
	S(\rho) = \tilde H (\bn) \; .
\end{equation}
Using this equation, the analysis of Sec.~\ref{sec:generalizing} yields the following bosonic adaptation of the result of Prop.~\ref{prop:Imax_max_over_bset}:
\begin{equation}
	C^{\blam}_{\max} (\rho) = \max_{\bb \in \bset^{\blam}_{\bn}} \left ( \tilde H(\bb) \right ) - \tilde H(\bn) \; \; \; \text{for bosons.}
\end{equation}
The analysis of Sec.~\ref{sec:math_struc_gen_prob} is unchanged.

\section{Regimes} \label{app:regimes}
We now clarify the notion of regimes, which first appears in the case $d=4$, $\blam=(2,1,1)$. There, a first example was considered where the function $\bn \mapsto \bb_{\max}^{\blam} (\bn)$ is piecewise-defined, i.e., defined by cases. We referred to the cases as regimes. To be slightly more precise, but without the aim of a fully rigorous analysis, we may define regimes along the following lines. For given $d$ and $\blam$, we call regimes the subsets into which the $\bn$-domain
\begin{equation}
	\text{domain}(\bb_{\max}^{\blam}) = \{ \bn \in [0,1]^d \text{ s.t. } n_1 \geq \ldots \geq n_d  \}
\end{equation}
gets partitioned based on the distinct functional form of $\bb_{\max}^{\blam} (\bn)$. For example, for $d=4$ and $\blam = (2,1,1)$, based on Eq.~\eqref{eq:bbest_2_1_1} the regimes are $\{ \bn \in [0,1]^4 \text{ s.t. } n_1 \geq \ldots \geq n_4, n_1 \geq \sum_j n_j/4 \geq n_2 \}, \ldots$ . We represent these regimes in Fig.~\ref{fig:regimes}. There, we fix $n_1$ and $n_4$ and show the corresponding two-dimensional slice of the $\bn$-domain in the $n_2,n_3$-plane. The specific values of $n_1$ and $n_4$ are not relevant, as changing them only results in an isotropic dilation and a shift of the shape of Fig.~\ref{fig:regimes}. This is due to the invariance of the underlying mathematical problem under any affine transformation of the type $n_j \rightarrow n'_j = \alpha n_j + \beta$ (with $\alpha, \beta \in \R$).
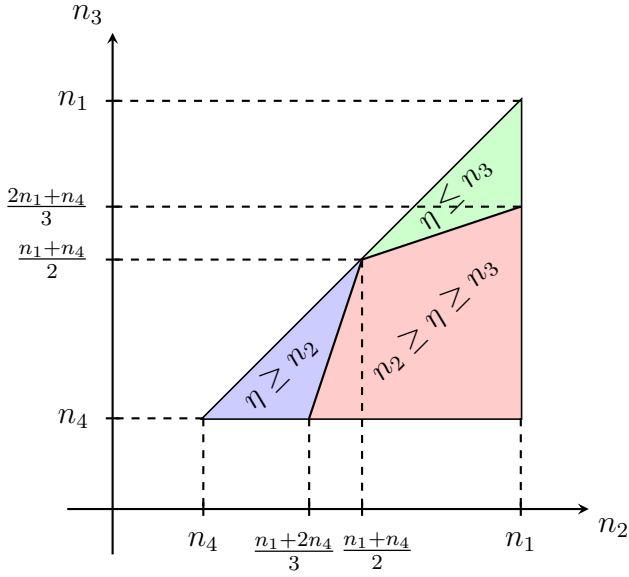
\begin{figure}
	
	\begin{center}
\begin{tikzpicture}[
    x=6cm, y=6cm, 
    font=\sffamily\large,
    axes/.style={thick,->,>=stealth},
    mainline/.style={very thick}
]

% Faint grid mimicking the graph paper (1 square = 0.05 units)
% \draw[step=0.05, blue!30, very thin] (-0.1, -0.1) grid (1.05, 1.05);

% Axes
\draw[axes] (-0.1, 0) -- (1.05, 0) node[below right] {$n_2$};
\draw[axes] (0, -0.1) -- (0, 1.05) node[above left] {$n_3$};

% Tick marks on x-axis
\draw[thick] (0.2, 0.015) -- (0.2, -0.015) node[below, yshift=-1mm] {$n_4$};
\draw[thick] (0.55, 0.015) -- (0.55, -0.015) node[below, yshift=-1mm, xshift =+2mm] {$\frac{n_1 + n_4}{2}$};
\draw[thick] (0.9, 0.015) -- (0.9, -0.015) node[below, yshift=-1mm] {$n_1$};
\draw[thick] (0.43333, 0.015) -- (0.43333, -0.015) node[below, yshift=-1mm, xshift =-2mm] {$\frac{n_1 + 2 n_4}{3}$};

% Tick marks on y-axis
\draw[thick] (0.015, 0.2) -- (-0.015, 0.2) node[left, xshift=-1mm] {$n_4$};
\draw[thick] (0.015, 0.55) -- (-0.015, 0.55) node[left, xshift=-1mm] {$\frac{n_1 + n_4}{2}$};
\draw[thick] (0.015, 0.9) -- (-0.015, 0.9) node[left, xshift=-1mm] {$n_1$};
\draw[thick] (0.015, 0.66667) -- (-0.015, 0.66667) node[left, xshift=-1mm] {$\frac{2 n_1 + n_4}{3}$};

% Coordinate Definitions based on your prompt
\coordinate (Origin) at (0.2, 0.2);
\coordinate (TopRight) at (0.9, 0.9);
\coordinate (BotRight) at (0.9, 0.2);
\coordinate (A) at (0.43333, 0.2);
\coordinate (B) at (0.55, 0.55);
\coordinate (C) at (0.9, 0.66667);

% Draw the main outer triangle boundaries
\draw[mainline] (Origin) -- (TopRight) -- (BotRight) -- cycle;

% Light background fills for the three regions
\fill[blue!20] (Origin) -- (A) -- (B) -- cycle;
\fill[green!20] (B) -- (C) -- (TopRight) -- cycle;
\fill[red!20] (A) -- (BotRight) -- (C) -- (B) -- cycle;

% Draw the interior lines dividing the regions
\draw[thick] (A) -- (B);
\draw[thick] (B) -- (C);

% Dashed projection lines to the axes
\draw[dashed, thick] (0.9, 0.2) -- (0.9, 0);
\draw[dashed, thick] (0.9, 0.9) -- (0, 0.9);
\draw[dashed, thick] (Origin) -- (0.2, 0);
\draw[dashed, thick] (Origin) -- (0, 0.2);
\draw[dashed, thick] (A) -- (0.43333, 0);
\draw[dashed, thick] (B) -- (0.55, 0);
\draw[dashed, thick] (B) -- (0, 0.55);
\draw[dashed, thick] (C) -- (0, 0.66667);

% Labels for points A, B, C
%\node[below, yshift=-2mm] at (A) {$A$};
%\node[above left] at (B) {$B$};
%\node[right, xshift=1mm] at (C) {$C$};

% Text defining the mathematical regions (rotated to roughly match the line slopes)
\node[rotate=45] at (0.37, 0.3) {$\eta \geq n_2$};
\node[rotate=45] at (0.755, 0.695) {$\eta \leq n_3$};
\node[rotate=45] at (0.71, 0.42) {$n_2 \geq \eta \geq n_3$};

\end{tikzpicture}
\end{center}
	\caption{A slice of the $\bn$-domain for $d = 4$ and $\blam = (2,1,1)$, for fixed $n_1$ and $n_4$. The blue, red, and green regions are slices of the first, second, and third regimes, respectively, as shown in Eq.~\eqref{eq:bbest_2_1_1}. We emphasize the symmetry of the figure with respect to the line $n_2 + n_3 = n_1 + n_4$.}
	\label{fig:regimes}
	
\end{figure}

To validate the concept of regimes and our formulas for $\bb_{\max}^{\blam} (\bn)$, see Tables~\ref{tab:tabled4} and~\ref{tab:tabled5}, we checked for $d \leq 5$ and for each given partition $\blam$ the following:
\begin{itemize}
	\item Each regime has non-zero measure;
	\item The partition into regimes is complete, i.e., the union of the regimes is the whole $\bn$-domain;
	\item Regimes intersect only on their boundaries, i.e., the intersection of any two regimes is contained in the intersection of their boundaries (and thus has zero measure);
	\item The formulas for $\bb_{\max}^{\blam}(\bn)$ agree on the intersection of any two regimes; incidentally, this implies $\bb_{\max}^{\blam}$ is a continuous function.
\end{itemize}
In view of the last point, in the Appendices we harmlessly regard $\bn$ vectors that are on the boundary of two neighboring regimes as elements of both regimes. Note, for example, the irrelevant differences in the regime definitions of Eqs.~\eqref{eq:bbest_2_1_1} and~\eqref{eq:bbest_2_1_1_repeated} (in the former, some inequalities are strict).

\section{Symmetries of the general problem} \label{app:symmetries}

We observe a symmetry of the general problem of Sec.~\ref{sec:generalizing}, namely the particle-hole symmetry of fermionic systems\til\cite{Zirnbauer_2021}. Consider two free states that are related by a particle-hole transformation, i.e., they have the same natural orbitals, and NON vectors $ \bn = (n_1, \ldots, n_d)$ and $ \mathbf{1} - \bn = (1-n_1, \ldots, 1-n_d)$, respectively. These two states are equivalent from the point of view of correlations. More precisely, for given $\blam$, we have $\bb^{\downarrow_B} \in \bset^{\blam}_{\bn^\downarrow} \Leftrightarrow \mathbf{1} - \bb^{\uparrow_B} \in \bset^{\blam}_{\mathbf{1} -\bn^\uparrow}$, where the superscripts $\uparrow_B,\downarrow_B$ indicate blockwise ordering. Furthermore, the function $H$ to be maximized over $\bset^{\blam}_{\bn^\downarrow}$/$\bset^{\blam}_{\mathbf{1}-\bn^\uparrow}$ is invariant under $\bb \rightarrow \mathbf{1}-\bb$. Therefore, the instances $d, \blam, \bn^\downarrow$ and $d, \blam, \mathbf{1} - \bn^\uparrow$ of our maximization problem are equivalent, and the optimal vector $\bb^{\blam}_{\max}(\bn^\downarrow)$ for one instance will mirror the one $\bb^{\blam}_{\max}(\mathbf{1}-\bn^\uparrow)$ for the other. Indeed, this symmetry property can be traced throughout the formulas for $\bb_{\max}$ in Tables\til\ref{tab:tabledgen},\til\ref{tab:tabled2},\til\ref{tab:tabled3},\til\ref{tab:tabled4}, and\til\ref{tab:tabled5}.

This symmetry observation can be in part adapted to the bosonic case.

Another symmetry of our problem concerns the case of any partition $\blam$ with a repetition. Whenever $\lambda_\mu = \lambda_\nu$, then the corresponding blocks of $\bb_{\max}^{\blam}(\bn)$ coincide for any $\bn$. This follows from the invariance of $\bset^{\blam}_{\bn} $ and of $H$ under swapping equal-size blocks, the convexity of $\bset^{\blam}_{\bn} $, and the strict concavity of $H$.

\section{Trivial partitions and arbitrary parity-superselected states} \label{app:finest_partition_arbitrary_states}

As anticipated in Sec.\til\ref{subsec:2D_space} and Sec.~\ref{subsec:trivial_partitions}, we show that our results for trivial partitions $\blam = (1, \ldots, 1)$ can be extended to \emph{arbitrary parity-superselected states}. These are the states on the Fock space that commute with the number parity operator~\cite{Friis_2016}.

Let $\rho$ be one such state, let $\bn$ be its NONs, and denote with $\eta$ their average. Then
\begin{equation}\label{eq:I_max_finest_partition_general_PSSR}
	C^{\blam}_{\max}(\rho) = d \, h(\eta) - S(\rho) \; .
\end{equation}
If $\rho$ is free, then $S(\rho) = H(\bn)$ and the above equation reduces to Eq.~\eqref{eq:I_max_finest_partition}. To prove Eq.~\eqref{eq:I_max_finest_partition_general_PSSR}, the key quantity is the single-orbital reduced density operator $\sigma_j$ of a parity-superselected state $\sigma$ relative to the $j$th lattice site. A simple computation yields
\begin{equation}
	\sigma_j \rightarrow
	\begin{pmatrix}
		1 - \gamma_{jj} & 0 \\
		0 & \gamma_{jj}
	\end{pmatrix} \; ,
\end{equation}
where the matrix is expressed in the empty–occupied orbital basis, and $\gamma$ is the 1RDM of $\sigma$ expressed in the lattice site basis. The off-diagonal terms vanish because $\sigma$ commutes with the parity operator. This allows us to compute the entropy of $\sigma_j$ as we would if $\sigma$ were a free state:
\begin{equation}
	S(\sigma_j) = h(\gamma_{jj}) \; .
\end{equation}
Then,
\begin{equation}
	C^{\blam}(\sigma) = \left ( \sum_{j=1}^d S(\sigma_j) \right ) - S(\sigma) = \left ( \sum_{j=1}^d h(\gamma_{jj}) \right ) - S(\sigma) .
\end{equation}
Finally, we use this equation for $\sigma = U \rho U^\dagger$ and maximize $C^{\blam}(U \rho U^\dagger)$ over one-particle unitaries $U$ to get $C_{\max}^{\blam}(\rho)$. The maximization is the same as for free states, and yields Eq.~\eqref{eq:I_max_finest_partition_general_PSSR}. As for free states, optimal one-particle unitaries $U$ are those that ensure that $U \rho U^\dagger$ has discrete plane waves as natural orbitals.

\section{1P space dimension $d = 3$} \label{app:3D_space}
We explicitly solve the problem of free state maximal orbital correlation for $d=3$. The partitions of 3 (other than $\blam = (3)$) are $(2,1)$ and $ (1,1,1)$. These are, respectively, a bipartition and a trivial multipartition, and thus represent a special case of Sec.~\ref{subsec:bipartitions} and~\ref{subsec:trivial_partitions}, respectively.

\subsection{\texorpdfstring{$\boldsymbol{\lambda}$}{lambda} $ = (2,1)$}
For $\blam = (2,1)$ and given NONs $1 \geq n_1 \geq n_2 \geq n_3 \geq 0$, we can characterize $\bset^{\blam}_{\bn}$ using the Cauchy interlacing inequalities of Thm.\til\ref{thrm:interlacing}. More in detail, it is convenient to parametrize the elements $\bb = (b_1, b_2 \sbs b_3) \in \bset^{\blam}_{\bn}$ by the first two components $b_1, b_2$. $b_3$ follows from $\sum_{j=1}^3 b_j = \sum_{j=1}^3 n_j$. We thus define
\begin{align} \nonumber
	\tilde{\bset}^{\blam}_{\bn} = \{ &\specdown(B_1^\gamma) \sbs \gamma = \gamma^\dagger \in \text{Mat}_{3 \times 3}(\C) \, , \\
	&\specdown(\gamma) = \bn \} \; ,
\end{align}
where $B_1^\gamma$ is the upper-left $2 \times 2$ principal submatrix of $\gamma$. $\tilde{\bset}^{\blam}_{\bn}$ is in 1-to-1 correspondence with $\bset^{\blam}_{\bn}$. By Thm.\til\ref{thrm:interlacing}, it is the rectangle 
\begin{equation} \label{eq:bset_lamb_21}
	\tilde{\bset}^{\blam}_{\bn} = \{ (b_1, b_2) \in \R^2 \sbs n_1 \geq b_1 \geq n_2 \geq b_2 \geq n_3 \} \; , 
\end{equation}
represented in Fig.\til\ref{fig:bset_d=3_bip}.
\begin{figure}
	
	\begin{center}
\begin{tikzpicture}[scale=2.5]
    % Draw axes
    \draw[->] (-0.3,0) -- (2.2,0) node[below] {$b_1$};
    \draw[->] (0,-0.3) -- (0,1.5) node[left] {$b_2$};
    \draw[dashed] (0,0) -- (1.5,1.5);
    
    % Draw rectangles
    \fill[red!20]  (1.1,.2) rectangle (1.4,1.1);
    
    % Label coordinates
    \draw (1.1,0) node[below] {$\, \, n_2$} -- ++(0,0.07);
    \draw (1.4,0) node[below] {$\, \, n_1$} -- ++(0,0.07);

    \draw (0,.2) node[left] {$n_3$} -- ++(0.07,0.);
    \draw (0,1.1) node[left] {$n_2$} -- ++(0.07,0.);
    \draw (0,0.8) node[left] {$\frac{n_1 + n_3}{2}$} -- ++(0.07,0.);
    
    % Label rectangle 
    \draw[red] (1.65,0.65) node {$\tilde{\bset}^{\blam}_{\bn}$};
    
    % maximum
    \fill[red] (1.1, .8) circle (1pt) node[below left] {$\tilde{\bb}_{best}$};
    
    % other things...

%    \draw (0,.9) node[left] {$\eta$} -- ++(0.07,0.);   
%    \draw (.9,0.) node[below] {$\eta \,$} -- ++(0,0.07);  

\end{tikzpicture}
\end{center}
	\caption{The set $\tilde{\bset}^{\blam}_{\bn}$, a projection of the polytope $\bset^{\blam}_{\bn}$ of block spectra for $\blam = (2,1)$. The point $\tilde{\bb}_{\max}$ corresponding to the $\succ$-least element $\bb_{\max} \in \bset^{\blam}_{\bn}$ is highlighted.}
	\label{fig:bset_d=3_bip}
	
\end{figure}

We now show that there is an element $\bb_{\max} \in \bset^{\blam}_{\bn}$ that is majorized by any other $\bb \in \bset^{\blam}_{\bn}$, namely
\begin{equation} \label{eq:bbset_lamb_21}
	\boxed{\bb_{\max} = \left ( \frac{n_1 + n_3}{2} , n_2 \sbs \frac{n_1 + n_3}{2} \right )}
\end{equation}
\newline \newline
up to decreasingly reordering the first two entries. As for $d=2$ and $\blam = (1,1)$, the components of $\bb_{\max}$ are averages of $\bn$ entries, as will always be the case. The corresponding element $\tilde{\bb}_{\max} \in \tilde{\bset}^{\blam}_{\bn}$ is highlighted in Fig.\til\ref{fig:bset_d=3_bip}.

First, let us show that $\bb_{\max} \in \bset^{\blam}_{\bn}$, i.e., we exhibit a matrix $\gamma = \gamma^\dagger$ with spectrum $\bn$ and block spectrum $\bb_{\max}$. An example is
\begin{equation} \label{eq:best_matrix_2,1}
	\gamma =
\begin{pmatrix}
	\begin{array}{cc|c}
		\frac{n_1 + n_3}{2} & 0 & \frac{n_1 - n_3}{2} \\
		
		0 & n_2 & 0 \\
		\hline
		\frac{n_1 - n_3}{2} &  0  & \frac{n_1 + n_3}{2}  \\
	\end{array} 
\end{pmatrix} \; .
\end{equation}
Following the same ideas of Sec.\til\ref{subsec:2D_space}, $\gamma$ can be obtained from $\text{diag}(n_1, n_2, n_3)$ by conjugation with the unitary matrix implementing $\mathrm{DFT}_2$ on $\text{span}(\be_1, \be_3)$. Here and throughout the article, $\be_j$ denotes the $j$th element of the canonical basis of $\C^d$. In terms of fermionic orbitals, this corresponds to transforming $\ket{1}$ and $\ket{3}$ to $(\ket{1} \pm \ket{3})/\sqrt{2}$.

Next, we show that
\begin{equation}
	\bb_{\max} \prec \bb
\end{equation}
for any given $\bb \in \bset^{\blam}_{\bn}$. We may consider only the case
\begin{equation}
	n_2 \geq \frac{n_1 + n_3}{2} \; .
\end{equation}
In fact, the case $n_2 < \frac{n_1 + n_3}{2}$ can be reduced to the other case by substituting $\bn \rightarrow \mathbf{1} - \bn^\uparrow$ (see Sec.\til\ref{sec:math_struc_gen_prob} for more on this symmetry argument). With this ordering assumption, the inequalities\til\eqref{eq:majorization_ineq} defining the relation $\bb_{\max} \prec \bb$ read
\begin{align}
	n_2 &\leq \max_j b_j \\
	n_2 + \frac{n_1 + n_3}{2} &\leq \max_{j\neq k}( b_j + b_k ) \; .
\end{align}
The first inequality directly follows from the interlacing inequalities of Eq.\til\eqref{eq:bset_lamb_21}. These also allow us to prove the second inequality, as follows:
\begin{align} \nonumber
n_2 + \frac{n_1 + n_3}{2} &= \frac{1}{2}\left(n_2 + \sum_{j=1}^3 n_j\right)	\\ \nonumber
&\leq \frac{1}{2}\left(b_1 + \sum_{j=1}^3 b_j\right)\\ \nonumber
&= b_1 + \frac{b_2+b_3}{2} \\
&\leq \max_{j\neq k}( b_j + b_k ) \; .
\end{align}
This concludes the mathematical treatment for the case $d=3$ and $\blam = (2,1)$, a good example for the techniques we will use throughout the following sections. In particular, we generalize the present case to arbitrary bipartitions $\blam$ of any space dimension $d$ in Sec.\til\ref{subsec:bipartitions}.

Returning to the physical interpretation, for any free state $\rho$ with NONs $\bn$, the maximal orbital correlation $C_{\max}^{\blam}(\rho)$ follows from Eq.\til\eqref{eq:bbset_lamb_21} using Eq.\til\eqref{eq:obtain_Imax_from_bmax}.

\subsection{\texorpdfstring{$\boldsymbol{\lambda}$}{lambda} $ = (1,1,1)$}
For $\blam = (1,1,1)$, the analysis is analogous to that for $d = 2$ and $\blam = (1,1)$ (see Sec.\til\ref{subsec:2D_space}). For given $1 \geq n_1 \geq n_2 \geq n_3 \geq 0$, the set $\bset_{\bn}^{\blam}$ contains the maximally mixed element
\begin{equation}
	\boxed{\bb_{\max} = \left ( \eta \sbs \eta \sbs \eta \right )} \; ,
\end{equation}
where $\eta = \sum_{j=1}^3 n_j/3$. In fact, observe that, e.g., the Hermitian matrix
\begin{equation}
	\mathrm{DFT}_3^\dagger \cdot \text{diag}(n_1, n_2, n_3) \cdot \mathrm{DFT}_3 \; ,
\end{equation}
with
\begin{equation}
	\mathrm{DFT}_3 =
	\frac{1}{\sqrt{3}}
	\begin{pmatrix}
		1 & 1 & 1 \\
		1 & \text{e}^{- 2 \pi \text{i}/3} & \text{e}^{2 \pi \text{i}/3} \\
		1 & \text{e}^{ 2 \pi \text{i}/3} & \text{e}^{- 2 \pi \text{i}/3}
	\end{pmatrix} \; ,
\end{equation}
has spectrum $\bn$ and diagonal $\bb_{\max}$. Alternatively, one can check $\bb_{\max} \in \bset^{\blam}_{\bn}$ using the full characterization of $\bset^{\blam}_{\bn}$ given by the Schur-Horn Theorem\til\ref{thrm:schur-horn}. We represent a projection of $\bset^{\blam}_{\bn}$ in Fig.\til\ref{fig:bset_d=3_trip}.

Fig.\til\ref{fig:bset_d=3_trip} also highlights the inclusion relation $\bset^{\blam'}_{\bn} \subset \bset^{\blam}_{\bn}$, with $\blam' = (2,1)$. In fact, a simple argument shows that an analogous relation actually holds true for any two partitions $\blam$, $\blam'$ of any given $d$ such that $\blam$ is a refinement of $\blam'$, i.e., $\blam$ can be obtained from the partition $\blam'$ by further partitioning some of the integers $\lambda'_j$. 

\begin{figure}
	
	\begin{center}
	\begin{tikzpicture}[scale=2.5]
		% remember: n = (1.4, 1.1, 0.2)
		
		% Draw axes
		\draw[->] (-0.3,0) -- (2.2,0) node[below] {$b_1$};
		\draw[->] (0,-0.3) -- (0,1.5) node[left] {$b_2$};
		
		%pentagon for (1,1,1)
		\fill[color=blue!10] (.65,.65) -- (1.25,1.25) -- (1.4,1.1) -- (1.4,.2) -- (1.1,.2) -- cycle;
		\draw[dashed] (0,0) -- (1.5,1.5);
		
		%label some vertices
		\fill[black] (.65,.65) circle (.5pt);
		\draw[black] (.4,.74) node[scale=0.6] {$(\tfrac{n_2+n_3}{2},\tfrac{n_2+n_3}{2})$};
		\fill[black] (1.25,1.25) circle (.5pt);
		\draw[black] (0.95,1.3) node[scale=0.6] {$(\tfrac{n_1+n_2}{2},\tfrac{n_1+n_2}{2})$};

		% Draw rectangles
		\fill[violet!20]  (1.1,.2) rectangle (1.4,1.1);
		
		% Label coordinates
		\draw (1.1,0) node[below] {$\, \, n_2$} -- ++(0,0.07);
		\draw (1.4,0) node[below] {$\, \, n_1$} -- ++(0,0.07);
		
		\draw (0,.2) node[left] {$n_3$} -- ++(0.07,0.);
		\draw (0,1.1) node[left] {$n_2$} -- ++(0.07,0.);
		
		% Label rectangle 
		\draw[red!50] (1.65,0.65) node {$\tilde{\bset}^{\blam'}_{\bn}$};
		
		% other things...
		
		\draw (0,.9) node[left] {$\eta$} -- ++(0.07,0.);   
		\draw (.9,0.) node[below] {$\eta \,$} -- ++(0,0.07);  
		\draw[blue] (.75,.35) node {$\tilde{\bset}^{\blam}_{\bn}$};
		\fill[blue] (.9,.9) circle (1pt) node[below] {$\;\;\;\;\; \tilde{\bb}_{best}$};
		
	\end{tikzpicture}
\end{center}

\begin{comment}

\begin{center}
\begin{tikzpicture}[scale=2.5]

    % Label rectangle 
    \draw[red] (.7,1.6) node {$\tilde{V}_w^Y$};

    % other things...

    \draw (0,.9) node[left] {$\eta$} -- ++(0.07,0.);   
    \draw (.9,0.) node[below] {$\eta \,$} -- ++(0,0.07);  
    \draw[blue] (.7,.3) node {$\tilde{V}_w^{Y'}$};
\fill[blue] (.9,.9) circle (1pt) node[below , font=\tiny] {$\; \, \; \; \; \; \;  \tilde{v}^{mm}$};

\end{tikzpicture}
\end{center}

\end{comment}
	\caption{The set $\tilde{\bset}^{\blam}_{\bn} = \{ (b_1, b_2) \, \text{s.t.} \, (b_1, b_2 \sbs \sum_{j=1}^3 n_j - \sum_{j=1}^2 b_j) \in \bset^{\blam}_{\bn} \}$ (blue), a projection of the polytope $\bset^{\blam}_{\bn}$ of block spectra for $\blam = (1,1,1)$. The point $\tilde{\bb}_{\max}$ corresponding to the maximally mixed element $\bb_{\max} \in \bset^{\blam}_{\bn}$ is highlighted. We also represent the subset $\tilde{\bset}^{\blam'}_{\bn}$ for $\blam' = (2,1)$ (red, faded), copied from Fig.\til\ref{fig:bset_d=3_bip}.}
	\label{fig:bset_d=3_trip}
	
\end{figure}

\section{Proof of Thm.~\ref{thrm:mikes_maj_rel} and generalization for unbalanced bipartitions}\label{app:proof_bipartition_lemma}

We provide a proof of Thm.~\ref{thrm:mikes_maj_rel} using the characterization of $\bset^{\blam}_{\bn}$ provided in Ref.~\cite{lipoon} in Sec.~\ref{subsec:mikes_thm_from_lipoon}, and prove a generalization of it from first principles in Sec.~\ref{subsec:mike_general}. The generalization allows us to prove that $\bb_{\max}$ as defined in Eq.~\eqref{eq:bbest_bip} is majorized by any other $\bb \in \bset^{\blam}_{\bn}$ for the case of unbalanced bipartitions, too.

\subsection{Proof of Thm.~\ref{thrm:mikes_maj_rel} using the Li-Poon inequalities} \label{subsec:mikes_thm_from_lipoon}

We aim to prove Eq.~\eqref{eq:mikes_maj_rel}. Since this relation is invariant under shifts, we assume without loss of generality that $\gamma$ is positive semi-definite, i.e., $n_i \geq 0 \, \forall i$. This will allow us to use the results of Ref.\til\cite{lipoon}. The equality $\sum_{j=1}^k \tilde{n}_j = \sum_{j=1}^k (\alpha_j + \beta_j)$ follows from $\sum_{j=1}^{2k} n_j = \tr (\gamma)$. We then have to prove the following inequalities: for $1 \leq q \leq k - 1$ and $1 \leq i_1 < i_2 < \ldots < i_q \leq k$
\begin{equation} \label{eq:target_ineq}
	\sum_{l = 1}^q (n_{i_l} + n_{2k + 1 - i_l}) \leq \sum_{i = 1}^q (\alpha_i + \beta_i) \; .
\end{equation}

First, we introduce some concepts and notations. Ref.~\cite{lipoon} gives a characterization of the possible triplets of decreasingly ordered spectra $(\balpha, \bbeta, n) \in \mathbb{R}^k \times \mathbb{R}^k \times \mathbb{R}^{2k}$ that are related as in Thm.~\ref{thrm:mikes_maj_rel}. Namely, the necessary and sufficient condition is that (i)
\begin{equation}
	\sum_{i = 1}^{2k} n_i = \sum_{i = 1}^k (\alpha_i + \beta_i)
\end{equation}
and, simultaneously, (ii) the inequality
\begin{equation} \label{eq:lipoon_ineq}
	\sum_{i \in J_0} n_i \leq \sum_{i \in J_1} \alpha_i + \sum_{i \in J_2} \beta_i \;\;\; (\alpha_i, \beta_i := 0 \text{ for } k + 1 \leq i \leq 2k)
\end{equation}
holds true for all triplets $(J_0, J_1, J_2)$ that satisfy a certain property \textbf{P}, detailed below. Here, $J_0, J_1, J_2$ are tuples of equal length $r \in \{ 1, \ldots, 2k - 1  \}$ of increasingly ordered indices, i.e.,
\begin{equation}
	J_0 = (i_1, \ldots, i_r)
\end{equation}
with $1 \leq i_1 < i_2 < \ldots < i_r \leq 2k$, and similar for $J_1$ and $J_2$. Note that applying the Li-Poon result is difficult because it is difficult to list all triplets $(J_0, J_1, J_2)$ that satisfy the property \textbf{P}.

Indeed, even stating the property \textbf{P} requires some effort. To that end, we define the $r$-tuple $\lambda(J)$ associated to a given index tuple $J = (i_1, \ldots, i_r)$ via
\begin{equation} \label{eq:def_lambda}
	\lambda(J) := (i_r - r, \ldots, i_1 -1)
\end{equation}
(see Ref.~\cite{lipoon}, Sec.~2). We will identify such non-increasing tuples of nonnegative integers with their associated partition, also denoted with $\lambda(J)$. By definition, a given triplet $(J_0, J_1, J_2)$ has the property \textbf{P} if and only if the partitions $\lambda(J_0),\lambda(J_1)$ and $\lambda(J_2)$ have a strictly positive Littlewood-Richardson coefficient $c^{\lambda(J_0)}_{\lambda(J_1), \lambda(J_2)} > 0$, i.e., if and only if it is possible to generate $\lambda(J_0)$ from $\lambda(J_1)$ and $\lambda(J_2)$ following certain rules. The Littlewood-Richardson coefficient and these rules are defined in Ref.~\cite{fulton2000}, Sec.~3, and we will assume the reader is familiar with it.

We now set out to prove the inequality~\eqref{eq:target_ineq} for given $1 \leq q \leq k - 1  $ and $1 \leq i_1 < i_2 < \ldots < i_q \leq k$.

Consider the $2q$-tuples
\begin{align}
	J_0 &:= (i_1, \ldots, i_q, 2k + 1 - i_q, \ldots, 2k + 1 - i_1) \\
	J_1 &:= J_2 := (1, \ldots, q, k+1, \ldots, k + q) \; .
\end{align}
With this, the inequality~\eqref{eq:lipoon_ineq} becomes precisely the same as the inequality~\eqref{eq:target_ineq} to be proven.

To show that this inequality holds true, we now check that $(J_0, J_1, J_2)$ satisfies the property \textbf{P}. First, using Eq.~\eqref{eq:def_lambda} we compute
\begin{align} \nonumber
	\lambda(J_0) = (&2k + 1 - i_1 - 2q, \ldots, 2k + 1 - i_q - (q+1), \\
	&i_q - q, \ldots, i_1 - 1) \; ,
\end{align}
i.e., for $1 \leq l \leq q$ we have
\begin{align}
	&\lambda(J_0)_l = 2k + 1 - i_l - (2q + 1 - l) \\
	&\lambda(J_0)_{2q + 1 - l} = i_l - l \; .
\end{align}
We get
\begin{equation} \label{eq:pairing_lambdaJ0}
	\lambda(J_0)_l + \lambda(J_0)_{2q + 1 - l} = 2(k - q) \; .
\end{equation}
Also, using $i_l - l \leq i_q - q \leq k - q$, we get
\begin{align} \label{eq:upper_rows_lambdaJ0}
	&\lambda(J_0)_l \geq k - q \\
	&\lambda(J_0)_{2q + 1 - l} \leq k - q \; .
\end{align}
Therefore, $\lambda(J_0)$ as a partition has the following structure: \\

\Yboxdim{1cm}
\begin{tikzpicture}[scale=.5]
	\tyng(0cm,0cm,10,9,9,7,6,6,6,4,4,4,3,1,1,0)
	
	% Side braces for height = q
	\draw [decorate,decoration={brace,amplitude=6pt},thick]
	(-0.2,-5.9) -- (-0.2, 1.) node[midway,xshift=-10pt,left] {\small $q$};
	
	\draw [decorate,decoration={brace,amplitude=6pt},thick]
	(-0.2,-13.) -- (-0.2, -6.1) node[midway,xshift=-10pt,left] {\small $q$};
	
	% Top braces for width = k - q
	\draw [decorate,decoration={brace,amplitude=6pt},thick]
	(0,1.2) -- (4.9,1.2) node[midway,above,yshift=4pt] {\small $k - q$};
	
	\draw [decorate,decoration={brace,amplitude=6pt},thick]
	(5.1,1.2) -- (10,1.2) node[midway,above,yshift=4pt] {\small $k - q$};
	
	\draw[<->](7.5,-2.5)--(8.5,-2.5)%
	to node[auto,swap, right]{$2(k - q)$ boxes}%
	(4.5,-9.5)--(3.5,-9.5);
	
	\draw[red, very thick] (0,1) -- (5,1);
	\draw[red, very thick] (0,1) -- (0,-6);
	
	\draw[red, very thick] (0,-6) -- (5,-6);
	\draw[red, very thick] (5,1) -- (5,-6);
\end{tikzpicture}

We highlighted in red the boxes of $\lambda(J_0)$ that must be present to abide by Eq.~\eqref{eq:upper_rows_lambdaJ0}. In this example, $k = 12, q = 7$ and
\begin{equation} \label{eq:our_indices}
	(i_1, \ldots, i_q) = (1,3,4,7,9,10,11) \; .
\end{equation}

Next, we compute
\begin{equation}
	\lambda(J_1) = \lambda(J_2) = (\underbrace{k - q, \ldots, k-q}_{q\text{ times}}, \underbrace{0, \ldots, 0}_{q\text{ times}}) \; .
\end{equation}
We see that 
\begin{equation}\label{eq:lambda_boxes_count}
	\sum_{l=1}^{2q}\lambda(J_0)_l = \sum_{l=1}^{2q} ( \lambda(J_1)_l + \lambda(J_2)_l) \; ,
\end{equation}
and that $\lambda(J_0)$, as a partition, contains $\lambda(J_1)$, i.e.,
\begin{equation} \label{eq:y_diag_inclusion}
	\lambda(J_0)_l \geq \lambda(J_1)_l \,, \;\;\; l  = 1, \ldots, 2q \; . 
\end{equation}
In the setting of the diagram depicted above, $\lambda(J_1)$ is the red rectangle. Eqs.~\eqref{eq:lambda_boxes_count} and~\eqref{eq:y_diag_inclusion} are necessary for the possibility of generating $\lambda(J_0)$ from $\lambda(J_1)$ and $\lambda(J_2)$, i.e., for $c^{\lambda(J_0)}_{\lambda(J_1), \lambda(J_2)} > 0$, i.e., for $(J_0, J_1, J_2)$ to satisfy the property \textbf{P}.

To see that the latter three equivalent facts indeed hold true, we now show one procedure to generate $\lambda(J_0)$ from $\lambda(J_1)$ and $\lambda(J_2)$. By that, we mean a procedure to fill the boxes of $\lambda(J_0) \smallsetminus \lambda(J_1)$ with natural numbers abiding by the conditions (i)-(iv), relative to $\lambda(J_1)$, listed in Ref.~\cite{fulton2000}, Sec.~3. By $\lambda(J_0) \smallsetminus \lambda(J_1)$ we mean the diagram obtained by removing from $\lambda(J_0)$ the boxes of $\lambda(J_1)$, e.g., the boxes outside the red rectangle in the picture above. Here is a picture for the procedure:\\

\begin{tikzpicture}[scale=.5]
	% gray boxes
	\Ylinecolour{black!30!white}
	\Ynodecolour{black!30!white}
	\tyoung(8cm,-3cm,r,\vdots,\vdots,q)
	
	% partition
	\Ylinecolour{black}
	\Ynodecolour{black}
	\tyoung(0cm, 0cm,
	\ \ \ \ \ 1\dots\dots\dots1,\ \ \ \ \ 2\dots\dots2,\ \ \ \ \ <\vdots>\ \ \vdots,\ \ \ \ \ \ \ ,\ \ \ \ \ \ ,\ \ \ \ \ \vdots,\ \ \ \ \ q,\ r\ \vdots,\ \vdots\ <q \!\! - \!\!\! 1>,\vdots \vdots\vdots q,<q \!\! - \!\!\! 2>q\dots,<q \!\! - \!\!\! 1>,q)
	\draw (0, -12) -- (0, -13);
	
	% Side braces for height = q
	\draw [decorate,decoration={brace,amplitude=6pt},thick]
	(-0.2,-5.9) -- (-0.2, 1.) node[midway,xshift=-10pt,left] {\small $q$};
	\draw [decorate,decoration={brace,amplitude=6pt},thick]
	(-0.2,-13.) -- (-0.2, -6.1) node[midway,xshift=-10pt,left] {\small $q$};
	
	% Top braces for width = k - q
	\draw [decorate,decoration={brace,amplitude=6pt},thick]
	(0,1.2) -- (4.9,1.2) node[midway,above,yshift=4pt] {\small $k - q$};
	\draw [decorate,decoration={brace,amplitude=6pt},thick]
	(5.1,1.2) -- (10,1.2) node[midway,above,yshift=4pt] {\small $k - q$};
	
	%black rectangle
	\draw[very thick] (0,1) -- (5,1);
	\draw[very thick] (0,1) -- (0,-6);
	\draw[very thick] (0,-6) -- (5,-6);
	\draw[very thick] (4.95,1) -- (4.95,-6);
			
	%diagonals
	\draw[very thick] (0,1) -- (5,-6);
	\draw[very thick] (5,1) -- (0,-6);
	
	%blue rectangle
	\draw[very thick, blue] (5,1) -- (10,1);
	\draw[very thick, blue] (5.05,1) -- (5.05,-6);
	\draw[very thick, blue] (5,-6) -- (10,-6);
	\draw[very thick, blue] (10,1) -- (10,-6);
	
	% red ellipses
	\draw[red, thick] (1.5,-8) ellipse [x radius=0.7, y radius=2.3];
	\draw[red, thick] (8.5,-4) ellipse [x radius=0.7, y radius=2.3];
	\node[red] at (9.1,-2,5) {$E_l$};
	
	% red transportation arrow
	\draw[->, red, thick ](8.5,-6.33)%
	to node[auto,swap, right]{}%
	(8.5,-7.5)--(1.5,-11.5)%
	to node[auto,swap, below=0.3cm]{}%
	(1.5,-10.5);
		
\end{tikzpicture}

One first fills the blue rectangle with $k-q$ $1$'s in the first row, $k-q$ $2$'s in the second row, ..., $k-q$ $q$'s in the $q$th row. Thereafter, one moves the boxes in excess within the blue rectangle (i.e., those that do not belong to $\lambda(J_0)$) below the black crossed rectangle, according to the following rule. For $l = 1, \ldots, k-q$, one takes the box excess $E_l$ in the $l$th column from the right of the blue rectangle, and repositions it below the $l$th column from the left of the black crossed rectangle. Once all excess boxes have been moved, the resulting partition (in black) is indeed $\lambda(J_0)$, as a consequence of Eq.~\eqref{eq:pairing_lambdaJ0}. For our specific example ($k = 12, q = 7$ and $i_l$'s given by Eq.~\eqref{eq:our_indices}), the resulting version of $\lambda(J_0) \smallsetminus \lambda(J_1)$ filled with numbers is as follows:\\

\begin{tikzpicture}[scale=.5]

	% partition
	\Ylinecolour{black}
	\Ynodecolour{black}
	\tyoung(0cm, 0cm,
	\ \ \ \ \ 11111,\ \ \ \ \ 2222,\ \ \ \ \ 3333,\ \ \ \ \ 44,\ \ \ \ \ 5,\ \ \ \ \ 6,\ \ \ \ \ 7,2445,3556,4667,577,6,7)
	\draw (0, -12) -- (0, -13);
	
	%black rectangle
	\draw[very thick] (0,1) -- (5,1);
	\draw[very thick] (0,1) -- (0,-6);
	\draw[very thick] (0,-6) -- (5,-6);
	\draw[very thick] (4.95,1) -- (4.95,-6);
	
	%diagonals
	\draw[very thick] (0,1) -- (5,-6);
	\draw[very thick] (5,1) -- (0,-6);

\end{tikzpicture}

We now check that our construction abides by rules (i)-(iv) detailed in Ref.~\cite{fulton2000}, Sec. 3.

(i) Clearly, the entries in any row are weakly increasing from left to right.

(ii) By construction, the entries in each column are strictly increasing from top to bottom.

(iii) For $l = 1, \ldots, 2q$, $l$ appears in exactly $\lambda(J_2)_l$ boxes, i.e., the numbers $1, \ldots, q$ appear $k - q$ times each and the numbers $q + 1, \ldots, 2q$ never appear, by construction.

(iv) Let $1\leq p \leq \sum_{l} \lambda(J_2)_l = q ( k-q)$. With regards to the relevant ordering~\cite{fulton2000}, consider the first $p$ boxes of $\lambda(J_0) \smallsetminus \lambda(J_1)$. Let $1 \leq l \leq q-1$. We claim that the number of times $l$ occurs in the first $p$ boxes is at least as large as the number of times $l+1$ occurs in these first $p$ boxes. First, suppose that the $p$th box is in the upper-right part of the diagram (blue rectangle). Then, our claim is visually clear. Now, instead, suppose the $p$th box is in the lower left part of the diagram (below the black, crossed rectangle). Recall that $\lambda(J_2)_l = \lambda(J_2)_{l+1}$, i.e., $l$ occurs as many times as $l+1$ does. Therefore, checking our claim is equivalent to checking that the number of times $l$ occurs in the last $q ( k-q) - p$ boxes is \emph{at most} as large as the number of times $l+1$ occurs in those last $q ( k-q) - p$ boxes. From this perspective, the claim becomes visually clear.

This concludes the proof.

\subsection{Generalization of Thm.~\ref{thrm:mikes_maj_rel}: unbalanced bipartitions} \label{subsec:mike_general}
We now derive Thm.~\ref{thm:mike_main}, which is to be regarded as an extension of Thm.~\ref{thrm:mikes_maj_rel} to the case of arbitrary bipartitions $\blam = (\lambda_1, \lambda_2)$ of any space dimension $d$. This will require introducing new tools and notations, leading to mathematical results of independent interest. For the goals of our work, the main result is presented in Cor.\til\ref{cor:maj_relation_uneven_bip} on the optimality of $\bb_{\max}$, as defined in Eq.~\eqref{eq:bbest_bip}. 

For a hermitian operator $\gamma$ on an inner product space $V$, let $\gamma_{U}$ denote the orthogonal compression of $\gamma$ onto a subspace $U\subset V$.

In the following, $\bx,\by,\bz$ denote binary sequences. ``$\newmoon$'' is synonymous to $1$, and ``$\fullmoon$'' to $0$. Juxtaposition indicates concatenation. For a binary sequence $\bx$, we will write $|\bx|$ to denote its length.

We start by defining a family of logical statements parametrized by three binary sequences (as well as by the direction of inequality):

\begin{defn}
Let $\bx,\by, \bz$ be binary sequences such that $|\bx|+|\by| = | \bz|$. Then
the string $\bx | \by \ge  \bz$ is defined to be the following statement:\\
``For all hermitian operator $\gamma$ on an $| \bz|$-dimensional inner product space $V$,
for all orthogonal decompositions $V = V_1\oplus V_2$ with $\dim V_1 =|\bx|, \dim V_2 = |\by|$,
it holds that 
\begin{equation}
	\sum_i x_i  \alpha_i + \sum_i y_i \beta_i \ge \sum_i z_i n_i,
\end{equation}
where $\balpha := \specdown(\gamma_{V_1})$, $\bbeta:=\specdown(\gamma_{V_2})$, $\bn:=\specdown(\gamma)$.''

Define $\bx|\by \le  \bz$ in the same way.
\end{defn}

In the following, ``$\forall \bx,\by, \bz$'' will always be a shorthand
for ``for all binary sequences $\bx,\by, \bz$ such that
$|\bx|+|\by|=| \bz|$.

\newcommand{\recursivelawa}{
\bx|\by \ge  \bz \Rightarrow \bx \fullmoon|\by\fullmoon \ge \fullmoon \bz\fullmoon
}

\begin{lemm}
\label{lemm:recursive_law_1}
$\forall\bx,\by, \bz:\recursivelawa$.
\end{lemm}
\begin{proof}
	Let $V$ be any $(| \bz|+2)$-dimensional inner product space and let $\gamma$ be a hermitian operator.
	Let $V = V_1\oplus V_2$ be an orthogonal decomposition with $\dim V_1 = |\bx|+1$ and $\dim V_2 = |\by|+1$.
	Let $\gamma = \sum_{i=1}^{| \bz|+2}n_i \ket{\psi_i}\bra{\psi_i}$ be the spectral decomposition.
	We can find subspaces $U_1\subset V_1, U_2\subset V_2$ such that
	$\psi_{| \bz|+2}\perp U_1$ and $\psi_{| \bz|+2}\perp U_2$ and $\dim U_1 = |\bx|$, $\dim U_2 = |\by|$.
	We then have an orthogonal decomposition $V = U_1 \oplus U_2 \oplus R
	\oplus \mathbb{C}\psi_{| \bz|+2}$, where $R$ is a one-dimensional
	subspace.

	Let $\balpha':= \specdown(\gamma_{U_1}), \bbeta':= \specdown(\gamma_{U_2})$. We then have
	\begin{equation}
	\sum_{i=1}^{|\bx|} x_i  \alpha_i  + 0\cdot \alpha_{|\bx|+1}\ge \sum_{i=1}^{|\bx|}x_i  \alpha_i'
	\end{equation}
	\begin{equation}
	\sum_{i=1}^{|\by|} y_i \beta_i  + 0\cdot \beta_{|\by|+1}\ge \sum_{i=1}^{|\by|}y_i \beta_i'
	\end{equation}
	Taking the sum and using $\bx|\by \ge  \bz$, we get
	\begin{equation}
	\sum_{i}x_i \alpha_i + \sum_i y_i\beta_i \ge \sum_{i=1}^{| \bz|}z_i n_i',
	\end{equation}
	where $\bn':= \specdown(\gamma_{U_1\oplus U_2})$. Since
	$\specdown(\gamma_{U_1\oplus U_2\oplus R})= (n_1, \dotsb,
	n_{| \bz|+1})$, we have $n_i'\ge n_{i+1}$, so we have
	\begin{equation}
	\sum_{i}x_i \alpha_i + \sum_i y_i\beta_i \ge \sum_{i=1}^{| \bz|}z_i n_{i+1},
	\end{equation}
	It follows that $\bx\fullmoon | \by\fullmoon \ge \fullmoon \bz\fullmoon$.
\end{proof}

For a binary sequence $\bx$, we will use the following notation:
\begin{itemize}
\item $\dot\bx$ is the reversed sequence. That is, $\dot x_i = x_{|\bx|-i+1}$.
\item $\bar \bx$ is the bit-flipped sequence. That is, $\bar x_i = 1-x_i$.
\end{itemize}

\begin{lemm}
$\forall\bx,\by, \bz: \bx|\by \ge  \bz\Leftrightarrow \dot\bx|\dot\by\le \dot \bz$.
\end{lemm}
\begin{proof}
	It suffices to prove $\bx|\by \ge  \bz \Rightarrow \dot\bx|\dot\by\le \dot \bz$.

	Let $\gamma$ be a hermitian operator on a $| \bz|$-dimensional inner product space $V$
	with orthogonal decomposition $V_1\oplus V_2$ such that $\dim V_1 = |\bx|,\dim V_2 = |\by|$.
	Let $\balpha:=\specdown(\gamma_{V_1}), \bbeta:=\specdown(\gamma_{V_2}), \bn:= \specdown(\gamma)$.

	Apply $\bx|\by\ge  \bz$ to $-\gamma$ to get
	\begin{equation}
	-\sum_i x_i \alpha_{|\bx|-i+1} - \sum_i y_i \beta_{|\by|-i+1}
	\ge -\sum_i z_i n_{| \bz|-i+1}
	\end{equation}
	or
	\begin{equation}
	\sum_i \dot x_i \alpha_i + \sum_i \dot y_i \beta_i \le \sum_i \dot z_i n_i.
	\end{equation}
	This shows $\dot\bx|\dot\by\le \dot \bz$.

\end{proof}

\begin{lemm}
$\forall\bx,\by, \bz:$
$\bx|\by \ge  \bz\Leftrightarrow \bar\bx| \bar\by\le \bar \bz$.
\end{lemm}
\begin{proof}
	It suffices to show  $\bx|\by \ge  \bz\Rightarrow\bar\bx| \bar\by\le \bar \bz$.

	Let $V=V_1\oplus V_2$, $\gamma$, $\bn$, $\balpha$, $\bbeta$ be as always. Then
	\begin{equation}
	\sum_{i=1}^{|\bx|}  \alpha_i + \sum_{i=1}^{|\by|} \beta_i  = \sum_{i=1}^{| \bz|} n_i
	\end{equation}
	\begin{equation}
	\sum_{i=1}^{|\bx|} x_i  \alpha_i + \sum_{i=1}^{|\by|}  y_i \beta_i \ge \sum_{i=1}^{| \bz|} z_i n_i.
	\end{equation}
	Taking the difference, we get
	\begin{equation}
	\sum_{i=1}^{|\bx|} \bar x_i  \alpha_i + \sum_{i=1}^{|\by|}  \bar y_i
	\beta_i \le \sum_{i=1}^{| \bz|} \bar z_i n_i.
	\end{equation}

\end{proof}

\newcommand{\recursivelawb}{\bx|\by \ge  \bz \Rightarrow \newmoon\bx|\newmoon\by \ge \newmoon \bz\newmoon}

\begin{coro}
\label{cor:recursive_law_2}
$\forall\bx,\by, \bz:\recursivelawb$.
\begin{proof}
	\begin{equation}
	\bx|\by\ge \bz \Rightarrow \dot{\bar\bx}|\dot{\bar\by} \ge \dot{\bar \bz}\Rightarrow 
	\dot{\bar\bx}\fullmoon | \dot{\bar\by}\fullmoon \ge \fullmoon \dot{\bar \bz}\fullmoon
	\Rightarrow \newmoon \bx | \newmoon \by \ge \newmoon \bz\newmoon.
	\end{equation}
\end{proof}
\end{coro}

We will use exponents to indicate repetition. For example, $\newmoon^3 = \newmoon\newmoon\newmoon$. We get the following significant extension of Thm.~\ref{thrm:mikes_maj_rel}.

\begin{thrm} \label{thm:mike_main}
$\newmoon^u\fullmoon^v \mid \newmoon^{u}\by\fullmoon^v\ge \bx \by \dot\bx$ and, equivalently, $\newmoon^u\by\fullmoon^v \mid \newmoon^{u}\fullmoon^v\ge \bx \by \dot\bx$ for any binary sequences $\bx,\by$ with $|\bx|= u+v$ and $\sum_i x_i = u$.
\begin{proof}
	Starting with the trivially true statement $|\by \ge \by$, we repeatedly apply Lemma~\ref{lemm:recursive_law_1} ($\recursivelawa$) and Corollary~\ref{cor:recursive_law_2} ($\recursivelawb$).
\end{proof}
\end{thrm}
In particular, the above theorem allows us to prove the following generalization of Eq.\til\eqref{eq:maj_eq_even_bip}, which implies the optimality of the vector of Eq.~\eqref{eq:bbest_bip} for arbitrary bipartitions $\blam$.
\begin{coro} \label{cor:maj_relation_uneven_bip}
	Let $A$ and $B$ be $\lambda_1 \times \lambda_1$ and $\lambda_2 \times \lambda_2$ Hermitian matrices, respectively, and
	\begin{equation}
	\gamma =
	\begin{pmatrix}
	\begin{array}{c|c}
		A & M \\
		\hline
		M^{\dagger} & B
	\end{array} 
	\end{pmatrix} \; .
	\end{equation}
	Let $\balpha, \bbeta$ and $\boldsymbol{n}$ be the decreasingly ordered spectra of $A, B$ and $\gamma$, respectively. Let $d = \lambda_1 + \lambda_2$, and
	\begin{align} \nonumber
&\bb_{\max} = \\ \nonumber
& \Bigg(
\frac{n_1 + n_d}{2},\ \frac{n_2 + n_{d-1}}{2},\ \ldots,\ \frac{n_{\lambda_2} + n_{\lambda_1 + 1}}{2}, n_{\lambda_2+1}, \, \ldots, \\ 
& n_{\lambda_1}\ ,  \ \frac{n_1 + n_d}{2},\ \frac{n_2 + n_{d-1}}{2},\ \ldots,\ \frac{n_{\lambda_2} + n_{\lambda_1 + 1}}{2}
\Bigg) \; .
\end{align}
Then
\begin{equation}
	\bb_{\max} \prec (\balpha, \bbeta) \; .
\end{equation}
\end{coro}

\begin{proof}
	Let 
	\begin{align} \nonumber
		\bs &= \bb_{\max}^\downarrow \\
		\btt &= (\balpha, \bbeta)^\downarrow \; .
	\end{align}
	We have to prove
	\begin{equation} \label{eq:target_maj_ineq_s_t}
		\sum_{i=1}^k t_i \geq \sum_{i=1}^k s_i
	\end{equation}
	for $k = 1, \dots, d - 1$.
	First, suppose that $k$ is such that
	\begin{itemize}
		\item either $s_k$ is one of the $n_i$'s, for some $i \in \{ \lambda_2 + 1, \dots, \lambda_1 \}$,
		\item or $(s_{k-1}, s_k)$ is one of the couples $(\frac{n_i + n_{d+1-i}}{2},\frac{n_i + n_{d+1-i}}{2})$.
	\end{itemize}
	Equivalently, suppose that neither or each of the elements of every couple $(\frac{n_i + n_{d+1-i}}{2},\frac{n_i + n_{d+1-i}}{2})$ appears in $(s_1, \dots, s_k)$. Then
	\begin{equation}
		\sum_{i=1}^k s_i = \sum_{i=1}^d(\bx \by \dot \bx)_i n_i 
	\end{equation}
	for some bit sequences $\bx, \by$ s.t.
	\begin{align}
		|\bx| &= \lambda_2 \\
		|\by| &= \lambda_1 - \lambda_2\\
		2 \sum_{i=1}^{|\bx| } x_i + \sum_{i=1}^{|\by| } y_i &= k \; .
	\end{align}
	More precisely, $\by$ has the form $\newmoon^m \fullmoon^{|\by| - m}$. Let
	\begin{align}
		u &= \sum_{i=1}^{|\bx| } x_i \\
		v &= |\bx| - u \; .
	\end{align}
	We can apply Thm.\til\ref{thm:mike_main} and get $\newmoon^u \by \fullmoon^v | \newmoon^u  \fullmoon^v \geq \bx \by \dot\bx$. For our matrix $\gamma$, this implies
	\begin{align} \nonumber
		\sum_{i=1}^k t_i &\geq \sum_{i=1}^{k-u} \alpha_i + \sum_{i=1}^u \beta_i \\ \nonumber
		&= \sum_{i=1}^{\lambda_1}(\newmoon^u \by \fullmoon^v)_i \alpha_i +  \sum_{i=1}^{\lambda_2}(\newmoon^u \fullmoon^v)_i \beta_i \\ \nonumber
		&\geq \sum_{i=1}^d(\bx\by\dot\bx)_i n_i \\
		&= \sum_{i=1}^k s_i \; ,
	\end{align}
	thus proving Eq.\til\eqref{eq:target_maj_ineq_s_t}.

	The remaining case is for $k$ s.t. $(s_k, s_{k+1})$ is one of the couples $(\frac{n_i + n_{d+1-i}}{2},\frac{n_i + n_{d+1-i}}{2})$. Importantly, in this case $s_k = s_{k+1}$. For that case, Eq.\til\eqref{eq:target_maj_ineq_s_t} follows from Eq.\til\eqref{eq:target_maj_ineq_s_t} itself relative to a shifted maximum index $k' = k-1, k+1$, which we proved above:
	\begin{align}\nonumber
		\sum_{i=1}^k t_i &\geq \frac{1}{2} \left ( \sum_{i=1}^{k+1} t_i + \sum_{i=1}^{k-1} t_i \right ) \\ \nonumber
		&\geq \frac{1}{2} \left ( \sum_{i=1}^{k+1} s_i + \sum_{i=1}^{k-1} s_i \right ) \\
		&= \sum_{i=1}^k s_i \; ,
	\end{align}
	where in the last passage we used $s_k = s_{k+1}$.
\end{proof}

Going back to the main text notation, the above corollary states that $\bb_{\max}$ as defined in Eq.~\eqref{eq:bbest_bip} is majorized by any other $\bb \in \bset^{\blam}_{\bn}$ for the case of unbalanced bipartitions, too. This completes the proof of Result~\ref{res:bipartitions} and concludes the analytical discussion for the case of bipartitions.

\section{1P space dimension $d = 4$: the nontrivial multipartition \texorpdfstring{$\boldsymbol{\lambda}$}{lambda} $ = (2,1,1)$ and the emergence of regimes} \label{app:d=4}

We examine the case $d = 4$ and $\blam = (2,1,1)$, the only nontrivial multipartition of $4$. Fix $n_1 \geq \dots \geq n_4$, let $\eta = \sum_{j=1}^4 n_j/4$. $\bb_{\max} = \bb_{\max}^{\blam}(\bn)$ is given by Eq.~\eqref{eq:bbest_2_1_1}, which we report here:
\begin{widetext}
\begin{equation} \label{eq:bbest_2_1_1_repeated}
	\boxed{
	\bb_{\max} = 
	\begin{cases} 
		\left( \frac{n_1+n_3+n_4}{3}, n_2 \sbs \frac{n_1+n_3+n_4}{3} \sbs \frac{n_1+n_3+n_4}{3} \right) & \text{if } n_1 \geq \eta \geq n_2 \\[10pt]
		(\eta, \eta \sbs \eta \sbs \eta) & \text{if } n_2 \geq \eta \geq n_3 \\[10pt]
		\left( n_3, \frac{n_1+n_2+n_4}{3} \sbs \frac{n_1+n_2+n_4}{3} \sbs \frac{n_1+n_2+n_4}{3} \right) & \text{if } n_3 \geq \eta \geq n_4 \; .
	\end{cases}
	}
\end{equation}
\end{widetext}
Independent of the regime, we observe once again that the entries of $\bb_{\max}$ are averages of $\bn$ entries. As we will show, the corresponding optimal unitaries are combinations of discrete Fourier transforms, or generalizations thereof.

We now prove that $\bb_{\max}$ as given in Eq.~\eqref{eq:bbest_2_1_1} is in $\bset^{\blam}_{\bn}$ and is majorized by any other $\bb\in \bset^{\blam}_{\bn}$---and therefore maximizes $H$ over $\bset^{\blam}_{\bn}$. We start out with a lemma.

\begin{lemm} \label{lemm:technical_maj_rel}
	Let $\ba = (\underbrace{x, \ldots, x}_{\ell \text{ times}}, \underbrace{y, \ldots, y}_{d-\ell \text{ times}}) \in \R^d$, with $x \geq y$, and let $\bb \in \R^d$. Suppose $\sum_{j=1}^d a_j = \sum_{j=1}^d b_j$, and that $\sum_{j=1}^\ell b_j^\downarrow \geq \ell x$ (equivalently, $\sum_{j=\ell+1}^d b_j^\downarrow \leq (d-\ell) y$). Then $\ba \prec \bb$.
\end{lemm}
\begin{proof}
	For $k = 1, \ldots, \ell$,
	\begin{equation}
		\sum_{j=1}^k a_j = k x \leq k \frac{\sum_{j=1}^\ell b_j^\downarrow}{\ell} \leq k \frac{\sum_{j=1}^k b_j^\downarrow}{k} = \sum_{j=1}^k b_j^\downarrow \; .
	\end{equation}
	For $k = \ell+1, \ldots, d-1$, the inequality $\sum_{j=1}^k a_j \leq \sum_{j=1}^k b_j^\downarrow$ is equivalent to
	\begin{equation}
		\sum_{j=k+1}^d a_j \geq \sum_{j=k+1}^d b_j^\downarrow \; ,
	\end{equation}
	and this holds since
	\begin{align} \nonumber
		\sum_{j=k+1}^d a_j &= (d-k) y \geq (d-k)\frac{\sum_{j=\ell+1}^d b_j^\downarrow}{d-\ell} \\
		&\geq (d-k) \frac{\sum_{j=k+1}^d b_j^\downarrow}{d-k} = \sum_{j=k+1}^d b_j^\downarrow \; .
	\end{align}
\end{proof}
Next, we consider the three regimes separately. \\

(i) \textbf{Regime $n_1 \geq \eta \geq n_2$.} Let $U$ be the $4 \times 4$ unitary acting as the discrete Fourier transform on span$(\be_1, \be_3, \be_4)$ and leaving $\be_2$ unchanged. Then the $\blam$-block spectrum of $U\text{diag}(\bn)U^\dagger$ is
\begin{equation}
	\left( \frac{n_1+n_3+n_4}{3}, n_2 \sbs \frac{n_1+n_3+n_4}{3} \sbs \frac{n_1+n_3+n_4}{3} \right) = \bb_{\max} \; ,
\end{equation}
which proves $\bb_{\max} \in \bset^{\blam}_{\bn}$. We now prove that $\bb_{\max} \prec \bb$ for any $\bb \in \bset^{\blam}_{\bn}$. Write $\bb = (b_1 , b_2 \sbs b_3 \sbs b_4)$, with $b_1 \geq b_2$. Also, we can write
\begin{equation}
	\bb_{\max}^{\blam,\downarrow} = (x,x,x,y) \; ,
\end{equation}
with $x =  \frac{n_1+n_3+n_4}{3} \geq n_2 = y$. The Cauchy interlacing inequalities, see Thm.~\ref{thrm:interlacing}, imply
\begin{equation}
	\min_i b_i \leq b_2 \leq n_2 = y \; .
\end{equation}
We can then apply Lemma~\ref{lemm:technical_maj_rel} (with $\ba = \bb_{\max}^{\blam,\downarrow}$ and $\ell = 3$), and obtain $\bb_{\max} \prec \bb$.\\

(ii) \textbf{Regime $n_2 \geq \eta \geq n_3$.} In this case $\bb_{\max} = (\eta, \eta \sbs \eta \sbs \eta)$ is maximally mixed. So we only need to prove $\bb_{\max} \in \bset^{\blam}_{\bn}$, while the majorization relation is trivial. We start with diag$(\bn)$ and subsequently conjugate it with three unitaries, diag($\bn$)$\rightarrow \gamma \rightarrow \gamma'\rightarrow \gamma''$, to obtain a matrix $\gamma''$ whose $\blam$-block spectrum is $\bb_{\max}$:
\begin{enumerate}
	\item the first unitary acts only on span($\be_1, \be_3$) and ensures $\gamma_{11} = \eta$---this can be achieved because $n_1 \geq \eta \geq n_3$;
	\item the second unitary acts only on span($\be_2, \be_4$) and ensures $\gamma'_{22} = \eta$---this can be achieved because $n_2 \geq \eta \geq n_4$;
	\item the third unitary acts only on span($\be_3, \be_4$) and ensures $\gamma''_{33} = \gamma''_{44} = \eta$.
\end{enumerate}
Note that the unitaries are not uniquely defined. It is easy to check that this construction ensures $\gamma''_{12}=\gamma''_{21}=0$. Therefore the $\blam$-block spectrum of $\gamma''$ is the same as its diagonal, namely $(\eta, \eta \sbs \eta \sbs \eta ) = \bb_{\max}$, as desired.\\

(iii) \textbf{Regime $n_3 \geq \eta \geq n_4$.} The same discussion as in the first regime applies. Equivalently, one can \bt map' the third regime to the first one via the \bt particle-hole' change of variable $\bn \rightarrow \bn' = \mathbf{1} - \bn $.

\section{1P space dimension $d = 5$} \label{app:d=5}
We study the 1P space dimension $d = 5$ and its partitions $\blam$ that were not covered in the main text, i.e., the nontrivial multipartitions $\blam = (3,1,1), (2,2,1)$ and $(2,1,1,1)$.

\subsection{\texorpdfstring{$\boldsymbol{\lambda}$}{lambda} $ = (3,1,1)$}

Set $\blam = (3,1,1)$. Fix $n_1 \geq \dots \geq n_5$. Then $\bb_{\max} = \bb_{\max}^{\blam}(\bn)$ is given by

\begin{widetext}
\begin{equation} \label{eq:bbest_3_1_1}
\boxed{
\bb_{\max} = 
\begin{cases} 
	\left( \frac{n_1+n_4+n_5}{3}, n_2, n_3 \sbs \frac{n_1+n_4+n_5}{3} \sbs \frac{n_1+n_4+n_5}{3} \right) & \text{if } n_1 \geq \frac{n_1+n_2+n_4+n_5}{4} \geq n_2 \\[12pt]
	\left( \frac{n_1+n_2+n_4+n_5}{4}, \frac{n_1+n_2+n_4+n_5}{4}, n_3 \sbs \frac{n_1+n_2+n_4+n_5}{4} \sbs \frac{n_1+n_2+n_4+n_5}{4} \right) & \text{if } n_2 \geq \frac{n_1+n_2+n_4+n_5}{4} \geq n_4 \\[12pt]
	\left( n_3, n_4, \frac{n_1+n_2+n_5}{3} \sbs \frac{n_1+n_2+n_5}{3} \sbs \frac{n_1+n_2+n_5}{3} \right) & \text{if } n_4 \geq \frac{n_1+n_2+n_4+n_5}{4} \geq n_5 \; ,
\end{cases}
}
\end{equation}
\end{widetext}
up to ordering decreasingly the entries within the first block of the spectrum. We notice that in this case there are three different regimes, as for the case $d=4$ and $\blam = (2,1,1)$. Actually, if we `pin' $n_3$ as an entry of the first block, then the structure of $\bb_{\max}$ is exactly the same as the one of Eq.~\eqref{eq:bbest_2_1_1}. This is not a coincidence: the case $d=5$, $\blam = (3,1,1)$ can be reduced to $d=4$, $\blam = (2,1,1)$ thanks to the following result.
\begin{prop} \label{prop:technical_3_1_1}
	For $\blam = (3,1,1)$, pick $\bb \in \bset^{\blam}$. Then there exists a Hermitian matrix $\gamma'$ with the block-diagonal form
	\begin{equation} \label{eq:1_4_block_diag_matrix}
		\gamma' = (n_3) \oplus  \gamma'_{2:5}
	\end{equation}
	with spectrum $\bn$ and $\blam$-block spectrum $\bb'$ which satisfies $\bb'\prec \bb$.  
\end{prop}
\begin{proof}
Let $\gamma$ be a matrix with spectrum $\bn$ and $(3,1,1)$-block spectrum $\bb$. Let $\tilde{\bn}$ be the decreasingly ordered spectrum of $\gamma_{1:4}$. We define $\gamma'$ as a Hermitian matrix with the form
\begin{equation}
	\gamma' = 
	\begin{pmatrix}
		n_3 & 0 & 0 & 0 & 0 \\
		0 & \frac{\tilde{n}_1 + \tilde{n}_4}{2} & 0 & \frac{\tilde{n}_1 - \tilde{n}_4}{2} & a \\
		0 & 0 & \tilde{n}_2+\tilde{n}_3-n_3 & 0 & b \\
		0 & \frac{\tilde{n}_1 - \tilde{n}_4}{2} & 0 & \frac{\tilde{n}_1 + \tilde{n}_4}{2} & c \\
		0 & \bar{a} & \bar{b} & \bar{c} & d
	\end{pmatrix}
\end{equation}
for some $a, b, c \in \C, d \in \R$ such that $\specdown (\gamma') = \bn$, i.e., such that $\specdown (\gamma'_{2:5}) = (n_1, n_2, n_4, n_5)$. Such a matrix $\gamma'$ exists thanks to Thm.~\ref{thrm:interlacing}, because
\begin{equation}
	\specdown (\gamma'_{2:4}) = (\tilde{n}_1,  \tilde{n}_2+\tilde{n}_3-n_3, \tilde{n}_4)
\end{equation}
interlaces $(n_1, n_2, n_4, n_5)$, i.e., one can show that
\begin{equation}
	n_1 \geq \tilde{n}_1 \geq n_2 \geq  \tilde{n}_2+\tilde{n}_3-n_3 \geq n_4 \geq \tilde{n}_4 \geq n_5 \; .
\end{equation}
These inequalities follow from the fact that $\tilde{\bn} = \specdown(\gamma_{1:4})$, in turn, interlaces $\bn = \specdown(\gamma)$.

Moreover, $\tr[\gamma'] = \tr[\gamma]$ implies $\gamma'_{55} = d = \gamma_{55}$.

The $(3,1,1)$-block spectrum of $\gamma'$ is
\begin{eqnarray}
	\bb'= \left( n_3, \frac{\tilde{n}_1 + \tilde{n}_4}{2}, \tilde{n}_2+\tilde{n}_3-n_3 \sbs \frac{\tilde{n}_1 + \tilde{n}_4}{2} \sbs \gamma_{55} \right)  , 
\end{eqnarray}
up to reordering the entries in the first block decreasingly. We now show that $\bb \succ \bb'$.
First, notice that
\begin{equation} \label{eq:bb_majs_bb*}
	\bb \succ \left( \frac{\tilde{n}_1 + \tilde{n}_4}{2}, \tilde{n}_2, \tilde{n}_3 \sbs \frac{\tilde{n}_1 + \tilde{n}_4}{2} \sbs \gamma_{55} \right) = \bb^* \;.
\end{equation}
This follows from the fact that $b_5 = \gamma_{55} = b^*_5$, so that checking $\bb \succ \bb^*$ reduces to checking $\bb_{1:4} \succ \bb^*_{1:4}$, which is proven by applying our result of Sec.~\ref{subsec:bipartitions} to $\gamma_{1:4}$ with the bipartition $(3,1)$. Next, we observe that $\bb'$ can be obtained from $\bb^*$ by mixing the components $b^*_2 = \tilde{n}_2$ and $b^*_3 = \tilde{n}_3$ so as to obtain $b'_1 = n_3$ and $b'_3 = \tilde{n}_2 + \tilde{n}_3 - n_3$ (in view of the interlacing inequality $\tilde{n}_2 \geq n_3 \geq \tilde{n}_3$). Therefore
\begin{equation}
	\bb^* \succ \bb' \; .
\end{equation} 
Along with Eq.~\eqref{eq:bb_majs_bb*}, this implies $\bb \succ \bb'$. This concludes the proof.
\end{proof}
Prop.~\ref{prop:technical_3_1_1} allows us to restrict the maximization of $H$ over the generic block spectra $\bb \in \bset^{\blam}_{\bn}$ to one over block spectra of matrices of the type of Eq.~\eqref{eq:1_4_block_diag_matrix}. These are the same as the block spectra in $\bset^{(2,1,1)}_{(n_1,n_2,n_4,n_5)}$, up to appending the entry $n_3$. Eq.~\eqref{eq:bbest_3_1_1} then follows from Eq.~\eqref{eq:bbest_2_1_1}. More precisely, Prop.~\ref{prop:technical_3_1_1} and our discussion from Appendix.~\ref{app:d=4} imply that $\bb_{\max}$ as per Eq.~\eqref{eq:bbest_3_1_1} is in $\bset^{\blam}_{\bn}$ and is majorized by any other $\bb \in \bset^{\blam}_{\bn}$.

\subsection{\texorpdfstring{$\boldsymbol{\lambda}$}{lambda} $ = (2,2,1)$} \label{subsubsec:2,2,1}

Set $\blam = (2,2,1)$. Fix $n_1 \geq \dots \geq n_5$, let $\eta = \sum_{j=1}^5 n_j/5$. Then $\bb_{\max} = \bb_{\max}^{\blam}(\bn)$ is given by

\begin{widetext}
	\begin{equation} \label{eq:bbest_2_2_1_d5}
\boxed{
\bb_{\max} =
	\begin{cases}
	\left( \frac{n_1+n_4+n_5}{3}, \frac{n_2+n_3}{2} \sbs \frac{n_1+n_4+n_5}{3}, \frac{n_2+n_3}{2} \sbs \frac{n_1+n_4+n_5}{3} \right) & \text{if } n_1 \geq \eta \geq \frac{n_2+n_3}{2} \\
	\left( \frac{n_2+n_3+n_5}{3}, \frac{n_1+n_4}{2} \sbs \frac{n_2+n_3+n_5}{3}, \frac{n_1+n_4}{2} \sbs \frac{n_2+n_3+n_5}{3} \right) & \text{if } \frac{n_2+n_3}{2} \geq \eta \geq \frac{n_1+n_4}{2} \\
	(\eta, \eta \sbs \eta, \eta \sbs \eta) & \text{if } \frac{n_1+n_4}{2}, \frac{n_2 + n_3}{2} \geq \eta \geq \frac{n_3 + n_4}{2}, \frac{n_2+n_5}{2} \\
	\left( \frac{n_2+n_5}{2}, \frac{n_1+n_3+n_4}{3} \sbs \frac{n_2+n_5}{2}, \frac{n_1+n_3+n_4}{3} \sbs \frac{n_1+n_3+n_4}{3} \right) & \text{if } \frac{n_2+n_5}{2} \geq \eta \geq \frac{n_3+n_4}{2} \\
	\left( \frac{n_3+n_4}{2}, \frac{n_1+n_2+n_5}{3} \sbs \frac{n_3+n_4}{2}, \frac{n_1+n_2+n_5}{3} \sbs \frac{n_1+n_2+n_5}{3} \right) & \text{if } \frac{n_3+n_4}{2} \geq \eta \geq n_5 .
	\end{cases}
}
\end{equation}
\end{widetext}
We now prove that $\bb_{\max}$ is in $\bset^{\blam}_{\bn}$ and is majorized by any other $\bb \in \bset^{\blam}_{\bn}$. \\

(i) \textbf{Regime $n_1 \geq \eta \geq \frac{n_2+n_3}{2}$.} $\bb_{\max} \in \bset^{\blam}_{\bn}$ is proven by conjugating diag$(\bn)$ with two discrete Fourier transform matrices acting separately on span$(\be_1, \be_4, \be_5)$ and span$(\be_2, \be_3)$. Next, let $\gamma = \gamma^\dagger$ be a $5 \times 5$ matrix with spectrum $\bn$. Let $\bb = (\alpha_1, \alpha_2 \sbs \beta_1, \beta_2 \sbs \delta)$ be its block spectrum, with $\alpha_1 \geq \alpha_2$ and $\beta_1 \geq \beta_2$. Also, let $\tilde{n}_1 \geq \ldots \geq \tilde{n}_4$ be the eigenvalues of $\gamma_{1:4}$. They interlace $\bn$, by Thm.~\ref{thrm:interlacing}. Also, by Thm.~\ref{thrm:mikes_maj_rel} we get
\begin{equation}
	\balpha + \bbeta \succ (\tilde{n}_1 + \tilde{n}_4, \tilde{n}_2 + \tilde{n}_3 ) \; .
\end{equation}
Combining these two facts, we get
\begin{align} \label{eq:tech_221_A}
	\alpha_2 + \beta_2 &\leq \tilde{n}_2 + \tilde{n}_3 \leq n_2 + n_3 \\ \label{eq:tech_221_B}
	\alpha_2 + \beta_2 &\leq \tilde{n}_1 + \tilde{n}_4 \leq n_1 + n_4  \; .
\end{align}
From inequality~\eqref{eq:tech_221_A} we get $\alpha_2 + \beta_2  \leq 2 \frac{ n_2 + n_3}{2}$, which allows us to use Lemma~\ref{lemm:technical_maj_rel}, leading to $\bb \succ \bb_{\max}$. \\

(ii) \textbf{Regime $\frac{n_2+n_3}{2} \geq \eta \geq \frac{n_1+n_4}{2}$.} The inclusion $\bb_{\max} \in \bset^{\blam}_{\bn}$ is proven as for the first regime. The rest of the argument is completely analogous, too. We adopt the same notation. From inequality~\eqref{eq:tech_221_B} we get $\alpha_2 + \beta_2  \leq 2 \frac{ n_1 + n_4}{2}$, which allows us to use Lemma~\ref{lemm:technical_maj_rel}, leading to $\bb \succ \bb_{\max}$. \\

(iii) \textbf{Regime $\frac{n_1+n_4}{2}, \frac{n_2 + n_3}{2} \geq \eta \geq \frac{n_3 + n_4}{2}, \frac{n_2+n_5}{2}$.} For this regime, $\bb_{\max}$ is maximally mixed, so the task is to prove the inclusion relation $\bb_{\max} \in \bset^{\blam}_{\bn}$. $\bb \succ \bb_{\max}$ holds trivially true $\forall \; \bb \in \bset^{\blam}_{\bn}$.

The regime-defining inequalities imply the existence of $\mu_1 \geq \mu_2 \geq 0$ satisfying
\begin{align} \label{eq:tech1}
	n_2 - \eta , \eta - n_4 &\leq \mu_1 \leq n_1 - \eta, \eta - n_5 \\ \label{eq:tech2}
	n_3 - \eta , \eta - n_3 &\leq \mu_2 \leq n_2 - \eta, \eta - n_4 \; .
\end{align}
Consider the Hermitian matrix
\begin{equation}
	\tilde{\gamma} =
\begin{pmatrix}
	\begin{array}{cc|cc}
		\eta & 0 & \mu_1 & 0 \\			
		0 & \eta & 0 & \mu_2 \\
		\hline
		\mu_1 & 0   & \eta & 0  \\
		0 & \mu_2 & 0 & \eta \\
	\end{array} 
\end{pmatrix} \; .
\end{equation}
The eigenvalues of $\tilde{\gamma}$ are $\eta \pm \mu_1$, $\eta \pm \mu_2$. By inequalities~\eqref{eq:tech1} and~~\eqref{eq:tech2}, they interlace $\bn$. Then, by Thm.~\ref{thrm:interlacing}, there exists a $5 \times 5$ matrix $\gamma = \gamma^\dagger$ with spectrum $\bn$ s.t. $\gamma_{1:4} = \tilde{\gamma}$. Also, $\gamma_{55} = \tr[\gamma] - \tr[\tilde{\gamma}] = \eta$, so $\gamma$ has block spectrum $(\eta, \eta \sbs \eta, \eta \sbs \eta) = \bb_{\max}$, which is thus an element of $\bset^{\blam}_{\bn}$. \\

(iv) \textbf{Regime $\frac{n_2+n_5}{2} \geq \eta \geq \frac{n_3+n_4}{2}$.} One can map this regime to the second one via the \bt particle-hole' change of variable $\bn \rightarrow \bn' = \mathbf{1} - \bn $. \\

(v) \textbf{Regime $\frac{n_3+n_4}{2} \geq \eta \geq n_5$.} One can map this regime to the first one via the \bt particle-hole' change of variable $\bn \rightarrow \bn' = \mathbf{1} - \bn $.

\subsection{\texorpdfstring{$\boldsymbol{\lambda}$}{lambda} $ = (2,1,1,1)$}

Set $\blam = (2,1,1,1)$. Fix $n_1 \geq \dots \geq n_5$, let $\eta = \sum_{j=1}^5 n_j/5$. Then $\bb_{\max} = \bb_{\max}^{\blam}(\bn)$ is given by

\begin{widetext}
	\begin{equation} \label{eq:bbest_2_1_1_1}
\boxed{
\bb_{\max} =
\begin{cases}
\left( \frac{n_1+n_3+n_4+n_5}{4}, n_2 \sbs \frac{n_1+n_3+n_4+n_5}{4} \sbs \frac{n_1+n_3+n_4+n_5}{4} \sbs \frac{n_1+n_3+n_4+n_5}{4} \right) & \text{if } n_1 \geq \eta \geq n_2 \\
(\eta, \eta \sbs \eta \sbs \eta \sbs \eta) & \text{if } n_2 \geq \eta \geq n_4 \\
\left( n_4, \frac{n_1+n_2+n_3+n_5}{4} \sbs \frac{n_1+n_2+n_3+n_5}{4} \sbs \frac{n_1+n_2+n_3+n_5}{4} \sbs \frac{n_1+n_2+n_3+n_5}{4} \right) & \text{if } n_4 \geq \eta \geq n_5 .
\end{cases}
}
\end{equation}
\end{widetext}

We now prove that $\bb_{\max}$ is in $\bset^{\blam}_{\bn}$ and is majorized by any other $\bb \in \bset^{\blam}_{\bn}$. Both Eq.~\eqref{eq:bbest_2_1_1_1} and the derivation follow closely Appendix~\ref{app:d=4}. \\

(i) \textbf{Regime $n_1 \geq \eta \geq n_2$.} The inclusion relation $\bb_{\max} \in \bset_{\bn}^{\blam}$ is proven by conjugating diag$(\bn)$ with a discrete Fourier transform unitary acting on span$(\be_1, \be_3, \be_4, \be_5)$. Next, let $\bb \in \bset^{\blam}_{\bn}$, so that $b_2$ is the second component of the first two-dimensional block of $\bb$. By Thm.~\ref{thrm:interlacing}, $b_2 \leq n_2$. We can thus use Lemma~\ref{lemm:technical_maj_rel} and conclude that $\bb \succ \bb_{\max}$. \\

(ii) \textbf{Regime $n_2 \geq \eta \geq n_4$.} Since $\bb_{\max}$ is maximally mixed, the task is to prove the inclusion relation $\bb_{\max} \in \bset^{\blam}_{\bn}$, while $\bb \succ \bb_{\max}$ is evident $\forall \; \bb \in \bset^{\blam}_{\bn}$. We can construct a Hermitian matrix with spectrum $\bn$ and a maximally mixed $\blam$-block spectrum as follows. We start with diag$(\bn)$ and subsequently conjugate it with three unitaries, diag$(\bn) \rightarrow \gamma \rightarrow \gamma' \rightarrow \gamma''$, as we did in Appendix~\ref{app:d=4}.
\begin{enumerate}
	\item the first unitary acts only on span($\be_1, \be_4$) and ensures $\gamma_{11} = \eta$---this can be achieved because $n_1 \geq \eta \geq n_4$;
	\item the second unitary acts only on span($\be_2, \be_5$) and ensures $\gamma'_{22} = \eta$---this can be achieved because $n_2 \geq \eta \geq n_5$;
	\item the third unitary acts only on span($\be_3, \be_4, \be_5$) and ensures $\gamma''_{33} = \gamma''_{44} = \eta$.
\end{enumerate}
Note that the unitaries are not uniquely defined. It is easy to check that this construction ensures $\gamma''_{12}=\gamma''_{21}=0$. Therefore the $\blam$-block spectrum of $\gamma''$ is the same as its diagonal, namely $(\eta, \eta \sbs \eta \sbs \eta \sbs \eta ) = \bb_{\max}$, as desired.\\

(iii) \textbf{Regime $n_4 \geq \eta \geq n_5$.} One can map this regime to the first one via the \bt particle-hole' change of variable $\bn \rightarrow \bn' = \mathbf{1} - \bn $.

\section{Slater determinants and multipartitions} \label{app:SDs}
In this appendix we focus on SDs, as in Sec.~\ref{sec:main_problem}. We consider an arbitrary space dimension $d$ and multipartition $\blam$. We provide a systematic simplification of the problem of maximal orbital total correlation. This is allowed by the simple structure of the spectrum
\begin{equation}
	\bn = (\underbrace{1,\ldots, 1}_{N^P \text{ times}}, \underbrace{0,\ldots, 0}_{N^H \text{ times}})
\end{equation}
of the 1RDM of any SD with $N^P$ particles and $N^H = d - N^P$ holes. Equivalently, the possible $\gamma$ matrices appearing in Eqs.~\eqref{eq:Iblam_max_rewritten} and~\eqref{eq:def_bset} are the rank-$N^P$ projectors on $\C^d$, which $1:1$ correspond to $N^P$-dimensional subspaces of $\C^d$. So instead of optimizing over $U(d)$ unitaries (used for conjugating diag$(\bn) \rightarrow \gamma$) we can optimize over the much lower-dimensional Grassmanian $\text{Gr}_{N^P}(\C^d)$.

We first introduce some new notation specific to the case of SD spectra. We group all relevant parameters in the 3-tuple $(N^P,N^H,\blam)$, where $\blam$ is a partition of $d = N^P + N^H$. It fixes a set of projectors $\{ \Pi_{\mu} \}_{\mu=1}^{|\blam|}$, whose ranges are the $\lambda_{\mu}$-dimensional subspaces $ V_{\mu}=\text{span}(|\lambda_1+...+\lambda_{\mu-1}+1 \rangle, ..., | \lambda_1+...+\lambda_{\mu-1} + \lambda_{\mu}  \rangle)$. Using the projectors $\Pi_\mu$, we introduce the block-diagonal compression of a $d \times d$ matrix $\gamma$,
\begin{equation} \label{eq:gammablock}
	\gamma^B := \sum_{\mu=1}^{|\blam|} \Pi_{\mu} \gamma \Pi_{\mu}  \; .
\end{equation}
Also, we conveniently adapt the notation and treat the block spectrum as a function of $\gamma$ by setting
\begin{equation}
	\bb^{\blam}(\gamma) := (\specdown(B^\gamma_1) \sbs \ldots \sbs \specdown(B^\gamma_{|\blam|}) ) \; ,
\end{equation}
where $B_\mu^{\gamma}$ is the $\mu$th diagonal block of $\gamma$, so that
\begin{equation}
	\bb^{\blam}(\gamma) = \specdown(\gamma^B) \;\;\; \text{   up to a permutation.}
\end{equation}
Given an arbitrary SD state $\rho$ of $N^P$ particles and $N^H$ holes, we set 
\begin{equation}
\Tilde{C}^{\blam}_{\max}(N^P,N^H):=C^{\blam}_{\max}(\rho) \til 
\end{equation}
1P unitary invariance of $C^{\blam}_{\max}$ ensures that the right-hand side does not depend on the specific SD state $\rho$. Using the fact that $n_j = 0, 1$ for SD states, Eq.~\eqref{eq:Iblam_max_rewritten} yields
\begin{equation}
\Tilde{C}^{\blam}_{\max}(N^P,N^H) = C^{\blam}_{\max}(\rho) =  \max_{\gamma}  H(\bb^{\blam}(\gamma))\til,
\end{equation}
where the maximization is over all rank-$N^P$ projectors $\gamma$. In view of Schur-concavity of the objective function, the idea to perform this maximization is to try to get the block spectrum $\bb^{\blam}(\gamma)$ to be as mixed as possible. This is the same approach used for generic free states. But now the SD restriction, i.e., $n_i\in\{0,1\}$, is expected to bring about simplifications. The basic questions are:
\begin{itemize}
	\item (Q1) Are there any eigenvalues in $\bb^{\blam}(\gamma)$ which are \bt trivially' $0$ or $1$, due to the SD structure? How many?
	\item (Q2) Is it possible to get the non-trivial eigenvalues to be all equal, and thus have $C^{\blam}_{\max}$ hit an upper bound of the form of Eq.~\eqref{eq:basic_bounds}?
\end{itemize}

This appendix is organized as follows. In Sec.~\ref{SDexamples} we show simple examples. In Sec.~\ref{trivialandsmallblocks}, in view of answering (Q1), we rigorously define the number $T(N^P,N^H,\blam)$ of trivial eigenvalues. There we also define a \bt small-blocks condition' on 3-tuples $(N^P,N^H,\blam)$. This condition is expected to select out those 3-tuples with no trivial eigenvalues, $T(N^P,N^H,\blam)=0$. In Sec.~\ref{trivialeigen}, Prop.~\ref{prop:trivialeigen} gives a simple formula for the number $T(N^P, N^H, \blam)$ of trivial eigenvalues; and in Cor.~\ref{coro} it is proved that $T(N^P, N^H, \blam)=0$ if and only if $(N^P,N^H,\blam)$ fulfills the small-blocks condition, indeed. In Sec.~\ref{smallblocksreduction} we present Prop.~\ref{prop:reduction}, which allows us to reduce the case of a generic 3-tuple $(N^P,N^H,\blam)$ to the case of a simpler 3-tuple $(N'^P,N'^H,\blam')$ which fulfills the small-blocks condition, i.e., which exhibits no trivial eigenvalues. Sec.~\ref{subsec:new3tuple} contains a technical calculation. In Sec.~\ref{nontrivial} we then focus on the small-blocks scenario only, and provide a counterexample where the spectrum $\bb^{\blam}(\gamma)$ cannot be made maximally mixed, even though there are no trivial eigenvalues $0$ or $1$. This answers (Q2) negatively, and shows that the function $\Tilde{C}^{\blam}_{\max}$ is interesting (non-trivial) even within the small-blocks scenario.

\subsection{Examples} \label{SDexamples}

In this paragraph we compute $\Tilde{C}^{\blam}_{\max}(N^P,N^H)$ for some simple 3-tuples $(N^P, N^H, \blam)$.\\

\nit (i) For any $N^P$ and $N^H$ and given the trivial, coarsest partition $\blam=(\lambda_1)=d=N^P+N^H$, there is no freedom: $\Pi_1 = \mathbb{I}$ and $\bb^{\blam}(\gamma)=\specdown(\gamma)$ is made up of $0$'s and $1$'s, which we call trivial eigenvalues. Then $\Tilde{C}^{\blam}_{\max}(N^P,N^H)=0$. \\

\nit (ii) For $N^P=1, N^H=2, \blam =(2,1)$, $\gamma^B$ has two blocks. The first block, $\Pi_1 \gamma \Pi_1$, is easily understood to have $\text{rank}(\Pi_1 \gamma \Pi_1) \leq \text{rank}(\gamma) = 1$. Thus $\Pi_{1}\gamma\Pi_{1} \big|_{V_1} = \Pi_{1}\gamma \big|_{V_1} $ has a trivial eigenvalue $0$, independent of how $\gamma$ is chosen. The second eigenvalue of the first block and (the eigenvalue of) the second block can be chosen to be $\frac{1}{2}$ and $\frac{1}{2}$. To achieve that, set, e.g., $\gamma = |\phi \rangle \langle \phi |$ with $| \phi \rangle = (|2\rangle + | 3 \rangle )/\sqrt{2}$. Then $\Tilde{C}^{\blam}_{\max}(N^P,N^H)=2~h(1/2)=2 \log 2$. Similar reasoning applies for $N^P=2,N^H=1, \blam =(2,1)$.\\

\nit (iii) For any $N^P \geq 1, N^H \geq 1$ and given the trivial partition $\blam =(1,...,1)$ we know from the case of general spectra that $\bb^{\blam}(\gamma)$ can be made maximally mixed. In particular there are no trivial $0, 1$ eigenvalues, and the upper bound Eq.~\eqref{eq:basic_bounds} is hit, i.e., $\Tilde{C}^{\blam}_{\max}(N^P,N^H)= d \cdot h(N^P/d)=d \cdot h(N^H/d)$.\\

\nit (iv) For $N^P=2,N^H=2, \blam =(2,2)$ we can put the two particles (or the two holes) in $| \phi \rangle = (| 1\rangle + | 3 \rangle)/\sqrt{2}$ and in $| \phi' \rangle = (| 2 \rangle + | 4 \rangle)/\sqrt{2}$. One then gets $\gamma = | \phi \rangle \langle \phi |+ | \phi' \rangle \langle \phi' |$ and $\gamma^B=\mathbb{I}/2$, i.e., a $\bb^{\blam}(\gamma)$ which is maximally mixed. We have no trivial eigenvalues, and the upper bound of Eq.~\eqref{eq:basic_bounds} is again hit.\\

\nit (v) For $N^P=2,N^H=2, \blam =(2,1,1)$, notice that $\blam$ is obtained from the partition in (iv) by refining, so that the same construction and results as in (iv) hold.

\subsection{Trivial eigenvalues and the small-blocks condition} \label{trivialandsmallblocks}

The examples above can be split into two groups:
\begin{itemize}
	\item In (i) and (ii) there are trivial eigenvalues $0,1$; therefore the bound Eq.~\eqref{eq:basic_bounds} is not reached;
	\item In (iii)-(v) there are no trivial eigenvalues $0,1$; therefore the bound Eq.~\eqref{eq:basic_bounds} can in principle be reached; and it is in fact reached.
\end{itemize}
In (i) and (ii) there are two reasons why there are trivial eigenvalues $0,1$: the high \emph{degeneracy} of the eigenvalues $0,1$, and the fact that the first block $\Pi_{1}\gamma\Pi_{1}$ is \emph{too large}, i.e., the dimension dim($V_1$) of the first subspace $V_1$ is too high. For example, the hand-wavy idea behind case (ii) is the following: $ V_2$ is one dimensional and can thus host at most one of the $N^H = 2$ holes. As a consequence, the big subspace $V_1$ must host one entire hole, to which a trivial eigenvalue $0$ of $\Pi_{1}\gamma \big|_{V_1} $ corresponds. More precisely $\lambda_1 > N^P =$ rank$(\gamma)$, so that $\Pi_{1}\gamma \big|_{V_1} $ cannot have full rank, i.e., must have trivial eigenvalues $0$ (see Lemma~\ref{lem:evalsubmatrix}). Similarly, there must be trivial eigenvalues $1$ whenever $\lambda_1 > N^H = $ rank$(\mathbb{I} - \gamma)$. Instead, in (iii)-(v) we have $\lambda_1 \leq N^P,N^H$, so that no trivial eigenvalues are expected.

In this explanation, the first key concept is that of trivial eigenvalues, which is captured in the following definition.

\begin{defn} \label{trivialeigendef}
The \emph{number of trivial eigenvalues} associated to a 3-tuple $(N^P,N^H,\blam)$ is

\begin{align} \nonumber
	T(N^P,N^H,\blam) &:= \min_{\gamma} \left[ \text{\emph{dim(ker $\gamma^B$)}} + \text{\emph{dim(fix $\gamma^B$)}} \right]
	\\ \label{trivialeigendefformula}
	&= \min_{\gamma} \left[ \text{\emph{\# of 0's and 1's in $\bb^{\blam}(\gamma)$}} \right] ~,
\end{align}
where the minimization is over rank-$N^P$ projectors $\gamma$ on $\H_1 \cong \mathbb{C}^{N^P+N^H} $ and $\gamma^B$ is given by Eq.~\eqref{eq:gammablock}.
\end{defn}
\nit In the examples above we have implicitly computed $T$: in (i) $T=d$, in (ii) $T=1$ and in (iii)-(v) $T=0$. 

The second key concept is the relation ($>,=$ or $<$) between the $\lambda_{\mu}$'s and $N^P, N^H$.

\begin{defn} \label{mnumbersmallblocksdef}
For $N^P,N^H,\lambda_{\mu} \in \mathbb{N}$ we define the coefficients
\begin{align} \nonumber
	m^{H}(N^P,\lambda_{\mu}) &:= \max\{ 0, \lambda_{\mu}-N^{P} \}\\
	m^{P}(N^H,\lambda_{\mu}) &:= \max\{ 0, \lambda_{\mu}-N^{H} \} ~.
\end{align}
%where $N^H=d-N^P$ and dependence on $(N^P,N^H,\blam)$ is implicit.\\
Also, we say that a 3-tuple $(N^P,N^H,\blam)$ satisfies the \emph{small-blocks condition} iff the corresponding coefficients $m^{H}(N^{P},\lambda_{\mu}),m^{P}(N^{H},\lambda_{\mu})$ all vanish, i.e., $\forall \mu =1,...,|\blam|$
\begin{equation}
\lambda_{\mu} \leq N^{P},N^H ~.
\end{equation}
In view of $\lambda_1 \geq ... \geq \lambda_{|\blam|}$, the small-blocks condition is equivalent to $\lambda_{1} \leq N^P,N^H$.
\end{defn}
\nit The coefficient $m^{H/P}(N^{P/H},\lambda_{\mu})$ has the meaning of minimal number of holes/particles that a $\lambda_{\mu}$-dimensional subspace of the 1P space must contain if one considers an SD state with $N^{P/H}$ particles/holes. This physical interpretation follows from Lemma~\ref{lem:evalsubmatrix}. Regarding the examples above: in (i) $m^{H/P}(N^{P/H},\lambda_{1})=N^{H/P}$, in (ii) $m^{H}(N^{P},\lambda_{1})=1$ and the other three $m^{P/H}$ coefficients vanish, and (iii)-(v) satisfy the small-blocks condition.

%In~\ref{smallblocksreduction} we show how one can \textbf{reduce} the case of a general 3-tuple $(d, N^P,\blam)$ to the case of a simpler 3-tuple $(d', N'^P,\blam)$ which fulfills the small-blocks condition. In~\ref{nontrivial} we show an example where the small-blocks condition is fulfilled and $\bb^{\blam}(\gamma)$ can actually be chosen to contain no $0$'s nor $1$'s at all, but nonetheless the upper bound~\ref{upperbound} is \textbf{not} hit.

\subsection{\bluecheck Formula for the number of trivial eigenvalues} \label{trivialeigen}

The number of trivial eigenvalues $T(N^P,N^H,\blam)$ defined in Def.~\ref{trivialeigendef} can be expressed with a simple formula in terms of the $m_{\mu}^{P/H}$ coefficients introduced in Def.~\ref{mnumbersmallblocksdef}:

\begin{prop} \label{prop:trivialeigen}

For any 3-tuple $(N^P,N^H,\blam)$, the number of trivial eigenvalues is
\begin{equation} \label{trivialeigenformula}
	T(N^P,N^H,\blam) = \sum_{\mu=1}^{|\blam|} \left [ m^H(N^P,\lambda_{\mu}) + m^P(N^H,\lambda_{\mu}) \right ] ~.
\end{equation}

\end{prop}
\begin{proof}
%'$\geq$': trivial direction\\
%'$\leq$': use reduction to small-blocks 3-tuples + hardcore bookkeeping within small-blocks scenario\\ \\
In the following, $m^{H/P}_{\mu}$ is short for $m^{H/P}(N^{P/H},\lambda_{\mu})$; their sum over $\mu$, appearing in the RHS of Eq.~\eqref{trivialeigenformula}, will be denoted by $M^{H/P}=\sum_{\mu=1}^{|\blam|}m^{H/P}_{\mu}$. If $N^P=0$, then the only rank-$0$ projector is $\gamma=0$, so that $\bb^{\blam}(\gamma)=(0,...,0)$; also, $m^H_{\mu}=\lambda_{\mu}$. Then $T(N^P,N^H,\blam)=d=M^H$. Similar reasoning applies for $N^H=0$. For the case $\blam =(\lambda_1)=(d)$ the result is straightforward, too.

Suppose $N^P,N^H \geq 1$ and $\lambda_1 < d$. First we prove $T(N^P,N^H,\blam) \geq M^H + M^P$. This amounts to proving that, for every rank-$N^P$ projector $\gamma$, the number of $0$'s and $1$'s in $\bb^{\blam}(\gamma)$ is at least $M^H + M^P$. This follows from the fact that $\forall \mu=1,...,|\blam|$ the block $\Pi_{\mu}\gamma\Pi_{\mu}\big|_{V_{\mu}}$ has at least $m^H_{\mu}$ eigenvalues $0$ and $m^P_{\mu}$ eigenvalues $1$, according to Lemma~\ref{lem:evalsubmatrix}. 

Secondly, we prove $T(N^P,N^H,\blam) \leq M^H + M^P$. To do that, we construct a rank-$N^P$ projector $\gamma$ s.t. $\bb^{\blam}(\gamma)$ has only $M^H$ eigenvalues $0$ and only $M^P$ eigenvalues $1$.

According to the first part of the present proof, for any $\gamma$ we can choose $\forall \mu=1,...,|\blam|$ (in a non-univocal way) an orthogonal decomposition of $V_\mu$ of the form
\begin{equation}
V_{\mu} = V^{H}_{\mu} \oplus V^{P}_{\mu}\oplus V'_{\mu}~,
\end{equation}
where $V^{H}_{\mu} \subset \text{ker($\Pi_{\mu}\gamma\big |_{V_\mu}$)}$, $V^{P}_{\mu}\subset \text{fix($\Pi_{\mu}\gamma\big |_{V_\mu}$)}$, dim($V^{H}_{\mu}$)$=m^H_{\mu}$ and dim($V^{P}_{\mu}$)$=m^P_{\mu}$. Plus, up to applying a block-diagonal unitary which leaves $\bb^{\blam}(\gamma)$ unchanged, we can assume $V^{H}_{1}=\text{span}(|1\rangle,...,|m^H_{1}\rangle)$, $V^{P}_{1}=\text{span}(|m^H_{1} + 1\rangle,...,|m^H_{1}+m^P_{1}\rangle)$ etc, i.e., we may assume the subspaces $V^{H}_{\mu}, V^{P}_{\mu} $ and $ V'_{\mu}$ are constructed from the reference basis as follows:
\begin{align} \nonumber
	&\{| j \rangle \}_{j=1}^d = \{ \underbrace{| 1 \rangle, ..., | m^H_{1} \rangle,}_{\text{basis of } V^{H}_{1}}  \underbrace{| m^H_{1} + 1  \rangle, ..., | m^H_{1}+m^P_{1} \rangle,}_{\text{basis of } V^{P}_{1}} \\ \label{eq:basis_splitting}
	&\underbrace{| m^H_{1}+m^P_{1} +1 \rangle,...,|\lambda_{1}\rangle}_{\text{basis of } V'_{1}},
	\underbrace{| \lambda_{1} + 1 \rangle, ..., | \lambda_{1} + m^H_{2} \rangle,}_{\text{basis of } V^{H}_{2}} ...\}~ .
\end{align}
Next, an important observation is that any eigenvector $|u \rangle \in V_{\mu}$ of $\Pi_{\mu}\gamma|_{V_\mu}$ with eigenvalue $\lambda =0$ or $1$ is also an eigenvector of $\gamma$, with the same eigenvalue. In fact, using $\Pi_{\mu} |u\rangle=|u \rangle $ and $\gamma^2=\gamma=\gamma^{\dagger}$ we get
\begin{equation}
	\lVert \gamma | u \rangle \rVert^2 = \langle u | \gamma | u \rangle = \langle u | ( \Pi_{\mu}\gamma|_{V_\mu} ) | u \rangle = \lambda \lVert | u \rangle \rVert^2 ~.
\end{equation}
From that, using $\lambda =0$ or $1$, and using $0 \leq\gamma\leq \mathbb{I}$, we get $\gamma |u \rangle = \lambda |u \rangle$. Notice that the last passages are based on features which are specific to the 1RDM spectra of SD states.

It follows that $V^{H}_{\mu} \subset \text{ker($\gamma$)}$ and $V^{P}_{\mu} \subset \text{fix($\gamma$)}$. That is, the subspaces $V^{H}_{\mu}$ and $V^{P}_{\mu}$ are made up of eigenvectors of $\gamma$ with trivial eigenvalues $0,1$ (respectively). We can then write a spectral representation of $\gamma$ in the following \emph{normal form}:
\begin{equation} \label{normalgamma}
	\gamma = 0 \cdot \mathbb{P}^H + 1 \cdot \mathbb{P}^P +\mathbb{P}' \, \gamma \, \mathbb{P}' ~.
\end{equation}
Here,
\begin{itemize}
	\item $\mathbb{P}^{H/P}$ and $\mathbb{P}'$ are three projectors with orthogonal ranges, $\bigoplus_{\mu=1}^{|\blam|} V_{\mu}^{H/P}$ and $\bigoplus_{\mu=1}^{|\blam|} V_{\mu}'$ respectively; they add up to $\mathbb{I}$. Notice that the ranges $\bigoplus_{\mu=1}^{|\blam|} V_{\mu}^{H/P}$ have dimension $M^{H/P}$.
	\item In particular we name $\H_1':=\bigoplus_{\mu=1}^{|\blam|} V_{\mu}'$ the range of $\mathbb{P}'$. It has dimension $d' = N'^P+N'^H$, where $N'^P := N^P - M^P $ and $N'^H := N^H - M^H$ are the number of particles/holes which are left after putting $M^{P/H}$ of them in the $V_{\mu}^{P/H}$'s.
	\item $\mathbb{P}' \, \gamma  \, \mathbb{P}'$ should be thought of as the non-trivial part of $\gamma$. We will soon only need its compression $\gamma':=\gamma\big|_{\H_1'}$. Notice that $\gamma'$ entails the remaining eigenvalues of $\gamma$, i.e., $\spec(\gamma')$ contains exactly $N'^P$ eigenvalues $1$ and $N'^H$ eigenvalues $0$, i.e., $\gamma'$ is a rank-$N'^P$ projector acting on $\H_1'$. 
\end{itemize}
The corresponding matrix normal form, with reference to Eq.\til\eqref{eq:basis_splitting} is given by
\begin{equation} \nonumber \label{normalgammamatrix}
	\gamma \rightarrow
\begin{bmatrix}
	\begin{array}{ccc|ccc|c}
		\mathbb{O}_{m^H_1} &  & &  & & &  \\
%      \hline
		& \mathbb{I}_{m^P_1} &  & & & &  \\
%     \hline
		&  & * &  & & * &  \\
		\hline 
		&  &  &  \mathbb{O}_{m^H_2} & & &  \\
		&  &  &  & \mathbb{I}_{m^P_2} & &  \\
%      \hline
		&  & *  &   & & *  &    \\
		\hline
		&  &  &   &   & &    \ddots  \\
	\end{array}
\end{bmatrix}
\end{equation}
\begin{equation}
	\rightarrow  \begin{bmatrix}
	\begin{array}{c|c|ccc}
		\mathbb{O}_{M^H} &  &  & & \\
		\hline
		& \mathbb{I}_{M^P} &  & &  \\
		\hline
		&   &  * & * &  \\
	&  & *  & * &  \\

		&  &   &  & \ddots  \\
	\end{array}
\end{bmatrix} .
\end{equation}
Here all blank blocks are filled with zeros. Asterisks denote nonzero subblocks. From the first matrix to the second one we have changed the ordering in the underlying basis $\{ |j \rangle \}$. In the first matrix, the blocks correspond to the subspaces $V_{\mu}$; in the second matrix the blocks refer to the normal form Eq.~\eqref{normalgamma}. For the last matrix, the first group of basis elements spans range($\mathbb{P}^H$), the second group of basis elements spans range($\mathbb{P}^P$), and the third group of basis elements spans $\H_1'$. The third block then represents $\gamma'$.

Now we consider the consequences of these manipulations on $\gamma^B$ and its spectrum. Notice that $\Pi_{\mu}$, $\mathbb{P}^P$, $\mathbb{P}^H$, $\mathbb{P}'$ are all built on the reference basis $\{|j\rangle \}$, and hence commute with each other. Then, recalling Eq.~\eqref{eq:gammablock} and using Eq.~\eqref{normalgamma}, we get
\begin{equation}
	\gamma^B = 0 \cdot \mathbb{P}^H + 1 \cdot \mathbb{P}^P + \sum_{\mu=1}^{|\blam|} \Pi'_{\mu}  \gamma  \Pi'_{\mu} ~.
\end{equation}
Here we introduced the projectors $\Pi'_{\mu}:=\mathbb{P}' \Pi_{\mu}=\Pi_{\mu}\mathbb{P}'$ with orthogonal ranges $V'_{\mu}$. When we restrict to $\H_1' = \bigoplus_{\mu=1}^{|\blam|} V_{\mu}'$, we get $\sum_{\mu=1}^{|\blam|} \left( \Pi'_{\mu}\big|_{\H_1'} \right) = \mathbb{I}\big|_{\H_1'}$. For later convenience, we also set
\begin{equation}
	\lambda'_{\mu}:=\text{dim}(V'_{\mu})=\lambda_{\mu}-m^P_{\mu}-m^H_{\mu} \; ,
\end{equation}
so that $\blam'$ is a partition of $d'$. It is straightforward to check that
\begin{equation}
	\lambda'_{\mu} \geq 1 \; .
\end{equation}

The matrix representing $\gamma^B$, starting from the last matrix representation of $\gamma$ in Eq.~\eqref{normalgammamatrix}, is
\begin{equation} \label{normalgammablockmatrix}
	\gamma^B \rightarrow  \begin{bmatrix}
	\begin{array}{ccc|ccc|ccc}
		0 &  &  & & & & & & \\
		& \ddots &  & & & & & & \\
		&  &  0 &  & & & & & \\
		\hline
		&  &  & 1 &  & & & &  \\
		&  &  &   & \ddots & & &  & \\
		&  &  &   &  & 1 & & &  \\
		\hline 
		&&&&  &  & * &   &   \\
		&&&&  &  &   & * &  \\
		&&&&  &  &   &  & \ddots  \\
	\end{array}
\end{bmatrix} ,
\end{equation}
where we simply set to $0$ the off-diagonal subblocks of the third block. The spectrum $\bb^{\blam}(\gamma)$ of $\gamma^B$ is
\begin{align} \nonumber
	\bb^{\blam}(\gamma) = &\big( \underbrace{0,...,0,}_{M^{H} \; \text{times}} \underbrace{1,...,1,}_{ M^P \; \text{times}} \;\underbrace{\specdown(\gamma'^B)}_{d'\text{ entries}} \big) \\ \label{specblocknormal}
	&\text{up to a permutation.} 
\end{align}
Here we have introduced the operator
\begin{equation} \label{gammaprimeB}
	\gamma'^B = \sum_{\mu}\Pi_{\mu}\big|_{\H_1'} \gamma' \, \Pi_{\mu}\big|_{\H_1'} = \Big(\sum_{\mu} \Pi'_{\mu}  \gamma  \Pi'_{\mu} \Big)\Big|_{\H_1'} ~,
\end{equation}
which acts on $\H_1'$ and is represented by the third block in Eq.~\eqref{normalgammablockmatrix}. \\
Now we focus on what happens as $\gamma$ is varied. We recall that $\gamma$ can be any generic rank-$N^P$ projector acting on $\H_1$ s.t. $V^{P}_{\mu} \subset \text{fix}(\gamma)$ and $V^{H}_{\mu} \subset \text{ker}(\gamma)$. These conditions fix the eigenvectors of $\gamma$ in $\H_1'^{\perp}$, but we are free to choose the eigenvectors of $\gamma$ in $\H_1'$. More precisely, $\gamma$ can be chosen such that $\gamma'$ be \textit{any} fixed rank-$N'^P$ projector acting on $\H_1'$. From this, and in view of Eq.~\eqref{specblocknormal}, it follows that the following are equivalent:
\begin{enumerate}
	\item $\exists$ rank-$N^P$ projector $\gamma$ on $\H_1$ s.t. $\gamma^B$ has exactly $M^P + M^H$ eigenvalues $0,1$;
	\item $\exists$ rank-$N'^P$ projector $\gamma'$ on $\H_1'$ s.t. $\gamma'^B$ has \emph{no} eigenvalues $0,1$.
\end{enumerate}
This means that we can \textit{reduce} the number of dimensions and just check the existence of a $\gamma'$ with the desired properties. To summarize, we need to look at a new 3-tuple $(N'^P,N'^H,\blam')$, following this reduction scheme:
\begin{equation}
(N^P,N^H,\blam) \rightarrow (N'^P,N'^H,\blam')
\end{equation}
\begin{align} \nonumber
\text{with} ~~~~~~~~~~ N'^P&=N^P-M^P \\ \nonumber
	N'^H&=N^H- M^H \\ 
%   \nonumber    d'&=N'^P + N'^H \\ 
\label{reductionscheme}
	\lambda'_{\mu}&= \lambda_{\mu} - m^P_{\mu} - m^H_{\mu} \geq 1 ~.
\end{align}
A careful check, see Sec.~\ref{subsec:new3tuple}, shows that the new 3-tuple $(N'^P,N'^H,\blam')$ \textit{fulfills the small-blocks condition}.
%\textcolor{blue}{,so that the reduction~\ref{reductionscheme} can only be applied once}
It is then in principle possible to build a rank-$N'^P$ projector $\gamma'$ on $\H_1'$ s.t. $\gamma'^B$ has no eigenvalues $0,1$. In fact, the reduction $\gamma \rightarrow \gamma'$ was all about getting rid of trivial eigenvalues. Below is a construction of such $\gamma'$.

\hideorblue{We have reduced the proof for an arbitrary $(N^P,N^H,\blam')$ to the proof for a $(N'^P,N'^H,\blam)$ fulfilling the small-blocks condition, for which $m^P_{\mu} =m^H_{\mu} = 0$.}

We set $\Pi'_{\mu}\big|_{\H_1'} = \Pi'_{\mu}$, etc.; and drop all \bt primes', e.g, $\H_1 \rightarrow \H'_1$. The construction is broken down into three steps.\\

\nit \textbf{Step 1:} We define a suitable partition of the reference basis $\mathcal{B}$ into $\lambda_1$ subsets $( \mathcal{B}_k)_{k=1}^{\lambda_1}$ in the following way:
\begin{equation} \nonumber
	\underline{| 1 \rangle}, \underline{\underline{| 2 \rangle}},...,\underline{\underline{\underline{| \lambda_1 \rangle}}},\underline{| \lambda_1+1 \rangle},\underline{\underline{| \lambda_1 + 2 \rangle}}, ... , \underline{| m \lambda_1 + 1\rangle}, ... ~ .
\end{equation}
By that, we mean
\begin{equation}
	\mathcal{B}_k =\{ | j \rangle ~ | ~ 1 \leq j \leq d, j = k \text{ mod }\lambda_1 \}~ .
\end{equation}
We define $S_k = \text{span}(\mathcal{B}_k)$. These subspaces and the partition of $\mathcal{B}$ above should not be confused with the subspaces $V_{\mu}$ and the corresponding partition of $\mathcal{B}$. We have dim$(S_k) \in \{p, p+1\}$ with $p=\lfloor \frac{d}{\lambda_1} \rfloor \geq 2$. The latter inequality follows from the small-blocks condition via $d=N^P+N^H \geq 2 \lambda_1$. Also, in view of $\lambda_1 \geq \lambda_{\mu}$, by construction $\mathcal{B}_k$ and the subset of $\mathcal{B}$ generating $V_{\mu}$ have at most one element $|j \rangle$ in common, i.e., $S_k \cap V_{\mu}$ is at most one-dimensional, $\forall k,\mu$.\\

\nit \textbf{Step 2:} We will designate the subspaces $S_k$ as invariant subspaces (but not eigenspaces) of the to-be-defined $\gamma$. In that view we distribute particles (eigenvalues 1) and holes (eigenvalues 0) to the subspaces $S_k$ in a way such that $S_k$ receives dim($S_k$) among particles and holes, and at least 1 particle and 1 hole. More precisely, for $k=1,...,\lambda_1$ we choose $n_k^P,n_k^H$ s.t. 
\begin{itemize}
	\item $n_k^P,n_k^H \geq 1$;
	\item $n_k^P + n_k^H =$ dim($S_k$).
\end{itemize}
Such a choice is possible thanks to the small-blocks condition, and because $N^P+N^H = d = \sum_{k=1}^{\lambda_1}$dim($S_k$). For example, $S_1$ receives a particle-hole couple, $S_2$ receives a particle-hole couple, ..., $S_{\lambda_1}$ receives a particle-hole couple; and then the remaining $d-2\lambda_1 \geq 0$ particles and holes are distributed randomly to those subspaces $S_k$ which are not yet full.\\

\nit \textbf{Step 3:} For $k=1,...,\lambda_1$ we choose a basis for $S_k$, $\mathcal{B}'_k = \{ | f^{(k)}_{l} \rangle \}_{l=1}^{\text{dim}(S_k)}$, which is maximally unbiased with respect to $\mathcal{B}_k$, i.e., $|\langle j | f^{(k)}_{l} \rangle |^2 = \frac{1}{\text{dim}(S_k)}$ $\forall$ $| j \rangle \in \mathcal{B}_k$ and $l=1,...,\text{dim}(S_k)$. For example a basis $\mathcal{B}'_k $ connected to $\mathcal{B}_k $ via a discrete Fourier transformation matrix is maximally unbiased with respect to $\mathcal{B}_k $. We finally define $\gamma$ by using the \bt Fourier' modes $| f^{(k)}_{l} \rangle$ as its eigenvectors and assigning eigenvalues based on the coefficients $n_k^P,n_k^H$: \hideorblue{This corresponds to performing a complete Schur Horn procedure among $\mathcal{B}_k $ elements.}

\begin{align} \nonumber
	\gamma &:= 1 \cdot \left ( \sum_{k=1}^{\lambda_1}\sum_{l=1}^{n_k^P} | f^{(k)}_{l} \rangle \langle f^{(k)}_{l} | \right ) \\
		& + 0 \cdot \left ( \sum_{k=1}^{\lambda_1}\sum_{l=1}^{n_k^H} | f^{(k)}_{n_k^P+l} \rangle \langle f^{(k)}_{n_k^P+l} |  \right )~,
\end{align}
Now we check that $\gamma$ has the desired property, namely that $0,1 \notin \bb^{\blam}(\gamma)$. We then study the matrix elements of the $\mu$-th block $\Pi_{\mu}\gamma\big|_{V_{\mu}}$. For $(\sum_{\nu=1}^{\mu-1}\lambda_{\nu}) +1 \leq j,j' \leq (\sum_{\nu=1}^{\mu-1}\lambda_{\nu}) +\lambda_{\mu}$, 
\begin{align}
	\langle j | \gamma | j' \rangle &=\sum_{k=1}^{\lambda_1}\sum_{l=1}^{n_k^P} \langle j | f^{(k)}_{l} \rangle \langle f^{(k)}_{l} | j' \rangle \\
	&= \delta_{jj'} \sum_{k=1}^{\lambda_1}\sum_{l=1}^{n_k^P} | \langle j | f^{(k)}_{l} \rangle |^2 \\
	&= \delta_{jj'} \sum_{l=1}^{n_{k_j}^P} | \langle j | f^{(k_j)}_{l} \rangle |^2
\end{align}
Here we used that, for any $k=1,...,\lambda_1$, if $j\neq j'$ then either $| j \rangle \perp S_k$ or $| j' \rangle \perp S_k$ or both (see Step 1). Similarly we noticed that $\langle j | f^{(k)}_{l} \rangle \neq 0 $ can only happen if $k=k_j$, where we defined $k_j\in\{1,...,\lambda_1\}$ via $k=j$ mod($\lambda_1$). We then see that the block $\Pi_{\mu}\gamma\Pi_{\mu}\big|_{V_{\mu}}$ is diagonal in the standard basis of $V_{\mu}$. The corresponding eigenvalues are
\begin{equation}
	\mu_{k_j} =\sum_{l=1}^{n_{k_j}^P} | \langle j | f^{(k_j)}_{l} \rangle |^2 = \frac{n_{k_j}^P}{\text{dim}(S_{k_j})} \in (0,1) ~.
\end{equation}
where we used the maximal unbiasedness of $\mathcal{B}'_k$ with respect to $\mathcal{B}_k$ and the fact that $0 < n_{k_j}^P < \text{dim}(S_{k_j})$.\\
Putting all blocks together, we get
\begin{align} \nonumber
		\bb^{\blam}(\gamma) = &\big( \underbrace{\mu_1,...,\mu_1,}_{\text{dim}(S_1) \text{ times}}\underbrace{\mu_2,...,\mu_2,}_{\text{dim}(S_2) \text{ times}}..., \underbrace{\mu_{\lambda_1},...,\mu_{\lambda_1}}_{\text{dim}(S_{\lambda_1}) \text{ times}} \big) \\
		&\text{up to a permutation,} 
\end{align}
which does not contain any eigenvalues $0,1$.
\end{proof}
\nit The meaning of Prop.\til\ref{prop:trivialeigen} is the following: within the $\mu$th subspace $V_{\mu}$ of $\H_1$, $m^H(N^P,\lambda_{\mu})$ dimensions must be occupied by holes (trivial eigenvalues $0$), $m^P(N^H,\lambda_{\mu})$ dimensions must be occupied by particles (trivial eigenvalues $1$), but we can choose $\gamma$ in such a way that no other dimensions are completely occupied by particles or holes (no further trivial eigenvalues in any of the subspaces).

A straightforward and important consequence of Prop.~\ref{prop:trivialeigen} is the following characterization of those 3-tuples with no trivial eigenvalues:

\begin{coro} \label{coro}
A 3-tuple $(N^P,N^H,\blam)$ has no trivial eigenvalues, i.e., $T(N^P,N^H,\blam)=0$, if and only if it fulfills the small-block condition.
\end{coro}

\begin{proof}
Combine Def.~\ref{mnumbersmallblocksdef} with Prop.~\ref{prop:trivialeigen}.
\end{proof}

\subsection{\bluecheck Systematic reduction to small-blocks scenario} \label{smallblocksreduction}
%\textcolor{blue}{Make the logic and order of things across the two lemmas as smooth as possible}

In Sec.~\ref{trivialeigen} we saw that when a 3-tuple $(N^P, N^H,\blam)$ does not fulfill the small-blocks condition, there are $T(N^P, N^H,\blam)>0$ trivial eigenvalues $0,1$. One can get rid of those with a systematic reduction method, which is part of the proof of Prop.~\ref{prop:trivialeigen}. After doing so, one ends up with a simpler 3-tuple $(N'^P, N'^H,\blam')$ fulfilling the small-blocks condition. This reduced 3-tuple can be then used to compute the value of $\Tilde{C}^{\blam}_{\max}(N^P,N^H)$ for the initial 3-tuple, as stated in the lemma below. We recall that $m^{H/P}_{\mu}$ is short for $m^{H/P}(N^{P/H},\lambda_{\mu})$ and $M^{H/P}=\sum_{\mu=1}^{|\blam|}m^{H/P}_{\mu}$.

\begin{prop} \label{prop:reduction}

Pick a 3-tuple $(N^P,N^H,\blam)$ with $N^P,N^H \neq 0$ and $\lambda_1 < d=N^P+N^H$. Consider the reduced 3-tuple $(N'^P,N'^H,\blam')$ given by

\begin{align} \nonumber
	N'^P&=N^P-M^P \\ \nonumber
	N'^H&=N^H-M^H \\ 
%   \nonumber    d'&=N'^P + N'^H \\ 
	\lambda'_{\mu}&= \lambda_{\mu} - m^P_{\mu} - m^H_{\mu} \geq 1 ~,
\end{align}
where $\mu =1,...,|\blam|=|\blam'|$. Then
\begin{equation} \label{reductionforImax}
	\Tilde{C}^{\blam}_{\max}(N^P,N^H) = \Tilde{C}^{\blam'}_{\max}(N'^P,N'^H)~.
\end{equation}
Plus, $(N'^P,N'^H,\blam')$ fulfills the small-blocks condition, and thus cannot be further reduced.
\end{prop}

\begin{proof}
We prove Eq.~\eqref{reductionforImax} by simplifying the maximization entailed in $\Tilde{C}^{\blam}_{\max}(N^P,N^H)=\max_{\gamma}\tr \, h(\gamma^B)$, which by definition runs over all rank-$N^P$ projectors on the underlying 1P space $\H_1$. As discussed in the proof of Prop.~\ref{prop:trivialeigen}, using the the \bt gauge' group of block-diagonal unitaries, we can restrict the maximization to one over projectors $\gamma$ which take the normal form Eq.~\eqref{normalgamma}. We recall that any rank-$N'^P$ projector on $\mathcal{H'}$ can be realized as the block $\gamma'$ of a suitable $\gamma$ in normal form. Finally we need to use Eq.~\eqref{specblocknormal} and $h(0)=h(1)=0$. Then the maximization can be reformulated as follows:

\begin{equation} \label{restrictmin}
	\max_{\substack{\gamma \rightarrow  N^P}}\tr \, h(\gamma^B) = \max_{\substack{\gamma \rightarrow N^P\\\text{norm. form}}}\tr \, h(\gamma^B) = \max_{\substack{\gamma' \rightarrow N'^P}}\tr \, h(\gamma'^B)~.
\end{equation}
Here the arrows indicate the rank. The new projectors $\gamma'$ act on $\mathcal{H'}$. $\gamma'^B$ (see Eq.~\eqref{gammaprimeB}) is made up of $|\blam|=|\blam'|$ blocks, each of which acts nontrivially on a $\lambda'_{\mu}$-th dimensional subspace of $\H_1'$. Eq.~\eqref{restrictmin} is by definition the same as Eq.~\eqref{reductionforImax}.

The second part of Prop.~\ref{prop:reduction}, i.e., the fact that $(N'^P,N'^H,\blam')$ fulfills the small-blocks condition, is proved in Sec.~\ref{subsec:new3tuple}.
\end{proof}

\subsection{Reduced 3-tuples fulfill the small-blocks condition} \label{subsec:new3tuple}

We finally show that the 3-tuple reduction $(N^P,N^H,\blam) \rightarrow (N'^P,N'^H,\blam')$, defined in the scheme~\ref{reductionscheme}, is such that the output 3-tuple $(N'^P,N'^H,\blam')$ fulfills the small-blocks condition, Def.~\ref{mnumbersmallblocksdef}. This fact is used in Prop.~\ref{prop:trivialeigen} and is the second statement in Prop.~\ref{prop:reduction}.

With reference to Def.~\ref{mnumbersmallblocksdef} and the shortened notation introduced in Sec.~\ref{trivialeigen}, we recall the reduction procedure~\ref{reductionscheme}:

\begin{equation}
(N^P,N^H,\blam) \rightarrow (N'^P,N'^H,\blam')
\end{equation}
\begin{align} \label{reductionschemebis}
\text{with} ~~~~~~~~~~ N'^P&=N^P-M^P \\ 
	N'^H&=N^H- M^H \\ 
%   \nonumber    d'&=N'^P + N'^H \\ 
	\lambda'_{\mu}&= \lambda_{\mu} - m^P_{\mu} - m^H_{\mu} \geq 1 ~.
\end{align}

We need to prove the statement under the nontriviality requirements $0 < N^P, N^H < d$ and $\lambda_1 < d$, from which $\lambda'_{\mu} \geq 1$ follows. Also, for $\mu<\nu$ we have $\lambda'_{\mu} \geq \lambda'_{\nu}$, i.e., $\blam'$ comes out of the scheme~\ref{reductionschemebis} automatically ordered. Without loss of generality, we check that in the following two cases:\\
%(i) $\lambda_{\mu} \leq N^P, N^H \Rightarrow$ $\lambda'_{\mu} = \lambda_{\mu} \geq \lambda_{\nu} = \lambda'_{\mu}$;
\\(i) $N^H \geq \lambda_{\mu} > N^P \Rightarrow$ $\lambda'_{\mu} \exclaeq N^P \geq \lambda'_{\nu}$;
\\(ii) $\lambda_{\mu}  > N^H,N^P$, which can only happen for $\mu=1$ $\Rightarrow$ $\lambda'_1 \exclaeq d - \lambda_1 \geq \lambda_{\nu} \geq \lambda'_{\nu} $, $\forall \nu \geq 2$.\\\\
Then checking the small-blocks condition on $(N'^P,N'^H,\blam')$ reduces to checking $\lambda'_{1} \leq N'^P, N'^H$.\\
%\emph{Case 0:} $\lambda_1\leq N^P, N^H$. The reduction does not change the 3-tuple, i.e., $(N'^P,N'^H,\blam')=(N^P,N^H,\blam)$, which already fulfills the small-blocks condition.\\

\nit \textbf{Case 1:} $\lambda_1 \leq N^P, N^H$. The 3-tuple $(N^P,N^H,\blam)$ fulfills the small-blocks condition and the reduction~\ref{reductionschemebis} leaves the 3-tuple unchanged.\\

\nit \textbf{Case 2:} $N^P < \lambda_1 \leq N^H$ (the case $N^H < \lambda_1 \leq N^P$ is equivalent). We have $M^P=0$, and hence $N'^P=N^P$. Instead, there are \bt trivial' holes which are removed in the reduction, so that $N'^H<N^H$. Let $\tilde{\mu} \geq 1$ be the biggest index satisfying $\lambda_{\tilde{\mu}} > N^P$. Then
\begin{align}
	\lambda'_{\mu} &= N^P ~~ \text{for } 1\leq \mu \leq \tilde{\mu} \\
	\lambda'_{\mu} &= \lambda_{\mu} ~~~ \text{for }  \tilde{\mu}+1 \leq \mu \leq |\blam|~.
\end{align}
On the one hand, $\lambda'_1 = N^P = N'^P$.\\
On the other hand, we have
\begin{equation}
M^H = \sum_{\mu = 1}^{\tilde{\mu}} m^H_{\mu}= (\sum_{\mu = 1}^{\tilde{\mu}} \lambda_{\mu} )-\tilde{\mu}\cdot N^P
\end{equation}
so that
\begin{align}\nonumber
N'^H &= N^H - M^H = (d - N^P)-(\sum_{\mu = 1}^{\tilde{\mu}} \lambda_{\mu} )+ \tilde{\mu}\cdot N^P \\
& \geq (\tilde{\mu} - 1) \cdot N^P \geq N^P = \lambda'_{1},
\end{align} 
where in the end we assumed $\tilde{\mu} \geq 2$. Instead, if $\tilde{\mu} = 1$, we similarly get $N'^H = d - \lambda_1 \geq N^P= \lambda'_{1}$.\\

\nit \textbf{Case 3:} $\lambda_1 > N^H, N^P  $. We have $m_{1}^{H/P} = \lambda_{1} - N^{P/H}$. Now pick $\mu \geq 2$: one has $\lambda_2 \leq d - \lambda_1 < d - N^{P/H} = N^{H/P}$, so that $m^{P/H}_{\mu} = 0$. Then

\begin{align}
	\lambda'_{1} & = d - \lambda_1 \\
	\lambda'_{\mu} &= \lambda_{\mu} ~~~~~~~ \text{for } 2 \leq \mu \leq |\blam| ~.
\end{align}
One then has $M^{P/H} = m^{P/H}_1 = \lambda_1 - N^{H/P}$, so that $N'^{P/H}=N^{P/H} - M^{P/H} = d-\lambda_1 =\lambda'_1$.

\subsection{A non-triviality example within the small-blocks scenario}\label{nontrivial}

Given the reduction procedure and Prop.~\ref{prop:reduction}, we only have to compute $\Tilde{C}^{\blam}_{\max}$ within the small-blocks scenario. The first key piece of information is given by the \bt only if' direction of Cor.~\ref{coro}: given a 3-tuple $(N^P,N^H,\blam)$ fulfilling the small-blocks condition, the number of trivial eigenvalues is always $0$, i.e., it is possible to find a rank-$N^P$ projector $\gamma$ s.t. $\bb^{\blam}(\gamma) \in (0,1)^d$ (no 0's, no 1's). We then go back to the initial question (Q2) and ask whether it is even possible to get $b^{\blam}_j(\gamma)=N^P/d$ $\forall j=1,...,d$. This is for example the case in the examples (iii)-(v) in Sec.~\ref{SDexamples}. A positive answer would imply that we can compute $\Tilde{C}^{\blam}_{\max}(N^P,N^H)$ for \textit{any} SD spectrum $\bn$ and partition $\blam$. The answer to (Q2) is in general \textit{negative}, though, as shown in the following counterexample. So even within the small-blocks scenario $\Tilde{C}^{\blam}_{\max}$ does \textit{not} necessarily hit the upper bound~\eqref{eq:basic_bounds}.

The counterexample is given by the following choice of 3-tuple:
\begin{equation}
	N^P=2, ~~~ N^H=3, ~~~\blam =(2,2,1) ~.
\end{equation}
It can be in fact shown that no rank-2 projector $\gamma$ on $\mathbb{C}^5$ satisfies $\bb^{\blam}(\gamma) = (\frac{2}{5},...,\frac{2}{5})$. More precisely, $\bb^{\blam}(\gamma)$ cannot even get arbitrarily close to $(\frac{2}{5},...,\frac{2}{5})$. This is confirmed by our result of Sec.~\ref{subsubsec:2,2,1}: the optimal block spectrum for $\bn = (1,1,0,0,0)$ is
\begin{equation}
	\bb_{\max}^{\blam}(\bn) = \left ( \frac{1}{2}, \frac{1}{3} \sbs \frac{1}{2}, \frac{1}{3} \sbs  \frac{1}{3} \right ) \; ,
\end{equation}
which is not maximally mixed. As a consequence,
\begin{equation}
	\Tilde{C}^{\blam}_{\max}(N^P,N^H) < 5 \cdot h\left(\frac{N^P}{5}\right) \til,
\end{equation}
\hideorblue{
- think of shape of $h$...\\
- Actually Kaustav's numerics show that probably $v_{\max}=(1/2,*,1/3,*,*)$, can you prove this? ...interlacing inequalities won't suffice...\\}
i.e., the upper bound~\eqref{eq:basic_bounds} is sometimes not reached, even within the small-block scenario. We chose the counterexample referring to a number of dimensions $d=5$ because no similar counterexamples can be found for $d<5$. That is, for $d<5$ and for any 3-tuple $(N^P, N^H, \blam)$ satisfying $N^P+N^H=d$ and the small-blocks condition, the corresponding SD spectrum $\bn$ always reaches the upper bound~\eqref{eq:basic_bounds}. For $d=5$, the counterexample shown in Sec.~\ref{nontrivial} and the one obtained by exchanging $N^P$ and $N^H$ are the only counterexamples. For $d=6$, again no counterexamples exist.
Are there any patterns? Whenever the bound~\eqref{eq:basic_bounds} is not reached, how can we compute $\Tilde{C}_{\max}^{\blam}$? These questions are currently still open.

\hideorblue{-The goal then would be to find $\Tilde{C}^{\blam}_{\max}(N^P,N^H)$ for any SD spectrum $\bn$, but I do not know if there is a complete answer. If we don't have one, we can show some qualitative behavior, e.g., for all $d<10$ and all $\blam$'s\\
-section Patterns, discussion and conjectures
}

\section{Treatment of superselection rule enforcement} \label{app:SSR}

In this Appendix, we present the mathematical details of Sec.\til\ref{subsec:extensions_SSR}.

\subsection{Superselected density operators} \label{subsec:superselected_dens_ops}

Fix a partition $\blam$ of the 1P space dimension, and define the corresponding subspaces $V_\mu \cong \C^{\lambda_\mu}$ as in the beginning of Sec.~\ref{sec:generalizing}. We aim to properly quantify the multipartite total correlation between the orbital subspaces $V_{\mu}$ which is contained in a given state $\rho$ on the fermionic Fock space. A sound measure for correlation must be based on correlations among outcomes of measurements that are local (i.e., involving orbitals of a single subspace $V_{\mu}$ only) and physically accessible, i.e., in compliance with the local-number SSR. The mathematical counterpart of the measurement constraints dictated by the SSR is the following set of physically accessible operators, which is a subalgebra of the total algebra $\mathcal{A}$ of operators on the Fock space:
\begin{equation}
	\mathcal{A}^{SSR} = \{ \mathcal{O}^{SSR} \in \mathcal{A} ~ | ~ [\mathcal{O}^{SSR},N_{\mu}] = 0 ~ \forall \mu = 1,...,|\blam|\}~,
\end{equation}
where $N_{\mu}$ is the operator for the local number of particles in $V_{\mu}$. Consider now the physical situation where only observables in $\mathcal{A}^{SSR}$ can be measured. Then we correspondingly restrict the action of states $\rho$ defined on the original algebra $ \mathcal{A}$ to the new subalgebra $ \mathcal{A}^{SSR}$. But then states get naturally grouped into equivalence classes: given $\rho,\rho'$ states on $ \mathcal{A}$,
\begin{align} \nonumber
\rho \sim \rho' &:\Leftrightarrow
\Tr[\rho \mathcal{O}^{SSR}] = \Tr[\rho' \mathcal{O}^{SSR}] ~ \forall \mathcal{O}^{SSR} \in  \mathcal{A}^{SSR} \\ \nonumber
&\stackrel{\text{!}}{\Leftrightarrow} \Pi_{1}^{m_1} \ldots \Pi_{|\blam|}^{m_{|\blam|}}\rho \Pi_{1}^{m_1} \ldots \Pi_{|\blam|}^{m_{|\blam|}} \\ \nonumber 
&= \Pi_{1}^{m_1} \ldots \Pi_{|\blam|}^{m_{|\blam|}}\rho' \Pi_{1}^{m_1} \ldots \Pi_{|\blam|}^{m_{|\blam|}} \\
& \forall \text{ local number eigenvalues } m_1,\ldots,m_{|\blam|} ~,
\end{align}
where $\Pi_{\mu}^{m_\mu}$ is the projector onto the eigenspace of $N_{\mu}$ relative to the eigenvalue $m_\mu$. Two equivalent states \emph{cannot be distinguished} by any measurements or more general operations that are described by operators in $ \mathcal{A}^{SSR}$. These equivalence classes are (or better: one-to-one correspond to) the states on $ \mathcal{A}^{SSR}$. An equivalent and more handy way to refer to such states on $ \mathcal{A}^{SSR}$ is the following. Consider this \bt projection' for density operators:
\begin{equation} \label{eq:SSR_state_projection}
	\rho \rightarrow \rho^{SSR} := \sum_{m_1,...,m_{|\blam|}} \Pi_{1}^{m_1} \ldots \Pi_{|\blam|}^{m_{|\blam|}}\rho \Pi_{1}^{m_1} \ldots \Pi_{|\blam|}^{m_{|\blam|}} ~.
\end{equation}
Then projected density operators, i.e., $\rho$ such that $\rho = \rho^{SSR}$, one-to-one correspond to states on  $\mathcal{A}^{SSR}$. We can rewrite $\rho^{SSR}$ as
\begin{equation} \label{eq:rhoSSR_as_mixture}
	\rho^{SSR} = \sum_{m_1 , \ldots, m_{|\blam|}} p_{m_1 , \ldots , m_{|\blam|}} \rho_{m_1 , \ldots ,m_{|\blam|}} \; ,
\end{equation}
where
\begin{equation} \label{eq:SSR_probs}
	p_{m_1, \ldots, m_{|\blam|}} = \tr [\rho \Pi_{1}^{m_1} \ldots \Pi_{|\blam|}^{m_{|\blam|}}]
\end{equation}
is the probability of measuring local numbers $m_\mu$ on the state $\rho$, and
\begin{equation} \label{eq:SSR_blocks}
	\rho_{m_1 , \ldots , m_{|\blam|}} =  \Pi_{1}^{m_1} \ldots \Pi_{|\blam|}^{m_{|\blam|}}\rho \Pi_{1}^{m_1} \ldots \Pi_{|\blam|}^{m_{|\blam|}} / p_{m_1, \ldots, m_{|\blam|}}
\end{equation}
is the corresponding post-measurement state. We remark the implicit dependence of the new quantities, and in particular of $\rho^{SSR}$, on the partition $\blam$ and the corresponding subspaces $V_\mu$.

\subsection{SSR-adapted total correlation} \label{subsec:SSR_mutualinfo}
Refs.~\cite{wiseman_vaccaro_ferreting, gesa_lexin_2024} provide the following operational interpretation of the projection $\rho \rightarrow \rho^{SSR}$, which is fundamental in our context. Suppose that one aims to extract the orbital entanglement by \bt swapping' the state $\rho$ of the fermionic orbitals and the state of some qubits. Then it turns out that the swap only partially succeeds, due to to the fermionic local-number SSR, and that the resulting state of the qubits is a copy of $\rho^{SSR}$ instead of $\rho$. This directly leads to the following operational definition of an SSR-adapted total correlation:
\begin{equation}
	C^{\blam, SSR}_{proj} (\rho) : = C^{\blam} (\rho^{SSR}) \; .
\end{equation}
Another valid option, advocated in Ref.~\cite{wisemanvaccaro}, is to take the average of the total correlation over all `blocks' $\rho_{m_1 , \ldots , m_{|\blam|}}$ (see Eqs.~\eqref{eq:SSR_probs} and~\eqref{eq:SSR_blocks}):
\begin{equation}
	C^{\blam, SSR}_{block} (\rho) : = \sum_{m_1 , \ldots, m_{|\blam|}} p_{m_1 , \ldots , m_{|\blam|}} C^{\blam} (\rho_{m_1 , \ldots , m_{|\blam|}}) \; .
\end{equation}
This quantity is the figure of merit for scenarios where the entanglement extraction protocol implies a measurement of the local particle numbers. For pure states, it coincides with twice the entanglement measure $E_P$ introduced in Ref.~\cite{wisemanvaccaro}.

We now state some basic properties of $C^{\blam, SSR}_{proj}$ and $C^{\blam, SSR}_{block}$. Pick a density operator $\rho$ which is global number-superselected, i.e., $[\rho, N]=0$. We denote the orbital-reduced states of $\rho^{SSR}$ ($\rho$) as $\rho^{SSR}_\mu$ ($\rho_\mu$), where $\mu=1,\ldots,|\blam|$. Then computing $C^{\blam, SSR}_{proj}(\rho) = (\sum_{\mu = 1}^{|\blam|} S(\rho^{SSR}_{\mu})) - S(\rho^{SSR})$ is facilitated by the following proposition.
\begin{prop}~\label{prop:marginals_of_SSR_states}
Pick a state $\rho$ on the Fock space satisfying $[\rho, N] = 0$ and a partition $\blam$ of the 1P space dimension. Then
\begin{equation}\label{reducedstateunchanged}
	\rho^{SSR}_\mu = \rho_\mu \; \; \;  \mu=1,\ldots,|\blam| \;.
\end{equation}
\end{prop}
\begin{proof}

Eq.~\eqref{reducedstateunchanged} can be expected based on the fact that for any fermionic state $\rho$ which is global number-superselected, i.e., $[\rho, N]=0$, the orbital-reduced state $\rho_{V}$ relative to any subspace $V$ of the 1P space is automatically local number-superselected, i.e., $[\rho_{V}, N_{V}]=0$ (even when $[\rho,N_V] \neq 0$). We prove this result in Prop.~\ref{prop:rho_ssr_implies_rhoA_ssr}.

To show Eq.~\eqref{reducedstateunchanged}, we use the fundamental definition of orbital-reduced states in terms of expectation values, see~\cite{Friis_2016}. Let $\mathcal{A}_{\mu} \subset \mathcal{A} $ be the subalgebra of operators relative to the 1P subspace $V_{\mu} \subset \H_1$, i.e., the subalgebra generated by $\{f_{|a\rangle},f^{\dagger}_{|a\rangle}\}_{|a\rangle \in V_{\mu}}$ via sums and products. Pick an operator $A \in \mathcal{A}_{\mu}$. Then we use the local number projectors $\Pi_{\mu}^{m_\mu}$ to write
\begin{align} \nonumber
	& \rho A = \sum_{\substack{m_1 , \dots, m_{| \blam |} \\ m'_1 , \dots, m'_{| \blam |} }}  \Pi^{m_1}_1 \dots \Pi^{m_{|\blam|}}_{| \blam|} \rho \Pi^{m'_1}_1 \dots \Pi^{m'_{|\blam|}}_{| \blam|} A \\ \nonumber
	& = \sum_{\substack{m_1 , \dots, m_{| \blam |} \\ m'_1 , \dots, m'_{| \blam |} }} \Pi^{m_1}_1 \dots \Pi^{m_{|\blam|}}_{| \blam|} \rho  \\
	& \;\;\;\;\;\;\;\;\;\;\;\;\;\;\;\;\;\;\;\;\; \cdot \Pi^{m'_\mu}_\mu A \Pi^{m'_1}_1 \dots \Pi^{m'_{\mu-1}}_{\mu-1}\Pi^{m'_{\mu+1}}_{\mu+1} \dots \Pi^{m'_{|\blam|}}_{| \blam|} ~,
\end{align}
where for $\nu =1,...,\mu-1,\mu+1,...,|\blam| $ we have commuted $\Pi_{\nu}^{m'_{\nu}}$ and $A$. But then when we take the trace $\Tr[\rho A]$, due to the cyclicity of the trace, we get nonzero contributions only from summands relative to indices $m_{\nu}=m'_{\nu} ~ \forall \nu \neq \mu$. We can then write
\begin{align} \nonumber
&\Tr[\rho A] = \\
&\sum_{\substack{m_1 , \dots, m_{| \blam |} \\ m'_\mu}} \tr \left [ \Pi^{m_1}_1 \dots  \Pi^{m_\mu}_\mu \dots \Pi^{m_{|\blam|}}_{| \blam|} \rho \Pi^{m_1}_1 \dots  \Pi^{m'_\mu}_\mu \dots \Pi^{m_{|\blam|}}_{| \blam|} A \right ] ~.
\end{align}
Now, $\Pi^{m_1}_1 \dots  \Pi^{m_\mu}_\mu \dots \Pi^{m_{|\blam|}}_{| \blam|} \rho \Pi^{m_1}_1 \dots  \Pi^{m'_\mu}_\mu \dots \Pi^{m_{|\blam|}}_{| \blam|}$ vanishes for $m_{\mu}\neq m'_{\mu}$. In fact, the range of $\Pi^{m_1}_1 \dots  \Pi^{m_\mu}_\mu \dots \Pi^{m_{|\blam|}}_{| \blam|}$ is a subspace of the $m_1 + \dots + m_\mu + \dots + m_{|\blam|}$-eigenspace of the global number $N$; the range of $\Pi^{m_1}_1 \dots  \Pi^{m'_\mu}_\mu \dots \Pi^{m_{|\blam|}}_{| \blam|}$ is a subspace of the $m_1 + \dots + m'_\mu + \dots + m_{|\blam|}$-eigenspace of $N$; and $\rho$ leaves invariant the eigenspaces of $N$. We are then left with
\begin{align} \nonumber
\Tr[\rho A] &= \sum_{m_1 , \dots, m_{| \blam |}} \tr \left [ \Pi^{m_1}_1 \dots \Pi^{m_{|\blam|}}_{| \blam|} \rho \Pi^{m_1}_1 \dots   \Pi^{m_{|\blam|}}_{| \blam|} A \right ]\\ \nonumber
&= \Tr \Big [\underbrace{\sum_{m_1 , \dots, m_{| \blam |}} \Pi^{m_1}_1 \dots \Pi^{m_{|\blam|}}_{| \blam|} \rho \Pi^{m_1}_1 \dots   \Pi^{m_{|\blam|}}_{| \blam|} }_{\rho^{SSR}} A  \Big ] \\
& = \Tr[\rho^{SSR} A] ~~~~ \forall  A\in \mathcal{A}_{\mu}~,
\end{align}
which is equivalent to Eq.~\eqref{reducedstateunchanged}.
\end{proof}
The following corollary is a straightforward consequence of Prop.~\ref{prop:marginals_of_SSR_states}.
\begin{coro} \label{cor:I_proj}
	For any state $\rho$ satisfying $[\rho, N] = 0$,
	\begin{align} \nonumber
		C_{proj}^{\blam, SSR} (\rho) &= \left ( \sum_{\mu = 1}^{|\blam|} S(\rho_{\mu}) \right ) - S(\rho^{SSR}) \\ \label{eq:proj_m_info_simplified}
		&= C^{\blam}(\rho)+S(\rho)-S(\rho^{SSR}) \; .
	\end{align}
\end{coro}

We now state some basic bounds among the three variants of the total correlation.
\begin{prop} \label{prop:SSR_bounds}
	For any fermionic state $\rho$ satisfying $[\rho, N] = 0$,
	\begin{equation} \label{eq:target_bounds_ssr}
	C^{\blam}(\rho) \geq C^{\blam, SSR}_{proj}(\rho) \geq C^{\blam, SSR}_{block}(\rho) \; .
\end{equation}
\end{prop}
\begin{proof}		
	For the first bound, the fact that $\rho^{SSR}$ is a projected version of $\rho$ implies that spec$(\rho) \succ $ spec$(\rho^{SSR})$, implying $S(\rho) - S(\rho^{SSR}) \leq 0$. $C^{\blam}(\rho) \geq C^{\blam, SSR}_{proj}(\rho)$ then follows from Eq.\til\eqref{eq:proj_m_info_simplified}.

	To prove the second bound, our starting point is Eq.\til\eqref{eq:rhoSSR_as_mixture}. Taking orbital-partial traces yields, for $\mu = 1, \dots, | \blam |$,
	\begin{equation}
		\rho_\mu^{SSR} = \rho_\mu = \sum_{m_\mu} p_{\mu, m_\mu} \rho_{\mu ,m_\mu} \; ,
	\end{equation}
	with
	\begin{align}
		& p_{\mu, m_\mu} = \sum_{\substack{m_1, \dots, m_{\mu-1}, \\ m_{\mu + 1} , \dots, m_{| \blam |}}} p_{m_1, \dots, m_{\mu}, \dots, m_{| \blam |}} \\ \nonumber
		& \rho_{\mu ,m_\mu} = \\
		& = \sum_{\substack{m_1, \dots, m_{\mu-1}, \\ m_{\mu + 1} , \dots, m_{| \blam |}}} \frac{p_{m_1, \dots, m_{\mu}, \dots, m_{| \blam |}} }{p_{\mu, m_\mu}} \rho_{\mu, m_1, \dots, m_{\mu}, \dots, m_{| \blam |}}  \; ,
	\end{align}
	where $\rho_{\mu, m_1, \dots, m_{\mu}, \dots, m_{| \blam |}}$ is the $\mu$-th reduced state of $\rho_{m_1, \dots, m_{\mu}, \dots, m_{| \blam |}}$.
	Importantly, the representation\til\eqref{eq:rhoSSR_as_mixture} of $\rho^{SSR}$ features summands $\rho_{m_1 , \ldots ,m_{|\blam|}}$ whose supports are mutually orthogonal. The same holds true for the summands $\rho_{\mu, m_\mu}$ of $\rho_\mu$. This allows us to simplify the entropies:
	\begin{align} \nonumber
		S(\rho^{SSR})& = H(\{ \bp \}) + \sum_{m_1, \dots, m_{|\blam|}} p_{m_1 , \ldots ,m_{|\blam|}} S(\rho_{m_1 , \ldots ,m_{|\blam|}}) \\
		S(\rho_\mu)& = H(\{ \bp_\mu \}) + \sum_{m_\mu} p_{\mu, m_\mu} S(\rho_{\mu, m_\mu}) \; ,
	\end{align}
	where $H$ here is the \emph{standard} (rather than binary) classical Shannon entropy. We can then write
	\begin{align} \nonumber
		& C(\rho^{SSR}) = \left ( \sum_{\mu = 1}^{| \blam |} S(\rho_\mu) \right ) - S(\rho^{SSR})\\ \nonumber
		& = \left ( \sum_{\mu = 1}^{| \blam |} H(\bp_\mu)  \right ) - H(\{ \bp \})\\ \nonumber
		& \,\,\, + \left ( \sum_{\mu = 1}^{| \blam |} \sum_{m_\mu} p_{\mu, m_\mu} S(\rho_{\mu, m_\mu}) \right ) \\ \label{eq:medium_align}
		& - \sum_{m_1, \dots, m_{|\blam|}} p_{m_1 , \ldots ,m_{|\blam|}} S(\rho_{m_1 , \ldots ,m_{|\blam|}}) \; .
	\end{align}
	In the last expression, the first line contains the classical total correlation of the joint probability distribution $\bp$ (in fact, $\bp_\mu$ are its marginals) and is thus nonnegative. Using
	\begin{align} \nonumber
		&S(\rho_{\mu, m_\mu}) \geq  \\
		& \sum_{\substack{m_1, \dots, m_{\mu-1}, \\ m_{\mu + 1} , \dots, m_{| \blam |}}}  \frac{p_{m_1, \dots, m_{\mu}, \dots, m_{| \blam |}} }{p_{\mu, m_\mu}} S(\rho_{\mu, m_1, \dots, m_{\mu}, \dots, m_{| \blam |}} )
	\end{align}
	we can bound the last two lines of Eq.\til\eqref{eq:medium_align} from below with
	\begin{align} \nonumber
			&\sum_{m_1, \dots, m_{| \blam |}}  p_{m_1, \dots, m_{| \blam |}} \cdot \\ \nonumber
			& \,\,\,\,\,\;\;\;\;\;\;\;\;\;\;\;\; \cdot \left [ \left ( \sum_{\mu = 1}^{| \blam |}  S(\rho_{ \mu, m_1, \dots, m_{| \blam |}}  ) \right ) - S(\rho_{m_1 , \ldots ,m_{|\blam|}}) \right ] \\ \nonumber
		& = \sum_{m_1, \dots, m_{| \blam |}}  p_{m_1, \dots, m_{| \blam |}} C^{\blam} (\rho_{m_1 , \ldots ,m_{|\blam|}}) \\
		& = C^{\blam, SSR}_{block} ( \rho ) \; .
	\end{align}
	This proves the second bound in Eq.\til\eqref{eq:target_bounds_ssr}.
\end{proof}

We are interested in the maximization of the modified total correlation with respect to 1P unitary rotations. Following Eq.~\eqref{eq:Iblam_max_def}, we define
\begin{align}
	C^{\blam, SSR}_{proj, \, \max}(\rho) & :=\max_{U}C^{\blam, SSR}_{proj}(U\, \rho \, U^{\dagger}) \\
	C^{\blam, SSR}_{block, \, \max}(\rho) & :=\max_{U}C^{\blam, SSR}_{block}(U\, \rho \, U^{\dagger})~,
\end{align}
where the maximization is over all 1P unitaries $U$ acting on the Fock space. Prop.~\ref{prop:SSR_bounds} directly implies the following bounds: if $[\rho, N] = 0$, then
\begin{equation} \label{eq:SSR_max_ineq}
	C^{\blam}_{\max}(\rho) \geq C^{\blam, SSR}_{proj, \, \max}(\rho) \geq C^{\blam, SSR}_{block, \, \max}(\rho) \; .
\end{equation}
We now restrict to \textit{free} states $\rho$, which have SDs as eigenstates and thus do commute with $N$.

Let us first consider $C^{\blam, SSR}_{proj, \, \max}(\rho)$. The terms $S(\rho^{SSR}_\mu) = S(\rho_\mu)$ can be computed using Prop.~\ref{prop:marginalsoffreestates} and Prop.~\ref{prop:entropyfreestates}, as in the SSR-agnostic problem, as $\rho^{SSR}_\mu = \rho_\mu$ is free. Then we need to understand the dependence of $S(\rho^{SSR})$ on $\rho$'s 1RDM $\gamma$. New challenges are expected to arise from the fact that $\rho^{SSR}$ is \emph{not a free state}, and therefore $S(\rho^{SSR})$ depends on $\rho$'s 1RDM $\gamma$ not only via its spectrum. We will show some examples in the rest of this section.

We now turn to $C^{\blam, SSR}_{block, \, \max}(\rho)$. For any fermionic state $\rho$, $C^{\blam, SSR}_{block}(\rho)$---and thus $C^{\blam, SSR}_{block, \, \max}(\rho)$, too---\emph{vanishes if the partition takes the form $\blam = (\lambda_1, 1, \dots, 1)$}, i.e., if the second, third, \dots, $|\blam|$-th subset of orbitals contain only one orbital each. In fact, in such cases the $\mu$-th local space for fixed local number eigenvalue $m_\mu \in \{ 0,1\}$ is just one-dimensional (for $\mu > 1$). Therefore, each of the projected states $\rho_{m_1, \dots, m_{|\blam|}}$ is a product state over the $|\blam|$ orbital subsystems, and thus $C^{\blam}(\rho_{m_1, \dots, m_{|\blam|}}) = 0$. For $d=2,3$, all partitions take the form $\blam = (\lambda_1, 1, \dots, 1)$. The first nontrivial partition is $\blam = (2,2)$ for $d=4$, which we will address in Sec.~\ref{subsec:SSR_2_2_SD} for SD states.

\subsection{$d=2$, \texorpdfstring{$\boldsymbol{\lambda}$}{lambda} $ = (1,1)$}

We now set out to compute $C^{\blam, SSR}_{proj, \, max}(\rho)$ for free states $\rho$ for the simplest case $d = 2$, $\blam = (1,1)$. Given the reference basis $\{ |1 \rangle, | 2 \rangle \}$ of $\H_1$, we denote with $\{ |0\rangle,|1 \rangle,|2 \rangle,|F \rangle \}$ the corresponding basis of the Fock space, where the phases of the vacuum $\ket{0}$ and the full state $\ket{F} \propto f_1^\dagger f_2^\dagger \ket{0}$ are unimportant.

A free state $\rho$ is specified by its 1P-reduced operator $\gamma$. The SSR-agnostic total correlation reads
\begin{align}\nonumber
C^{\blam}(\rho)  &= S(\rho_{1})+S(\rho_{2})-S(\rho)  \\ \nonumber
&= h(\gamma_{11}) + h(\gamma_{22}) \\ \label{nonsuperselected}
&- h(n_1)-h(n_2)~,
\end{align}
where the $n_j$'s are the eigenvalues of $\gamma$ and $\gamma_{jk}$ are its matrix elements relative to the reference 1P basis. To compute $C^{\blam, SSR}_{block}(\rho)$, we first express $\rho$ in terms of $\gamma$:
\begin{align}\nonumber
\rho &= n_1 n_2 | F \rangle \langle F |  + (1 - n_1) (1 - n_2)  | 0 \rangle \langle 0 |+ \\
&+ n_1(1 - n_2) | \phi_1 \rangle \langle \phi_1 |+ n_2(1 - n_1) | \phi_2 \rangle \langle \phi_2 | \til,
\end{align}
where $| \phi_j \rangle$ are orthonormal eigenvectors of $\gamma$. Then we project $\rho \rightarrow \rho^{SSR}$ as in Eq.~\eqref{eq:SSR_state_projection}:
\begin{align} \nonumber
	\rho^{SSR} &= n_1 n_2 | F \rangle \langle F | + (1 - n_1) (1 - n_2)  | 0 \rangle \langle 0 |  \\ \nonumber
	&+ n_1 (1 - n_2) \Big(|\langle 1 | \phi_1 \rangle|^2 | 1 \rangle \langle 1 | + |\langle 2 | \phi_1 \rangle|^2 | 2 \rangle \langle 2 | \Big) \\ \nonumber
	&+n_2 (1 - n_1)\Big(|\langle 1 | \phi_2 \rangle|^2 | 1 \rangle \langle 1 | + |\langle 2 | \phi_2 \rangle|^2 | 2 \rangle \langle 2 | \Big) \\ \nonumber
	&= n_1 n_2 | F \rangle \langle F | + (1 - n_1) (1 - n_2)  | 0 \rangle \langle 0 |  \\ \nonumber
	&+\underbrace{\Big ( n_1 (1 - n_2) |\langle 1 | \phi_1 \rangle|^2 + n_2 (1 - n_1) |\langle 1 | \phi_2 \rangle|^2 \Big)}_{\alpha} | 1 \rangle \langle 1| \\ 
	&+\underbrace{\Big ( n_1 (1 - n_2) |\langle 2 | \phi_1 \rangle|^2 + n_2 (1 - n_1) |\langle 2 | \phi_2 \rangle|^2 \Big)}_{\beta} | 2 \rangle \langle 2|~.
\end{align}
The eigenvalues $\alpha$ and $\beta$ can be expressed as
\begin{align} \nonumber
\alpha&=\gamma_{11}-n_1 n_2\\
\beta&=\gamma_{22}-n_1 n_2~.
\end{align}
As anticipated, the state $\rho^{SSR}$ is in general \emph{not free}, and its entropy is not a function of the $n_j$'s.

The single orbital-reduced density operators of $\rho^{SSR}$ are the same as those of $\rho$. Their matrix representation relative to the particle-hole basis reads
\begin{equation}
\rho^{SSR}_{j} = \rho_{j} \rightarrow \begin{pmatrix}
\gamma_{jj} & 0 \\
0 & 1-\gamma_{jj} \\
\end{pmatrix}
\text{ for }j=1,2~.
\end{equation}

We can finally compute, and maximize over all $\gamma$ with fixed spectrum $(n_1,n_2)$,
\begin{align} \nonumber
	&C^{\blam, SSR}_{proj}(\rho) = C^{\blam}(\rho^{SSR})\\ \nonumber
	&=S(\rho_{1})+S(\rho_{2})-S(\rho^{SSR})\\ \nonumber
	&= n_1 n_2 \log(n_1 n_2)+ (1-n_1)(1- n_2 ) \log((1-n_1)(1- n_2 ))\\
	&+h(\gamma_{11})+h(\gamma_{22})+ \alpha \log (\alpha) +\beta \log (\beta)~.
\end{align}

In the last expression, terms in the first line only depend on the spectrum of $\gamma$ and can be omitted in the optimization. We then have to maximize the remaining terms, to be compared with Eq.~\eqref{nonsuperselected}:
\begin{align} \nonumber
&h(\gamma_{11})+h(\gamma_{22})+ \\ \label{superselected}
&(\gamma_{11} - n_1 n_2) \log (\gamma_{11} - n_1 n_2) +(\gamma_{22} - n_1 n_2) \log (\gamma_{22} - n_1 n_2)~.
\end{align}
The first line is the function we had to optimize in the non-superselected context. The second line contains new terms. We still have no dependence on off-diagonal matrix elements of $\gamma$. Therefore, the problem is still about optimizing a function of the block spectrum $\bb = (\gamma_{11}\sbs\gamma_{22})$. On its domain $\bset^{\blam}_{\bn}$, this function is still Schur-concave (but, interestingly, not concave, in contrast with the non-superselected case: see Fig.~\ref{SSRcorrelationplot}). Thus, it reaches its maximum value for $\gamma_{11}=\gamma_{22} = (n_1 + n_2)/2$. In particular, this means that the 1P unitary to be applied to an initial free state to maximize the orbital correlation is the same whether the local-number SSR is taken into account or not. The maximal SSR-aware total correlation follows as
\begin{align} \nonumber
	&C^{\blam, SSR}_{proj, max}(\rho) = \\ \nonumber
	&= n_1 n_2 \log(n_1 n_2)+ (1-n_1)(1- n_2 ) \log((1-n_1)(1- n_2 ))\\ \nonumber
	&+2 h \left (\frac{n_1 + n_2}{2} \right ) \\
	&+(n_1+n_2-2 n_1 n_2) \log \left (\frac{n_1+n_2}{2} - n_1 n_2 \right )~.
\end{align}

We compare $C^{\blam}$ and $C^{\blam, SSR}_{proj}$ in Fig.~\ref{SSRcorrelationplot}. There, we can see the effects of enforcing the SSR: the total correlation gets reduced by a factor of two or more, for any free state, while the slices at fixed $\bn$ maintain the same shape.
\begin{figure*}
\centering
\includegraphics[width=0.83\textwidth]{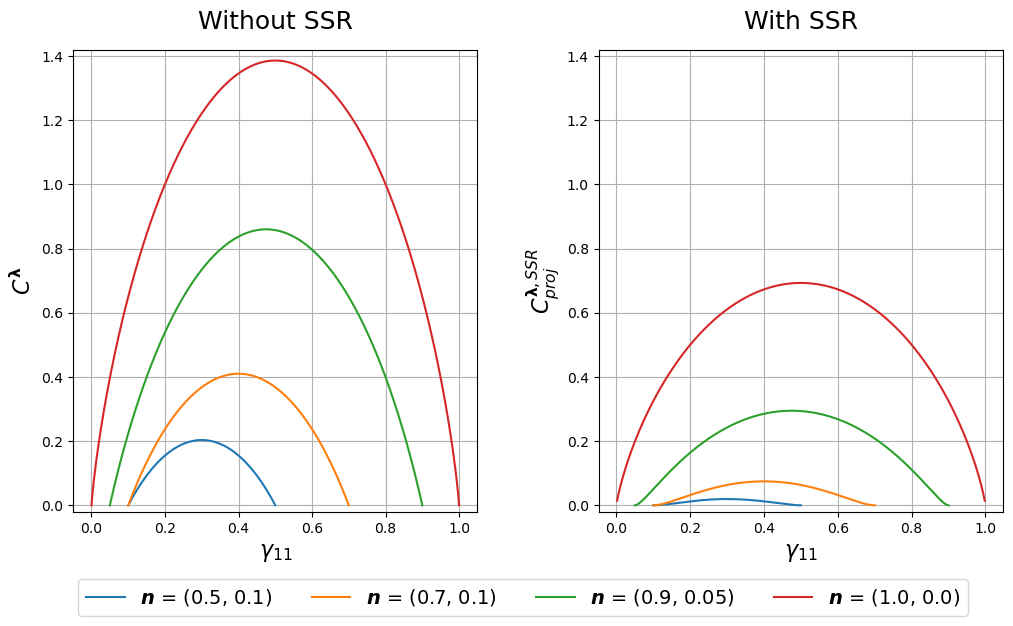}
\caption{The quantities $C^{\blam}$ (left panel) and $C^{\blam, SSR}_{proj}$ (right panel) for $d = 2, \blam = (1,1)$ and for free states, here described by their 1RDM elements $\gamma_{jk}$. For given NONs $\bn$, the constraint $\gamma_{11} + \gamma_{22} = n_1 + n_2$ allows us to express $C^{\blam}$ and $C^{\blam, SSR}_{proj}$ as functions of $\gamma_{11}$ only.}
\label{SSRcorrelationplot}
\end{figure*}

\subsection{$d=3$, \texorpdfstring{$\boldsymbol{\lambda}$}{lambda} $ = (2,1)$}

We now consider the second-to-easiest case, namely $d=3$, $\blam=(2,1)$. $C^{\blam, SSR}_{block}(\rho)$ vanishes for any $\rho$ (see Sec.~\ref{subsec:SSR_mutualinfo}), so we consider $C^{\blam, SSR}_{proj}$. Consider a free state $\rho$ and let $ \gamma = \sum_{j=1}^{3}n_j |\phi_j \rangle \langle \phi_j |$ be the spectral decomposition of its 1P-reduced operator. We define $\overline{\gamma}=\mathbb{I}-\gamma,\overline{n_j}=1-n_j,\delta = \text{det}(\gamma),\overline{\delta} = \text{det}(\overline{\gamma})$. For simplicity, we assume $\delta, \overline{\delta} \neq 0$. We will work with the Fock space basis made up of the following 8 vectors: $|0\rangle,|F\rangle=f^{\dagger}_1f^{\dagger}_2f^{\dagger}_3|0\rangle, |j^P\rangle=f^{\dagger}_j |0\rangle,|j^H\rangle=f_j |F\rangle$. Then $\rho$ can be written as
\begin{align} \nonumber
	\rho &= \overline{\delta} |0 \rangle \langle 0 | \\ \nonumber 
	&+ \overline{\delta} \sum_{j} \frac{n_j}{\overline{n}_j} |\phi_j^P \rangle \langle \phi_j^P | \\ \nonumber
	&+  \delta \sum_j \frac{\overline{n}_j}{n_j} |\phi_j^H \rangle \langle \phi_j^H | \\
	&+ \delta |F \rangle \langle F | ,
\end{align}
where the first line describes the part of $\rho$ with total number of fermions $N=0$, the second line corresponds to $N=1$, and so on. A convenient rewriting is
\begin{align} \nonumber
	\rho &= \overline{\delta} |0 \rangle \langle 0 | \\ \nonumber
	&+ \overline{\delta} \sum_{j,k} (\overline{\gamma}^{-1}-\mathbb{I})_{jk} |j^P \rangle \langle k^P | \\ \nonumber
	&+  \delta \sum_{j,k} (\gamma^{-1}-\mathbb{I})_{kj} |j^H \rangle \langle k^H | \\
	&+ \delta |F \rangle \langle F | \; .
\end{align}
The second (third) line should be thought of as an operator acting on the one-particle (one-hole) three-dimensional subspace as $\overline{\gamma}^{-1}-\mathbb{I}$ ($\gamma^{-1, T}-\mathbb{I}$), up to a prefactor $\overline{\delta}$ ($\delta$); and sending to zero vectors orthogonal to the one-particle (one-hole) subspace.

Computing $\rho^{SSR}$ then turns out to be all about keeping the diagonal blocks of said operators, e.g., $\overline{\gamma}^{-1}\rightarrow (\overline{\gamma}^{-1})^B$, using the same transformation as in Eq.~\eqref{eq:gammablock}:
\begin{align} \nonumber
	\rho^{SSR} &= \overline{\delta} |0 \rangle \langle 0 | \\ \nonumber
	&+ \overline{\delta} \sum_{j,k} \big((\overline{\gamma}^{-1})^B-\mathbb{I}\big)_{jk} |j^P \rangle \langle k^P | \\ \nonumber
	&+  \delta \sum_{j,k} \big((\gamma^{-1})^B-\mathbb{I}\big)_{kj} |j^H \rangle \langle k^H | \\
	&+ \delta |F \rangle \langle F | .
\end{align}
This implies that the eight eigenvalues $\omega_i$ of $\rho^{SSR}$ are given by 
\begin{align} \nonumber
	\omega_1 &= \overline{\delta} \\ \nonumber
	(\omega_2,\omega_3,\omega_4) & =\overline{\delta} ~ \times \textup{ block spectrum of $\overline{\gamma}^{-1}-\mathbb{I}$} \\ \nonumber
	(\omega_5,\omega_6,\omega_7) & =\delta ~ \times \textup{ block spectrum of $\gamma^{-1}-\mathbb{I}$} \\
	\omega_8 & = \delta \; .
\end{align}
These eigenvalues $\omega_i$ can then be used to compute 
\begin{equation} \label{eq:entropy_omega_i}
	S(\rho^{SSR})=-\sum_{i=1}^8\omega_i \log(\omega_i),
\end{equation}
which appears in $C^{\blam, SSR}_{proj}(\rho)$, see Eq.~\eqref{eq:proj_m_info_simplified}.

The natural question arises whether the eigenvalues $\omega_i$, and thus $C^{\blam, SSR}_{proj}(\rho)$, can be expressed solely in terms of the spectra $\bn$ of $\gamma$ and $\bb$ of $\gamma^B$. A numerical analysis has suggested that this is still the case, as for $C^{\blam}(\rho)$. However, we conjecture that this is due to the low dimensionality, $d=3$; and that for higher $d$ examples exist where $C^{\blam, SSR}_{proj}$ cannot be recast as a function over $\bset^{\blam}_{\bn}$.

Notwithstanding, we remark that $C^{\blam, SSR}_{proj}(\rho)$ is significantly harder to treat analytically than $C^{\blam}(\rho)$, even for the present low-dimensional example. For the maximization yielding $C^{\blam,SSR}_{proj,max}(\rho)$, we resorted to a numerical approach. This confirmed that the 1P unitary to be applied to an initial free state to maximize the orbital correlation is the same whether the local-number SSR is taken into account or not. For example, if the initial 1RDM is diag$(n_1, n_2, n_3)$ and $n_1 \geq n_2 \geq n_3$, then the (non-unique) optimal unitary is a discrete Fourier transform on span$(\be_1, \be_3)$. Equivalently, the (non-unique) optimally rotated 1RDM is given by Eq.~\eqref{eq:best_matrix_2,1}. Correspondingly, the numerical result for $C^{\blam,SSR}_{proj,max}(\rho)$ in terms of the NONs of $\rho$ is
\begin{equation}
	C^{\blam,SSR}_{proj,max}(\rho) \stackrel{\text{num.}}{=} 2 h \left ( \frac{n_1 + n_3}{2} \right )  + h(n_2) + \sum_{i=1}^8\omega_i \log(\omega_i) \; ,
\end{equation}
where the eigenvalues $\omega_i$ of the projected optimally rotated free state are
\begin{align} \nonumber
	\omega_1 &= \overline{\delta} \\ \nonumber
	\omega_2 = \omega_3  &= \overline{\delta} \left ( \frac{2-n_1-n_3}{2 \overline{\Delta}}  - 1 \right ) \\ \nonumber
	\omega_4 &= \overline{\delta} \left ( \frac{1}{1 - n_2} - 1 \right ) \\ \nonumber
	\omega_5 = \omega_6 &= \delta \left ( \frac{n_1 + n_3}{2 \Delta} - 1 \right )\\ \nonumber
	\omega_7 &= \delta \left ( \frac{1}{n_2} - 1 \right ) \\
	\omega_8 &= \delta \; .
\end{align}
Here, $\delta = n_1 n_2 n_3$, $\overline{\delta} = (1 - n_1)(1- n_2)(1- n_3)$, $\Delta = n_1 n_3$ and $\overline{\Delta} = (1-n_1)(1-n_3)$.

\subsection{$d=4$, \texorpdfstring{$\boldsymbol{\lambda}$}{lambda} $ = (2,2)$: SD states of two fermions} \label{subsec:SSR_2_2_SD}

Next, we consider the case $d = 4$, $\blam = (2,2)$. This is the simplest case where $C^{\blam, SSR}_{block}$ does not vanish for all states. We focus on SDs $\ket{SD}$ of $N = 2$ fermions, and look for the one, $\ket{SD_{\max}}$, that maximizes $C^{\blam, SSR}_{proj/block}$. It is convenient to think of $C^{\blam, SSR}_{proj/block}(\ket{SD}\bra{SD})$ as a function of the data of the two occupied orthonormal orbitals that constitute $\ket{SD}$; and to exploit the invariance of $C^{\blam, SSR}_{proj/block}$ under local orbital rotations. These techniques yield relatively simple expressions for $C^{\blam, SSR}_{proj}(\ket{SD}\bra{SD})$ and $C^{\blam, SSR}_{block}(\ket{SD}\bra{SD})$, but we still had to perform the maximization numerically, as in the last paragraph. Once again, the numerics showed that the correlation-maximizing SD is the same whether or not the local-number SSR is considered. Namely, the (non-unique) optimally rotated SD follows the construction of Result~\ref{prop:Emax_SD_bipartition}, and is given by
\begin{equation} \label{eq:SD_SSR_2_2}
	\ket{SD_{\max}} = \frac{f^\dagger_1 + f^\dagger_3 }{\sqrt{2}} \frac{f^\dagger_2 + f^\dagger_4 }{\sqrt{2}} \ket{0} \; , 
\end{equation}
where $f_{1/2}^\dagger$ ($f_{3/4}^\dagger$) are creation operators relative to any pair of orthonormal orbitals in the first (second) orbital subspace. This is the simultaneous maximizer for both $C^{\blam, SSR}_{proj}$ and $C^{\blam, SSR}_{block}$ (as well as of $C^{\blam}$). A straightforward computation then yields the following: for any SD $\ket{SD}$ of two fermions in four modes,
\begin{align} \nonumber
	C^{\blam,SSR}_{proj,max}(\ket{SD}\bra{SD}) &\stackrel{\text{num.}}{=} C^{\blam,SSR}_{proj}(\ket{SD_{\max}}\bra{SD_{\max}}) \\ 
	&= \frac{5}{2} \log 2 \\ \nonumber
	C^{\blam,SSR}_{block,max}(\ket{SD}\bra{SD}) &\stackrel{\text{num.}}{=} C^{\blam,SSR}_{block}(\ket{SD_{\max}}\bra{SD_{\max}}) \\
	&= \log 2 \; .
\end{align}

\subsection{Large-$N,d$ asymptotics for SD states} \label{subsec:SSR_asymptotics}

Generalizing and adapting the analysis of the previous paragraph to SDs of $N$ fermions in $d = N + N$ orbitals leads to the following result:

\begin{prop} \label{prop:asymptotics_SSR}
	Let $\blam_N = (N, N)$. Let $f(N)$/$f_{proj}(N)$/$f_{block}(N)$ be the value of $C^{\blam_N}_{\max}$/$C^{\blam_N, SSR}_{proj,max}$/$C^{\blam_N, SSR}_{block,max}$ evaluated at any SD state of $N$ fermions in $2N$ orbitals. Then
	\begin{equation} \label{eq:target_limit}
		\lim_{N \rightarrow \infty} \frac{f_{proj}(N)}{f(N)} = \lim_{N \rightarrow \infty} \frac{f_{block}(N)}{f(N)} = 1 \; .
	\end{equation}
\end{prop}
\begin{proof}
Given creation operators $f^\dagger_1, \ldots , f^\dagger_N$ ($f^\dagger_{N+1}, \ldots , f^\dagger_{2N}$) associated to the first (second) orbital subspace, we generalize the SD of Eq.~\eqref{eq:SD_SSR_2_2} as follows:
\begin{equation}
	\ket{SD_N} = \frac{f^\dagger_1 + f^\dagger_{N+1}}{\sqrt{2}} \ldots \frac{f^\dagger_N + f^\dagger_{2N} }{\sqrt{2}} \ket{0} \; .
\end{equation}
According to Result~\ref{prop:Emax_SD_bipartition}, this SD maximizes $C^{\blam_N}$ over all $N$-fermion SDs. Equivalently, if $\ket{SD}$ is an arbitrary $N$-fermion SD,
\begin{align} \nonumber
	f(N) &= C^{\blam_N}_{\max} (\ket{SD} \bra{SD}) \\
	&= C^{\blam_N} (\ket{SD_N} \bra{SD_N}) = 2N \log 2 = O(N) \; .
\end{align}
Also, Eq.~\eqref{eq:SSR_max_ineq} and the definition of $C^{\blam_N, SSR}_{block, max}$ imply
\begin{equation}
	f(N) \geq f_{proj}(N) \geq f_{block}(N) \geq C^{\blam_N, SSR}_{block}(\ket{SD_N} \bra{SD_N}) \; .
\end{equation}
We now show that the right-hand side equals $f(N) + O(\log N)$ and is therefore asymptotic to $f(N)$, thus proving Eq.~\eqref{eq:target_limit}. 

Let $\Pi_{m}$ be the projector onto states with $m$ particles in the first orbital subspace and $N - m$ in the second subspace. For fixed $N$, let $p_m = \braket{SD_N | \Pi_m |SD_N}$, and $\ket{\Phi_m} = \Pi_{m} \ket{SD_N} / \sqrt{p_m}$. Then
\begin{align} \nonumber
	C^{\blam_N, SSR}_{block}(\ket{SD_N} \bra{SD_N}) &= \sum_{m = 0}^{N} p_m C^{\blam_N}(\ket{\Phi_m}\bra{\Phi_m}) \\
	&=  2 \sum_{m = 0}^{N} p_m  S(\rho_{m, 1})  \;,
\end{align}
where $\rho_{m,1}$ is the orbital-reduced state of $\ket{\Phi_m}$ relative to the first subset of orbitals. A straightforward computation yields
\begin{align}
	p_m &= \binom{N}{m} 2^{-N} \\
	S(\rho_{m, 1})  &= \log \binom{N}{m} \; ,
\end{align}
hence
\begin{equation}
	C^{\blam_N, SSR}_{block}(\ket{SD_N} \bra{SD_N}) = 2^{1-N} \sum_{m = 0}^{N} \binom{N}{m} \log \binom{N}{m} \; .
\end{equation}
A standard computation of the sum shows that the above quantity is indeed $2N \log 2 + O(\log N) = f(N) + O(\log N)$.

\hideorblue{We start with $C^{\blam_N, SSR}_{proj}(\ket{SD_N} \bra{SD_N})$. Let $\rho_N = \ket{SD_N} \bra{SD_N}$ and, as usual, denote with $\rho_N^{SSR}$ its projected version relative to the partition $\blam_N$. It is given by
\begin{equation}
	\rho_N^{SSR} = \sum_{m = 0}^{N} \Pi_1^{m} \rho \Pi_1^{m} \; ,
\end{equation}
where $\Pi_1^{m}$ projects onto states with $m$ particles in the first orbital subspace and $N - m$ in the second subspace. Since $\rho_N$ and thus, $\forall \, m$, $\Pi_1^{m} \rho \Pi_1^{m}$ each have rank one, $\rho_N^{SSR} $ has rank at most $N+1$. Therefore, its entropy is bounded as
\begin{equation}
	S(\rho_N^{SSR} ) \leq \log (N+1 ) \; .
\end{equation}
Then, using Cor.~\ref{cor:I_proj}, we get
\begin{align} \nonumber
	C^{\blam_N, SSR}_{proj}(\rho_N) &=  C^{\blam_N}(\rho_N^{SSR}) \\ \nonumber
	& = S(\rho_{N, 1}) + S(\rho_{N, 2}) - S(\rho_N^{SSR}) \\ 
	& \geq S(\rho_{N, 1}) + S(\rho_{N, 2}) - \log (N+1 ) \\
	& = \underbrace{C^{\blam_N}(\rho_N)}_{f(N)=2N \log 2} - \log (N+1 ) \; .
\end{align}
We observe that the correction $- \log (N+1 )$ due to the SSR is negligible if compared to the leading term. Using the bounds of Prop.~\ref{prop:SSR_bounds}, we get
\begin{equation}
	f(N) \geq f_{proj}(N)
\end{equation}
}

\end{proof}
Prop.~\ref{prop:asymptotics_SSR} suggests that the effect of the local-number SSR on the maximal orbital correlation might be negligible in the large-$d$ limit, under appropriate conditions.

A similar conclusion, limited to the quantity $C^{\blam, SSR}_{block}$, is drawn in Ref.~\cite{wisemanvaccaro}. There, multiple copies of a given pure state are considered, and it is argued that `for asymptotically many copies of the state' $C^{\blam, SSR}_{block}$ and $C^{\blam}$ coincide. Similarly, in Ref.~\cite{botero_normal_form_ferm_Gaussian_states} on orbital entanglement in pure fermionic Gaussian states it is stated that the SSR-agnostic correlation measure $C^{\blam}$ `should give a reasonable indication of the usable entanglement'. 

\end{document}